\documentclass[abbrvbib,11pt]{article}
\usepackage{jmlr2eMODIFIED}

\usepackage{amsfonts,amsmath,amsxtra,amssymb,mathrsfs,dsfont,booktabs}

\usepackage{todonotes}
\usepackage{enumitem}

\usepackage{rotfloat}
\usepackage{colortbl}

\usepackage{mathtools}

\newcommand{\inv}{^{-1}}
\newcommand{\id}{\mathds 1}
 
\DeclareMathOperator*{\argmin}{\arg\!\min\;}

\newcommand{\N}{\mathbb{N}}

\newcommand{\W}{\mathbb{W}}
\newcommand{\X}{\mathbb{X}}

\renewcommand{\aa}{\mathbf{a}}
\newcommand{\kk}{\mathbf{k}}

\newcommand{\mm}{\mathbf{m}}
\newcommand{\uu}{\mathbf{u}}
\newcommand{\vv}{\mathbf{v}}

\newcommand{\xx}{\mathbf{x}}
\newcommand{\yy}{\mathbf{y}}
\newcommand{\zz}{\mathbf{z}}

\renewcommand{\AA}{\mathcal A}
\newcommand{\BB}{\mathcal B}
\newcommand{\CC}{\mathcal C}
\newcommand{\DD}{\mathcal D}

\newcommand{\II}{\mathcal I}

\newcommand{\Scal}{\mathcal{S}}
\newcommand{\TT}{\mathcal T}
\newcommand{\VV}{\mathcal V}

\newcommand{\pers}{\mathsf{pers}}
\newcommand{\birth}{\mathsf{birth}}

\newcommand{\M}{\mathsf{M}}
\newcommand{\Mult}{\mathsf{Mult}}
\newcommand{\supp}{\mathsf{supp}}

\ShortHeadings{Approximating Maps on Persistence Diagrams}{Perea, Munch, and Khasawneh}
\firstpageno{1}

\begin{document}

\title{Approximating Continuous Functions on Persistence Diagrams Using Template Functions}

\author{\name Jose A.~Perea\thanks{Corresponding author} \email j.pereabenitez@northeastern.edu \\
        \addr Department of Mathematics; and\\
       Khoury College of Computer Sciences\\
       Northeastern University\\
       Boston, MA 02115, USA
       \AND
        \name Elizabeth Munch \email muncheli@msu.edu\\
       \addr Department of Computational Mathematics, Science, and Engineering; and\\
       Department of Mathematics \\
       Michigan State University\\
       East Lansing, MI 48824, USA
       \AND
       \name Firas A.~Khasawneh \email khasawn3@egr.msu.edu \\
       \addr Department of Mechanical Engineering\\
       Michigan State University\\
       East Lansing, MI 48824, USA
       }

\editor{}

\maketitle

\begin{abstract}%   <- trailing '%' for backward compatibility of .sty file
The persistence diagram is an increasingly useful tool from Topological Data Analysis,
but its use alongside typical machine learning techniques requires mathematical finesse.
The most success to date has come from methods that map persistence diagrams into vector spaces,
in a way which maximizes the structure preserved. This process is commonly referred to
as featurization. In this paper, we describe a mathematical framework for featurization
called \emph{template functions}, and we show that it addresses
the problem of approximating continuous functions on compact subsets of the space of persistence diagrams.
Specifically, we begin by characterizing relative compactness with respect to the bottleneck distance,
and then provide explicit theoretical methods for constructing compact-open dense subsets of continuous functions
on persistence diagrams.
These dense subsets---obtained via template functions---are
leveraged for supervised learning tasks with persistence diagrams.
Specifically, we test the method for classification and regression algorithms on several examples including shape data and dynamical  systems.
\end{abstract}


\section{Introduction}
Many machine learning tasks can be reduced to the following problem:
Approximate a continuous function defined on a topological space, the ``ground truth,'' given the function values (or approximations thereof) on some subset of the points.
This task has been  studied extensively for data  in Euclidean space, but more work is necessary to extend these ideas to arbitrary topological spaces.
In this paper, we focus on the task of learning continuous functions on the space $\DD$ of persistence diagrams  endowed with the bottleneck distance $d_B$, as input to  supervised learning (e.g., regression and classification) on $\DD$.

Persistence diagrams are mathematical objects arising in the field of Topological Data Analysis (TDA).
They are signatures giving insight into the underlying structure of  data sets,
and  typically arise in the following  pipeline.
Given a sequence of simplicial complexes $K_1 \subseteq K_2 \subseteq \cdots \subseteq K_n$
obtained, for example, by connecting the points in a data set at increasing proximity scales (e.g., as in the Vietoris-Rips filtration),
then their homology with coefficients in a field and the maps induced by the inclusions $K_i \hookrightarrow K_{i+1}$
yield a sequence $H_p(K_1) \to H_p(K_2) \to \cdots \to H_p(K_n)$  of vector spaces and linear transformations.
One can then interrogate this sequence to determine when homological features appear (are born) and disappear (die),
and encode each such feature as a point (birth, death) in a so-called persistence diagram.
It is possible to weaken the implicit finiteness assumptions of the input data and instead start with a family
 $\{K_a \mid a  \in \mathbb{R}\}$ of spaces with $K_a \subseteq K_b$ for $a \leq b$,
 giving rise to a so-called persistence module $\{ H_p(K_a) \mid a \in \mathbb{R}\}$ with linear maps $\phi_a^b: H_p(K_a) \to H_p(K_b)$ such that
 $\phi_a^a$ is the identity map of $H_p(K_a)$ and $\phi_b^c\phi_a^b = \phi_a^c$ for $a \leq b \leq c$.
 To simplify definitions, we will assume that our input diagrams are defined over positive indices $a \geq 0$.
There has been extensive study on the various niceness restrictions that can be placed on a persistence module \citep{Bubenik2018b}.
With enough assumptions, a persistence module can be represented up to isomorphism as a persistence diagram, which is simply a collection of points with multiplicity in the ``wedge'' $\W =  \{(x,y)\in \mathbb{R}^2 \mid  0 \leq x < y\}$.
Persistence diagrams are particularly useful for data analysis due to the availability of metrics, and their stability \citep{Lesnick2015}.

The downside of all this mathematical structure is that the geometry of
the space of persistence diagrams
 $(\DD,d_B)$ is not directly amenable to the application of existing machine learning methodologies.
Thus, methods for utilizing persistence diagrams in statistics and machine learning contexts have taken two basic forms.
The first are attempts at working with the persistence diagrams directly; however, the issues with the geometry (particularly the lack of unique means) mean that the work in this direction is rather limited.
More recently, a great deal of success has been found in \textit{featurization}; that is, transforming each persistence diagram into a point in a vector space,  in a way that preserves as much of the structure as possible.
The work in this paper is part of the latter category, and provides a new method for featurization which sits on a solid mathematical foundation with respect to the structure of $(\mathcal{D}, d_B)$.

Mathematically, we are working with the following framework.
Suppose  one has a set  $\mathcal{S} \subset \DD $, typically compact,  and a continuous function $F: \mathcal{S} \longrightarrow \mathbb{R}$ encoding the ground truth of the phenomenon under study.
\textbf{Our goal} is to devise provably-correct and computationally feasible approaches to approximating $F$, given a finite sample $D_1,\ldots, D_n \in \mathcal{S}$ and their values $F(D_1),\ldots, F(D_n) \in \mathbb{R}$.
This encompasses, for instance,  supervised learning tasks such as regression and classification.
The problem at hand is thus (1) to characterize  compactness in $(\DD,d_B)$, for these are the sets
where sequential approximations are guaranteed to converge,
(2) to construct dense subsets of the space of continuous functions from $\DD$ to $\mathbb{R}$,
for these will be the search space for approximation (i.e., supervised learning) tasks,
and (3) to devise algorithms using said families to approximate real valued functions on compact subsets of $\DD$.

\subsection{Our contribution }

The first contribution of this paper is a characterization of (relative) compactness for subsets of
$\DD$ with respect to $d_B$ (see Figure \ref{fig:CompactExamples} and  Theorem \ref{thm:Compactness}).
These results can be viewed in parallel to the characterization  in \cite{Mileyko2011} of relative compactness with respect to the Wasserstein distance, though our results are shown to  capture  different phenomena (see Sec. \ref{sec:WassVSBottle}).
Our characterization also comes with some unexpected consequences for the  topology of
$(\DD, d_B)$: (1) Every compact subset of $\DD$ has empty interior (hence $\DD$ is not locally compact);
(2) $\DD$ cannot be written as a countable union of compact subsets; and if  $C(\DD, \mathbb{R})$ denotes  the set of continuous functions from $\DD$ to $\mathbb{R}$,
then (3) the compact-open topology on $C(\DD,\mathbb{R})$---which captures approximations on compact subsets of $\DD$---is not metrizable.
The main consequence of this last point is the impossibility of purely metric-based objective functions
for approximations (e.g., supervised learning) in $C(\DD,\mathbb{R})$ with respect to the topology of convergence on compact sets.

To circumvent this, we turn our attention to the problem of finding  compact-open dense subsets of $C(\DD, \mathbb{R})$.
Ideally, the elements of these sets should be succinctly represented (e.g., with a few parameters) and efficiently searched (e.g., via appropriate optimization routines),
in order to devise general computational schemes.
Our second contribution is a methodology for constructing infinitely many examples of said families.
The strategy   goes as follows:
First, we continuously embed $\DD$ in an appropriate topological vector space $V$---consistent with the monoidal structure of $\DD$ given by
disjoint union $\sqcup$ of multisets---and then restrict the
continuous $\mathbb{R}$-linear maps on $V$ to yield elements in $C(\DD, \mathbb{R})$.
Specifically, let $C_c(\W)$ denote the set of compactly supported continuous functions from
$\W $ to $\mathbb{R}$, endowed with the \emph{strict inductive limit topology} (see Section \ref{sec:LinInfDiag}).
Let
$C_c(\W)'$ be its topological dual, endowed with the corresponding \emph{weak-* topology}, and let
\[
\nu : \mathcal{D} \longrightarrow C_c(\W)'
\]
be the function
which assigns to each  $D\in \mathcal{D}$ the Radon measure on $\W$ consisting
of a Dirac delta mass at each $\xx \in D$.
We show in Theorem \ref{thm:Linearization} that $\nu$
is continuous, injective and satisfies $\nu(D \sqcup D') = \nu(D) + \nu(D')$
for all $D,D' \in \mathcal{D}$,
providing the aforementioned linear embedding.
Hence, each bounded linear operator $T: C_c(\W)' \longrightarrow \mathbb{R}$
yields a map $T\circ \nu \in C(\mathcal{D},  \mathbb{R}) $---a feature, in machine learning parlance---which respects the monoidal structure $\sqcup$ of $\mathcal{D}$.

We show (see Theorem \ref{thm:V_iso_V2dual} and discussion thereafter) that each such $T$
 uniquely determines an  $f\in C_c(\W)$, and viceversa, so that $T\circ \nu (D)$ is exactly the same as
integrating $f$ against the Radon measure $\nu(D)$.
We then show how to construct
compact-open dense subsets of $C(\DD,\mathbb{R})$
by taking countably many dilations and translations
of any nonzero  $f\in C_c(\W)$ (Theorems \ref{thm:approximation} and \ref{thm:templates}).
This is why we refer to the elements of $C_c(\W)$ as \emph{template functions}.

As the final contribution of this paper, we provide two explicit families of template functions---called respectively tent functions and interpolating polynomials (see Section \ref{sec:TemplateFunctionExamples})---so that the algebras they generate in $C(\DD, \mathbb{R})$ are
compact-open dense.
We then provide algorithms to perform regularized regression and classification using template functions
(in Section \ref{sec:RidgeRegression}),
and finally, we compare tent functions and interpolating polynomials in several tasks including
shape classification and inference in dynamical systems (Section \ref{sec:Experiments}).

\subsection{Related work}
Existing methods for applying statistics and machine learning methods to persistence diagrams can be loosely divided into two categories.
The first attempts to work in the space of persistence diagrams directly.
This can be done by studying the Fr\'echet mean for collections of diagrams \citep{Mileyko2011,Turner2014,Munch2015}, cofindence sets for persistence diagrams \citep{Fasy2014}, or simply passing the  Wasserstein or bottleneck distance matrix to metric learning methods for classification tasks (e.g.~\cite{Li2014}).
The main issue with this viewpoint is that the geometry of the space of persistence diagrams is ill-behaved \citep{wagner2019nonembeddability}, so directly working with these objects is often not advisable.

The second collection of methods maps the space of persistence diagrams into another, more well-behaved space where available mathematical machinery can be readily applied.
Our work on approximating continuous functions of persistence diagrams  with templates fits into this category.
These approaches take inspiration from different viewpoints including algebraic geometry \citep{Adcock2012,CarlssonVerovsek2016,Kalisnik2018,Fabio2015}, functional representations \cite{Bubenik2015,Berry2018,Chazal2014b,Padellini2017},
image encodings \cite{Adams2017,Chen2015,Donatini1998,Ferri1998,Rouse2015},
path representations \cite{Chevyrev2018},
kernel methods \cite{Reininghaus2015,Kwitt2015,Kusano2016,Kusano2017,Kusano2018,Carriere2017c,Corbet2018a,Carriere2018,Zhao2019,Anirudh2016,Le2018,Zhu2016,Carriere2019},
and other more ad hoc methodologies \cite{Bendich2016a,Singh2014,Chung2009,Pachauri2011,Carriere2015,Zielinski2018,Zielinski2018a}.
Our work is most closely related to the persistence images of \cite{Adams2017}, where each point in a persistence diagram contributes to a Gaussian bump, and the sum of these Gaussians provides a function on the upper half plane.
Our work is, in some sense, dual to this idea, where we start with bump functions and evaluate this function at the points of the persistence diagram

\subsection{Outline}

We go over the background needed for understanding persistence diagrams in Sec.~\ref{sec:Basics}.
In Sec.~\ref{sec:Compact}, we give a full characterization of compact sets in persistence diagram space with the bottleneck distance (Sec.~\ref{thm:Compactness}); the reader more interested in the featurization method than in its mathematical justification may safely skip this section.
We provide the mathematical justification for the template functions in Sec.~\ref{sec:LinearizingD}, and fit this into a function approximation scheme in Sec.~\ref{sec:Approximating}.
In Sec.~\ref{sec:TemplateFunctionExamples} we give two options for template functions: tent functions and (Chebyshev) interpolating polynomials.
In Sec.~\ref{sec:RidgeRegression} we fit these into a regression framework.
We give results of our experiments in Sec.~\ref{sec:Experiments} and discuss implications and future directions in Sec.~\ref{sec:Discussion}.

%--------------------
%--------------------
\section{Basics}
\label{sec:Basics}
%--------------------
%--------------------

Traditionally, persistence diagrams arise in the course of the following procedure.
Given  a  function  $f:\X \to \mathbb{R}$  on a topological space $\X$, denote the sublevel set by $\X_a = f\inv(-\infty,a]$.
For example, given a point cloud $X \subseteq \mathbb{R}^d$, one can define  $f:\mathbb{R}^d \to \mathbb{R}$ as $f(y)= \inf\limits_{x \in X}\|x-y\|$.
Such a function induces a filtration  $\X_a \subseteq \X_b$,  $a \leq b$, on $\X$ and
applying $k$-dimensional homology  yields the persistence module $(H_k(\X_a),\phi_a^b)$.
Namely, the collection of vector spaces\footnote{Homology is computed with coefficients in a field $\kk$.} $H_k(\X_a)$ with induced maps $\phi_a^b: H_k(\X_a) \to H_k(\X_b)$ for all $a \leq b$.
In full generality, persistence modules can simply be viewed as a collection of vector spaces and linear maps $\VV = (V_a, \phi_a^b)$ where $\phi_a^b:V_a \to V_b$, $\phi_a^a = \id_{V_a}$, and $\phi_b^c\phi_a^b = \phi_a^c$.

Under  appropriate tameness conditions, a persistence module can be decomposed uniquely.
The pieces of the decomposition are called \textit{interval modules}; these are persistence modules $\II_{U} = (I_a,i_a^b)$
where $U \subset \mathbb{R}$ is an (open, closed or half-open)  interval with end-points $-\infty \leq r \leq s \leq \infty $,
 $I_a = \kk$ if $a \in U$ and 0 otherwise.
The maps $i_a^b$ are identities whenever possible.
A persistence module $\VV = (V_a,\phi_a^b)$ is called pointwise-finite if $V_a$ is finite dimensional for every $a$.
Every pointwise-finite persistence module decomposes uniquely as a direct sum of interval modules, $\VV = \bigoplus_{U \in \AA}\II_{U}$
\cite{crawley2015decomposition}.
This decomposition is often visualized as a persistence diagram as seen in the center panel of Fig.~\ref{fig:PersDgmExample}.
The diagram consists of a point at $(r,s) \in (\mathbb{R}\cup \{\pm\infty\})^2$ for the endpoints of each $U \in \AA$.
We remark that diagonal points $(r,r)$ will be discarded  as they encode non-persistent features (these are called ephemeral modules),
and are not detected by our template functions $f \in C_c(\W)$ for  in this case $\mathsf{supp}(f) \subset \W$ does not intersect  the diagonal.
For the same reason, points $(r,s)$ with $\{r,s\} \cap \{\pm \infty\} \neq \emptyset$ (i.e., with infinite persistence) will also be discarded.

For simplicity of notation and definitions, we will be working with positive diagrams:  those with $r \geq 0$ for each point $(r,s)$ in the persistence diagram.
Note that we can always convert a finite set of persistence diagrams to have this assumption by shifting the index by adding the minimum birth value.
For infinite sets, we can reindex using any positive, monotone increasing function of the original index parameter to achieve the same results.

Certain visualizations will be on birth-lifetime coordinates, consisting of a point at $(r,s-r) \in \mathbb{R}^2$ for each interval $U \in \AA$ with endpoints (either closed or open) $r\leq s$.
See the right of Fig.~\ref{fig:PersDgmExample} for an example.
We note that the methods here developed also apply to zigzag persistence modules.
These  consist of vector spaces $V_a$ and
 linear maps between them that can point in either direction: i.e.,  $V_a \rightarrow V_b$ or $V_a \leftarrow V_b$ for $a\leq b$.
It turns out that Zigzag modules can also be decomposed as sums of (zigzag) intervals modules \cite{carlsson2010zigzag}, which in turn
can be represented as   persistence diagrams.

\begin{figure}[tb]
	\centering
	\includegraphics[width = \textwidth]{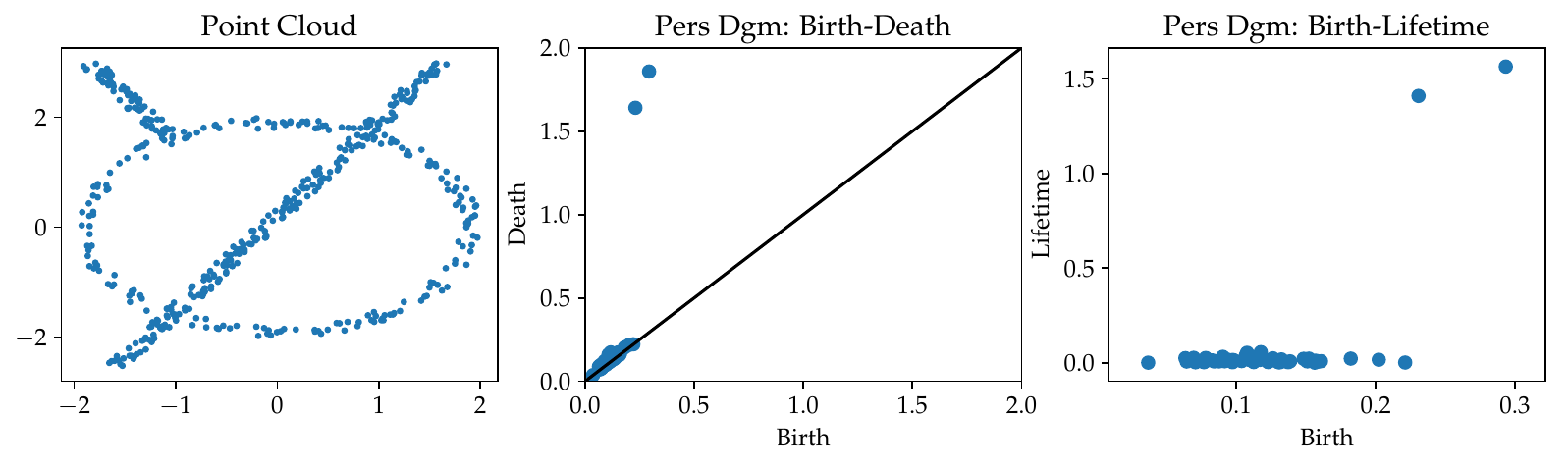}
	\caption{An example point cloud is shown at left, with its persistence diagram shown in the middle.  At right, we show the conversion of the persistence diagram into the birth-lifetime plane which is used throughout this paper. }
	\label{fig:PersDgmExample}
\end{figure}

\subsection{The space of persistence diagrams}
A persistence diagram $D$, then, can be thought of as a collection of points
\[
S \subset   \{(x,y) \in \mathbb{R}^2 \mid 0 \leq x < y\}
\]
with a notion of multiplicity, which we write as a function
$$\mu:S \to \N = \{1,2,\ldots\}.$$
We will often write $D = (S,\mu)$.
In order to make statements about the structure of the space of persistence diagrams we will need  a few notions.

\begin{definition}
	\label{defn:MultiplicityOfDgm}
Given $D = (S,\mu)$ and $U \subset \mathbb{R}^2$,
 the \textbf{multiplicity of $D$ in $U$} is
 \begin{equation*}
 \Mult(D,U) =
 \begin{cases}
 % \displaystyle{
 \sum\limits_{\xx \in S \cap U }
\mu(\xx)
% }
& \text{ if this is finite, or}\\
\infty & \text{ else.}
 \end{cases}
\end{equation*}
\end{definition}
Other common notions  we will use repeatedly are as follows.
The diagonal is denoted $\Delta := \{(x,x) \in \mathbb{R}^2 \mid x \geq 0 \}$, and
the wedge is
\[
\W := \left\{(x,y) \in \mathbb{R}^2 \mid 0 \leq x < y\right\}.
\]
Note that the boundary of $\W$ is $\Delta$, but it is not included in $\W$.
The persistence of a point $\xx = (x,y) \in \W$ is $\pers(\xx) = y-x$, and
the portion of $\W$ where persistence is greater than $\varepsilon$ is
\[	
\W^\varepsilon := \{\xx \in \W \mid \pers(\xx) > \varepsilon\}.
\]
Note that the ($\varepsilon$-diagonal) lower boundary is not included;
in order to do so we write
\[
\overline{\W^\varepsilon} =  \{\xx \in \W \mid \pers(\xx) \geq \varepsilon \}.
\]
If we want to   work with the portion of  points in $D = (S,\mu)$ in a  region $U \subseteq \mathbb{R}^2$, we write
\[
D \cap U := (S \cap U, \mu\vert_{S\cap U}).
\]
We further abuse notation by writing
	$D \subset U$
if $S \subset U$.
If $S = \emptyset$, we follow the set-theoretic convention
$\mu = \emptyset$ and denote by $\varnothing= (\emptyset, \emptyset)$ the resulting (empty) persistence diagram.
For the sake of figures, we sometimes plot persistence diagrams in the birth-lifetime plane.
That is, we plot $\xx = (x,y) $ at the point $(x, y-x) = (x, \pers(\xx))$.
In this representation, $\Delta$ gets mapped to the $x$-axis.  It should be noted that this transformation is different from the rotation used by \cite{Bubenik2015} and \cite{Adams2017}.
With these notions we define the space of persistence diagrams as follows.

\begin{definition}
The \textbf{space of persistence diagrams}, denoted $\DD$,
 is the collection
of pairs $D=(S,\mu)$ where:
\begin{enumerate}
  \item $S\subset \W$ is   the \textbf{underlying set of $D$},
  and $\mu: S \rightarrow \mathbb{N} = \{1,2,\ldots \}$ encodes the  \textbf{multiplicity}
  $\mu(\xx)\in \N$   of each $\xx \in S$. \\
  \item
 $\Mult(D,\W^\varepsilon)< \infty$ for any $\epsilon > 0$.
\end{enumerate}
The \textbf{space of finite persistence diagrams} is $
\DD_0 := \{(S,\mu) \in \DD \mid S \mbox{ is finite}\}
$.
\end{definition}
Finite persistence diagrams  were the first  to appear in the literature,
 and many current papers implicitly assume finiteness.
 We do not do so here, as  diagrams with infinitely many points can be used to encode fractal behavior---e.g., in attractors from dynamical systems.
Since we will be interested in studying subsets of $\DD$, we will extend Definition~\ref{defn:MultiplicityOfDgm} as follows.
\begin{definition}
Given $\mathcal{S}\subset\DD$ and $U \subset \mathbb{R}^2$,
then the total multiplicity of $\mathcal{S}$ in $U$ is
 \begin{equation*}
 \Mult(\mathcal{S},U) =
 \begin{cases}
 \displaystyle{
 \sum_{D \in \mathcal{S}}
\Mult(D,U)
}
& \text{ if this is finite, or}\\
\infty & \text{ else.}
 \end{cases}
\end{equation*}
In particular, for each $D\in \mathcal{D}$ we write
$\Mult(D,U)$ instead of $\Mult(\{D\}, U)$.
\end{definition}

\subsection{The Bottleneck distance}

The space of persistence diagrams can be endowed with a distance, which we now define.
To each persistence diagram $D = (S,\mu)$ one can associate a set
\begin{equation}
	\label{eq:OtherPersNotation}
S_\mu :=
\big\{
(\xx, k ) \mid
\xx \in S \mbox{ and } 1\leq k \leq \mu (\xx)
\big\}
\end{equation}
obtained by replicating the elements of $S$ and decorating them with integer labels
according to their multiplicity.
A \textbf{partial matching }between two persistence
diagrams $(S,\mu), (T,\alpha)$
is a
bijection  (with notation as in eq (\ref{eq:OtherPersNotation}))
\[
\begin{array}{rccc}
\M : & A & \longrightarrow & B\\
& \mathbin{\rotatebox[origin=c]{270}{$\subseteq$}} && \mathbin{\rotatebox[origin=c]{270}{$\subseteq$}}\\
& S_\mu && T_\alpha
\end{array}
\]
between a subset $A$ of $S_\mu$ and a subset $B$ of $T_\alpha$.
If $(\yy,n) = \M(\xx, k)$ we say that $(\xx,k)$ is matched
with $(\yy,n)$ and, conversely, that $(\yy,n)$ is matched with $(\xx,k)$.
If $(\zz,m)$ is in either $S_\mu \setminus A $ or $
T_\alpha \setminus B$, then we call it unmatched.
Given $\delta >0$, a partial matching $\M$
between  $(S,\mu)$ and $(T,\alpha)$
is a \textbf{$\delta$-matching} if two things happen:
\begin{enumerate}
  \item If $(\xx,k) \in S_\mu$ and $ (\yy,n) = \M(\xx, k)$ are matched, then  $\|\xx - \yy\|_\infty < \delta$,
   where  $\|(x_1,x_2)\|_\infty = \max\{\lvert x_1\rvert,\lvert x_2\rvert\}$ denotes the $L^\infty$ norm on $\mathbb{R}^2$.\\
  \item If $(\zz,m) \in S_\mu \cup T_\alpha$ is unmatched,  then $\pers(\zz) < 2\delta$.
\end{enumerate}

\begin{definition}
	\label{defn:bottleneck}
The \textbf{bottleneck distance}, $d_B : \DD \times \DD \longrightarrow [0,\infty)$,
is given by
\[
d_B(D,D') :=
\inf
\big\{
\delta > 0 \mid \mbox{there is a $\delta$-matching between $D$ and $D'$ }
\big\}
\]
\end{definition}
It has been shown that $d_B$ defines a metric on $\DD$ \cite{Cohen-Steiner2007}, and that $\DD$ is the metric completion of $\DD_0$ \cite{blumberg2014robust}.

For simplicity (though without loss of generality) we assume in this paper that all persistence points are finite; i.e.~for each point $(b,d)$ in the diagram, the lifetime $d-b$ is finite.
That is, we discount homological features with infinite lifetime which occur when $d = \infty$.
The assumptions on $\DD$ make it so that the bottleneck distance is still finite between diagrams in $\DD$.
In particular, this comes from the triangular inequality and the fact that
\[
d_B(D,\varnothing) = \tfrac{1}{2}\max \{\pers(\xx) \mid \xx \in S\} \;\;\;\; , \;\;\; \mbox{ for all } D= (S,\mu)\in \DD.
\]

\section{Compactness in $\mathcal{D}$}
\label{sec:Compact}
Our first contribution is Thm.~\ref{thm:Compactness}, which gives a criterion for characterizing  (relatively) compact sets in $(\DD, d_B)$.
This work can be viewed in parallel to Theorem 21 of \cite{Mileyko2011}, which does the same using the related Wasserstein distance
$d_{W_p}$ for persistence diagrams. For other structural properties of families of persistence modules see  \cite{Bubenik2018b}.

\begin{figure}%[!htb]
	\centering
	\includegraphics[width = .5\textwidth]{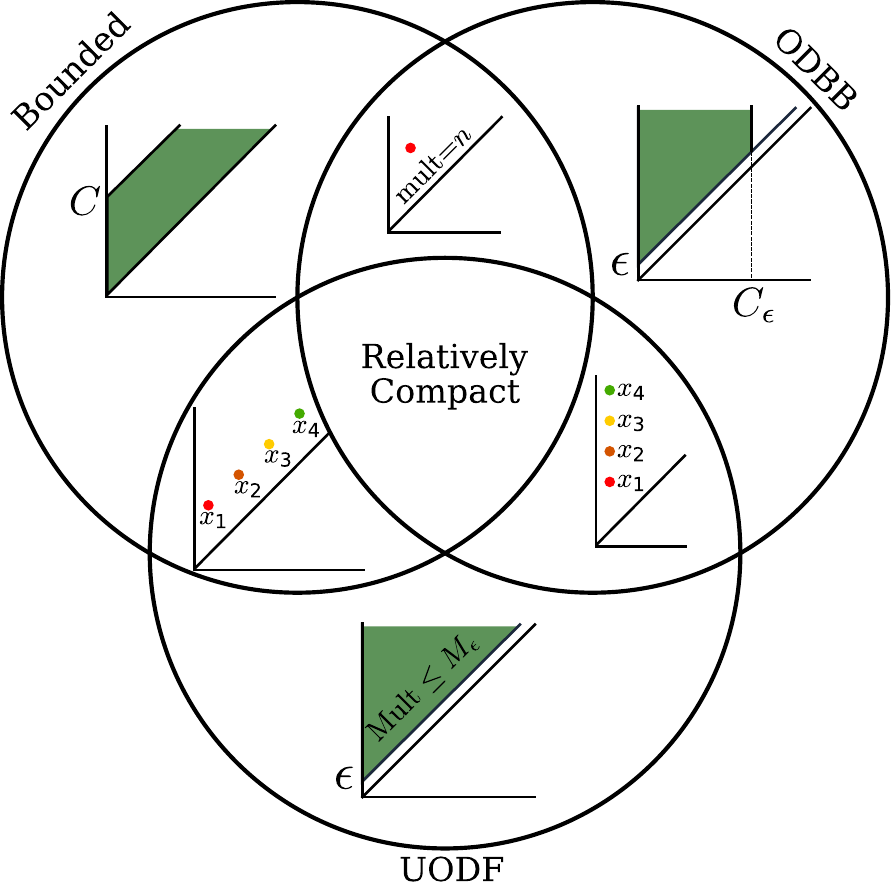}
	\caption{The three criterion for compact sets with examples given on intersections.}
	\label{fig:CompactExamples}
\end{figure}

\begin{definition}
A subspace of a topological space is \textbf{relatively compact} if its closure is compact.
\end{definition}
In what follows we will provide  a criterion to check whether a subset $\Scal \subseteq \DD$ is relatively compact.
Specifically, we will show that a set is relatively compact if it satisfies the following three properties.

%------------
%------------
\subsection{Bounded}
%------------
%------------
% \begin{definition}[Bounded]
The first property of interest is boundedness.
A subset of a metric space is said to be bounded if it is contained in an open ball of finite radius.
% \end{definition}
Let
\[
B_C(D) := \{D' \in \DD \mid d_B(D,D') < C\}
\]
denote the ball of radius $C > 0$ about the diagram $D$.
% \begin{remark}
In particular, it can be seen from the definition that   $\mathcal{S} \subset \mathcal{D}$ is bounded
if and only if there exists $C > 0 $
so that $\mathcal{S} \subseteq B_C(\varnothing)$.
% \end{remark}
\begin{proposition}\label{prop:RelCpcBounded}
Relatively compact subsets of $(\DD, d_B)$  are bounded.
\end{proposition}

\begin{proof}
Let $\Scal \subseteq \DD$ be relatively compact.
To see that $\Scal$ is in fact bounded,
consider the cover
$\left\{B_1(D) \mid D\in \overline{\Scal}\right\}$
by open balls of radius 1
and let $\{B_1(D_i)\}_{i=1}^N$ be a finite subcover.
If
\[
C > 1 + \max\{d_B(D_j, \varnothing) \mid 1\leq  j \leq N\}
\]
then
it follows that $\Scal \subset \overline{\Scal}\subset B_{C}(\varnothing)$,
as claimed.

\end{proof}

Note that this proof works for a general metric space, but we work in $\DD$ for clarity.

%------------
%------------
\subsection{Off-diagonally birth bounded}
%------------
%------------
The second property of interest controls   persistence diagrams with unbounded birth.

\begin{definition}\label{def:ODBB}%[Off-diagonally birth bounded]
A set $\mathcal{S} \subset \mathcal{D}$ is said to
be \textbf{off-diagonally birth bounded (ODBB)} if for
every $\varepsilon > 0$ there exists a constant $C_\varepsilon \geq 0$
so that if $\xx \in S \cap \overline{\W^\varepsilon}$ (i.e., $\pers(\xx)) \geq \varepsilon$) for $(S,\mu) \in \mathcal{S}$,
 % and $\pers(x,y) \geq  \varepsilon$
 then $\birth(\xx) \leq C_\varepsilon$.

\end{definition}
 See Fig.~\ref{fig:CompactExamples} for a visualization of the notation.

\begin{proposition}\label{prop:RelCpcBirthBounded}
Relatively compact subsets of $(\DD,  d_B)$ are ODBB. %off-diagonally birth bounded.
\end{proposition}
\begin{proof}
By way of contradiction, assume $\mathcal{S}\subseteq \DD$ is relatively compact but not ODBB.
% off-diagonally birth bounded,
Then there exist
(i) $\varepsilon > 0 $;
(ii) a sequence $\{D_n\}_{n\in \N} \subseteq \mathcal{S}$  with  $D_n = (S_n , \mu_n)$;
(iii) a fixed diagram $D = (S,\mu)\in \DD$
such that
\begin{equation*}
\lim\limits_{n \to \infty }D_n = D,
\end{equation*}
which exists because $\overline \Scal \subset \DD$ is compact; and
(iv) a chosen point in each diagram $D_n$, namely $\xx_n \in S_n$,
with
% $y_n \geq  x_n + \varepsilon$ and
$\pers(\xx_n) \geq \varepsilon$ for all $n\in \N$, and
$\lim\limits_{n \to \infty} \birth(\xx_n) = \infty$.

Let $\delta < \frac{\varepsilon}{2}$ and let $N\in \N$ be large enough so that
$d_B(D_n ,D) < \delta$ for all $n\geq N$.
For each $n \geq N$ fix a $\delta$-matching
% $\gamma_n : D_n \longrightarrow D $
\[
\begin{array}{rccc}
\gamma_n : & A_n & \longrightarrow & B_n\\
& \mathbin{\rotatebox[origin=c]{270}{$\subseteq$}} && \mathbin{\rotatebox[origin=c]{270}{$\subseteq$}}\\
& (S_{n})_{\mu_n} && S_\mu
\end{array}
\]
Since $\xx_n \in S_n$  has $\pers(\xx_n) \geq \varepsilon$ and $\delta < \varepsilon/2$, then $(\xx_n,1) \in A_n$.
Let $\yy_n \in S$ be such that $\gamma_n(\xx_n,1) = (\yy_n, k_n)$.
As $\gamma_n$ is a $\delta$-matching, then  $\|\xx_n - \yy_n\|_\infty < \delta$ which means, in particular, that $\pers(\yy_n) \geq \delta$.
Hence $\{\yy_n\}_{n \in \N} \subseteq S$ is an infinite set in $ \W^\delta$, contradicting   $D \in \DD$.

\end{proof}

%------------------
%------------------
\subsection{Uniformly off-diagonally finite}
%------------------
%------------------
The final property of interest controls the multiplicity of points across all diagrams:

\begin{definition}%[Uniformly off-diagonally finite]
\label{def:UODF}
A set $\mathcal{S} \subset \mathcal{D}$ is said to
be \textbf{uniformly off-diagonally finite (UODF)} if for every
$\varepsilon > 0$ there exists $M_\varepsilon \in \N$
so that
\[
\Mult\left( D,\overline{\mathbb{W}^{\varepsilon}} \right)
\leq \;M_\varepsilon
\]
for all $D\in \mathcal{S}$.
\end{definition}
Again, see Fig.~\ref{fig:CompactExamples} for a visualization of the notation.
\begin{proposition}\label{prop:RelCpcFinite}
Relatively compact subsets of $\mathcal{D}$ are uniformly off-diagonally finite.
\end{proposition}
\begin{proof}
By the contrapositive, if $\mathcal{S} \subset \mathcal{D}$ is  not UODF,
then there exist $\varepsilon > 0 $ and a sequence
$\{D_n\}_{n \in \N} \subset \mathcal{S}$
so that
\[
\Mult\left(D_n, \overline{\W^\varepsilon}\right) <
\Mult\left(D_{n+1},\overline{\W^\varepsilon}\right)
\]
for all $n\in \N$.
In particular, using the pigeonhole principle, any partial matching between $D_n$ and $D_{n+1}$ must have at least one point in $D_{n+1}$ unmatched.
As this point has persistence greater than $\varepsilon$, it follows that $d_B(D_n, D_{n+1}) \geq \varepsilon$
for every $n\in \N$
and therefore $\{D_n\}_{n\in \N}$ cannot have
a convergent subsequence.
This shows that $\overline{\mathcal{S}}$ is not compact.

\end{proof}

%--------------
%--------------
\subsection{Helpful counterexamples}
\label{sec:CounterExamples}
%--------------
%--------------
The three conditions bounded, off-diagonally birth bounded and
uniformly off-diagonally finite are independent.
Indeed, here are three examples of sets  $\mathcal{S} = \{(S_n,\mu_n) \mid n \in \N\}$ which satisfy only
two out of the three conditions.
See Fig.~\ref{fig:CompactExamples}.

% \begin{example}
\begin{enumerate}
  \item Bounded and ODBB, but not UODF.
  % Bounded and off-diagonally birth bounded but not uniformly off-diagonally finite.

   \noindent  $S_n = \{(0,1)\}$ with $\mu_{n}(0,1) = n$. \\
  \item Bounded and UODF, but not ODBB.
  % Bounded and uniformly off-diagonally finite but not off-diagonally birth bounded.

  \noindent   $S_n = \{(n, n+1)\}$ with $\mu_{n}(n,n+1) = 1$. \\
  \item UODF and ODBB, but not bounded.
  % Uniformly off-diagonally finite and off-diagonally birth bounded, but not bounded.

  \noindent   $S_n = \{(0,n)\}$ with $\mu_{n}(0,n)= 1$.
\end{enumerate}
% \end{example}

%--------------
%--------------
\subsection{Characterizing Compactness in $(\mathcal{D}, d_B)$}
\label{ssec:Compact}
%--------------
%--------------
With these definitions, we can now state our main compactness theorem.
%-
\begin{theorem}[Characterization of compactness in $(\mathcal{D}, d_B)$]
  \label{thm:Compactness}
A set $\mathcal{S} \subset \mathcal{D}$ is relatively compact in $(\DD, d_B)$
if and only if it is bounded, off-diagonally birth bounded (ODBB) and uniformly off-diagonally finite (UODF).
\end{theorem}
%-
Note that one direction is already provided by Propositions \ref{prop:RelCpcBounded}, \ref{prop:RelCpcBirthBounded} and \ref{prop:RelCpcFinite}, so our main job is to show that a set which satisfies the three conditions is relatively compact.
Before we prove this, however, we will need to build a bit of machinery.
First, notice that if $\mathcal{S} \subset \mathcal{D}$ is bounded ($\Scal \subset B_C(\varnothing)$) and ODBB,
then there exist a collection of finite ``boxes'' $\mathbb{B}_k$ in $\W$
 whose union contain all points in the diagrams.
Specifically, if $\Scal$ is bounded and ODBB, then there is a $C > 0 $ and
$\{C_k\}_{k \in \N} \subset \mathbb{R}_{> 0}$ non-decreasing,  so that  if
\begin{equation}\label{eq:defBox}
\mathbb{B}_k =
\left\{
\xx \in \mathbb{W} \mid  0 \leq  \birth(\xx) \leq C_k \;\;\;
\mbox{ and }\;\;\; \frac{C}{k+1} < \pers(\xx) \leq \frac{C}{k}
\right\}
\end{equation}
then  for all $ D \in \overline{\mathcal{S}}$,
% \[
$D  \subset \bigcup\limits_{k\in \N}  \mathbb{B}_k$.
% \]
While these are parallelograms in the birth-death plane, they become rectangles in the birth-lifetime plane, hence the moniker ``box'', see Fig.~\ref{fig:BoxNotationExample}.
% \end{remark}

\begin{figure}[!htb]
  \centering
    \includegraphics[width = .5\textwidth]{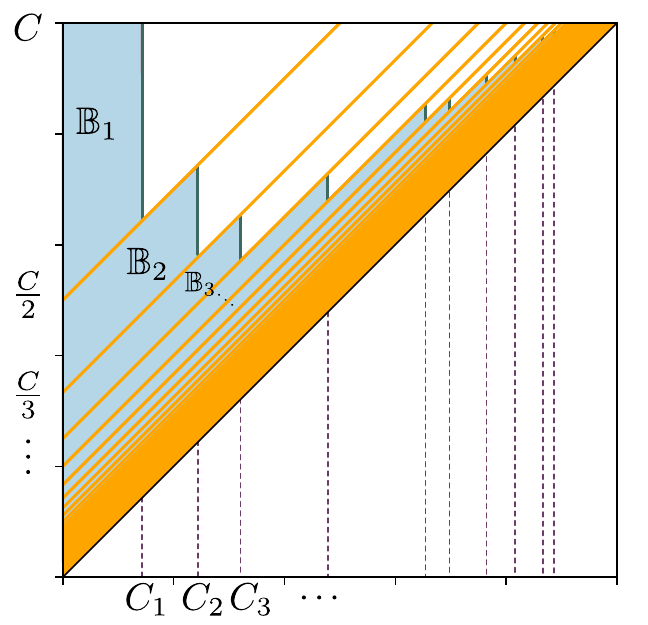}
    \caption{An example to show the notation used in Eq.~\protect\ref{eq:defBox}.}
    \label{fig:BoxNotationExample}
\end{figure}

Second, we can control the multiplicity of the diagrams in these boxes.
% \begin{remark}
  % \label{rmk:BoxFiniteNotation}
Indeed, if  $\mathcal{S} \subset \mathcal{D}$
is simultaneously bounded, ODBB, and UODF, then there exists a  sequence
$\{M_k\}_{k\in \N} \subset \N$ so that
% if $\mathbb{B}_k \subset \mathbb{W},\; k\in \N,$ is as in Eq.~\ref{rmk:BoxNotation} then
for every $D\in \overline{\mathcal{S}}$ and every $k \in \N$ we have that
\begin{equation}
  \label{eq:boxDefnMk}
\Mult(
D , \mathbb{B}_k
)\leq M_k
\end{equation}
% \end{remark}
We can use this to prove the following useful, technical lemma:

\begin{lemma}\label{lemma:LocalSeqCompc}
Let $R\subset \mathbb{W}$ be a relatively compact subset of $\mathbb{R}^2$.
If $\{D_n\}_{n\in \N} \subset \mathcal{D}$ is so that
\[
\sup\limits_{n\in \N} \;\;\Mult(D_n, R) < \infty
\]
then the restricted
sequence $\{D_n \cap R\}_{n\in \N}$ has a convergent
subsequence.
Specifically, there exists a diagram $D \in \DD$ with $D \subset \overline R$
% $D = (S,\mu)\in \mathcal{D}$
% with $S\subset \overline{R}$
and a strictly increasing function
$\varphi: \N \to \N$ so that
\[
\lim\limits_{n \to \infty} \left(D_{\varphi(n)}\cap R\right) = D.
\]
\end{lemma}
\begin{proof}
Let
 $M\in \N$  be so that
 \[
 \Mult(D_n , R) \leq M
 \]
for every $n\in \N$, which exists by hypothesis.
For each $n\in \N$, let $\mathbf{m}_{n} \in \{0, \ldots, M\}^M$ be the vector having as entries the integers
$
\mu_n(\xx)$ for $\xx \in D_n \cap R
$,
sorted in descending order and padded with zeros at the end as necessary.
The pigeonhole principle   implies that there is a   vector
$\mathbf{m} \in \{0,\ldots, M\}^{M}$ which repeats infinitely often.
That is, there is
a strictly increasing function $\phi: \N \longrightarrow \N$ such that $\mathbf{m} = \mathbf{m}_{\phi(n)}$ for all $n \in \N$.

If $\mathbf{m}$ is the zero vector, then $D_{\phi(n)}\cap R = \varnothing$ for all $n\in \N$, and thus we can let
$\varphi = \phi$,  $D = \varnothing$, which trivially satisfy
\[
\lim\limits_{n\to \infty} (D_{\varphi(n)} \cap R) = \varnothing.
\]

If on the other hand $\mathbf{m}$ is nonzero, then let $J\geq 1$ be so that
\[
\mathbf{m} = (m_1,\ldots, m_{J},0,\ldots, 0)
\]
and $M\geq m_1 \geq m_2 \geq \cdots \geq m_{J} >0$.
For $n\in \N$,
order the underlying set of $D_{\phi(n)} \cap R$ as $\{\xx^n_1,\cdots ,\xx^n_J\}$
% let $\xx^{n}_{j}\in D_{\phi(n)} \cap R$ be
in such a way that
$\mu_{\phi(n)}(\xx^{n}_{j}) = m_j$ for all $j=1,\ldots, J$ .
Then the collection
$\{(\xx^{n}_{1},\ldots, \xx^{n}_{J})\}_{n \in \N}$
is an infinite sequence in $R^{J}$
which, by compactness of $\overline{R}$, has an accumulation point
$(\xx_1,\ldots,\xx_{J})\in \overline{R}^J$.
Thus, let $\psi: \N \longrightarrow \N$ be strictly increasing
with the property that
\[
\lim\limits_{n\to \infty}
\left(\xx^{\psi(n)}_{1},\ldots, \xx^{\psi(n)}_{J}\right)
 = (\xx_{1},\ldots, \xx_{J}).
\]
Define  $D$ as the disjoint union
\[
D
=
\bigsqcup_{j=1}^{J}
\left(
\left\{\xx_j\right\}, m_j
\right)
\]
where $(\{\xx_j\}, m_j)$ is the persistence diagram having one off-diagonal point at $\xx_j \in \overline{R}$ with  multiplicity $m_j$.

We contend that the subsequence of $\{D_n \cap R\}_{n\in \N}$
defined by $\varphi = \phi \circ \psi$ converges to $D$.
Indeed, given $\varepsilon > 0$,  let $N\in \N$ be so that
$n\geq N$ implies
\begin{equation}
  \label{eq:lemmaMaxLinfty}
\max\limits_{1\leq  j \leq  J} \;
\left\|
\xx_j^{\psi(n)} - \xx_j
\right\|_\infty
< \varepsilon.
\end{equation}
Since the diagram $D_{\varphi(n)}\cap R$ has the collection
$\left\{\xx_1^{\psi(n)},\ldots, \xx_J^{\psi(n)}\right\}$ as underlying
set,  then
\[
\begin{array}{rrcl}
\gamma_n : & D_{\varphi(n)}\cap R & \longrightarrow & D \\
& \xx_j^{\psi(n)} & \mapsto & \xx_j
\end{array}
\]
defines
  a bijection (of multisets).
As no points are unmatched, Eq.~\ref{eq:lemmaMaxLinfty} implies that  this is an $\varepsilon$-matching.
Thus
$d_B(D_{\varphi(n)}\cap R ,D) \leq \varepsilon$ for all $n\geq N$, and convergence follows.

\end{proof}

With this lemma in place, we can return to the proof of the main theorem.

\begin{proof}[Thm.~\ref{thm:Compactness}]
($\Rightarrow$) If $\mathcal{S} \subset \mathcal{D}$ is
relatively compact, then it being bounded, ODBB and UODF
% off-diagonally birth bounded and uniformly off-diagonally finite
follow from
Propositions \ref{prop:RelCpcBounded}, \ref{prop:RelCpcBirthBounded}
and \ref{prop:RelCpcFinite}, respectively.

($\Leftarrow$) Let $\mathcal{S} \subset \mathcal{D}$ be bounded, ODBB, and UODF.
% off-diagonally birth bounded and uniformly off-diagonally finite.
Fix $\{M_k\}_{k\in \N} \subset \N$ and
$\mathbb{B}_k \subset \mathbb{W}$, $k\in \N$, as in Eq.~\ref{eq:defBox} and \ref{eq:boxDefnMk},
% Remark \ref{rmk:BoxFiniteNotation}
and let
$\{D_n\}_{n\in \N} \subset \overline{\mathcal{S}}$
be arbitrary.
We will use Lemma \ref{lemma:LocalSeqCompc}  inductively to
construct a sequence $\varphi_1, \varphi_2,\ldots, \varphi_k ,\ldots $
of strictly increasing functions $\varphi_k : \N \longrightarrow \N$
so that if $\Phi_k = \varphi_1 \circ \cdots \circ \varphi_k$,
then the subsequence of restricted diagrams
\[
\left\{
D_{\Phi_k(n)}\cap
 \mathbb{B}_k
\right\}_{n\in \N}
\]
converges to a diagram  $D^k \subset \overline{\mathbb{B}}_k$,
 for each $k\geq 1 $.
 Once we have built this, we will let
\[
\begin{array}{rccc}
\varphi: &\N &\longrightarrow & \N \\
& n& \mapsto & \Phi_n(n)
\end{array}
\]
and the main task for the proof is to  show that $\{D_{\varphi(n)}\}_{n\in \N}$ converges
to the diagram
\begin{equation*}
D =
\bigsqcup\limits_{k=1}^\infty D^k.
\end{equation*}

We now proceed inductively in $k$.
The base case follows from applying Lemma \ref{lemma:LocalSeqCompc} to the sequence $\{D_n\}_{n\in \N}$ and the relatively compact set $\mathbb{B}_1$.
This results in a strictly increasing function
$\varphi_1 : \N \to \N$ and a diagram
$D^1 \subset \overline{\mathbb{B}}_1$, so that
\[
\lim\limits_{n \to \infty} D_{\varphi_1(n)} \cap \mathbb{B}_1 = D^1.
\]
Now the inductive step.
Let $k\geq 1$ and assume that $\varphi_j : \N \to \N$
and $D^j \subset \overline{\mathbb{B}}_j$, $1\leq j \leq k$,
% $(D^1,\mu^1),\ldots (D^k, \mu^k) \in \mathcal{D}$
have been constructed in such a way that
 if
$\Phi_j = \varphi_1 \circ \cdots \circ \varphi_j$,
then
\[
\lim\limits_{n \to \infty} D_{\Phi_j(n)}\cap \mathbb{B}_j = D^j.
\]
The sequence $\{D_{\Phi_k(n)}\}_{n\in \N}$ and the set
$\mathbb{B}_{k+1}$ satisfy the hypotheses of Lemma \ref{lemma:LocalSeqCompc}.
Thus, there exist a strictly increasing function $\varphi_{k+1} : \N \longrightarrow  \N $ and a diagram   $D^{k+1} \subset \overline{\mathbb{B}}_{k+1}$
so that if $\Phi_{k+1} = \Phi_k\circ \varphi_{k+1}$, then
\[
\lim\limits_{n \to \infty} D_{\Phi_{k+1}(n)}
\cap \mathbb{B}_{k+1} = D^{k+1}.
\]

Now, let $\varphi(n) = \Phi_n(n)$, and let us show that $\{D_{\varphi(n)}\}_{n\in \N}$ converges
to
$D = \bigsqcup\limits_{k=1}^\infty D^k$.
To this end, fix $\varepsilon > 0$ and let $K \in \N$ be large enough so that $\tfrac{C}{K} < \tfrac{\varepsilon}{2}$.
For each $1 \leq k \leq K$, let $N_k \in \N$ be so that
$n\geq N_k$ implies
\[
d_B\left(D_{\Phi_k(n)}\cap \mathbb{B}_k , D^k\right)
<
\frac{\varepsilon}{2}
\]
and let $N = \max\{K,N_1,\ldots, N_K\}$.
Note  that if $n > N$ and $1\leq k \leq K$, then
\[
\varphi(n) =
\Phi_k (\varphi_{k+1} \circ \cdots \circ \varphi_n (n) ).
\]
Moreover,
since
$\varphi_{k+1} \circ \cdots \circ \varphi_n (n) \geq n  > N_k$,
then
$d_B\left(D_{\varphi(n)}\cap \mathbb{B}_k, D^k\right) <
\frac{\varepsilon}{2}$.
Thus, we can assume we have an $\varepsilon/2$-matching
\[
\gamma_n^k : D_{\varphi(n)}\cap \mathbb{B}_k \longrightarrow D^k.
\]
As the $\mathbb{B}_k$'s are disjoint,  then the union of the
$\gamma_n^k$'s
yields  a bijection (of multisets)
\[
\Gamma_n^K: \; D_{\varphi(n)} \cap
\bigcup_{k \leq K} \mathbb{B}_k  \;
\longrightarrow
\bigsqcup_{k \leq K} D^k
\]
Moreover, since all points in $D_{\varphi(n)} \cap \bigcup\limits_{k>K} \mathbb{B}_k$
have persistence at most $\varepsilon/2$,
% and $\bigsqcup\limits_{ k > K} A^k$ with the diagonal,
then $\Gamma_n^K$ defines an $\varepsilon/2$-matching
between  $D_{\varphi(n)}$ and  $D$, and
hence $d_B(D_{\varphi(n)}, D) < \varepsilon$ for all $n > N$.

\end{proof}

\subsection{Bottleneck vs Wasserstein Compactness}\label{sec:WassVSBottle}
The Wasserstein distance $d_{W_p}$, for $p\in \N$, is another common measure of similarity between persistence diagrams.
It is given by
\begin{equation}\label{eq:WassDist}
d_{W_p}(D,D') \;\;:=\;\;
\inf\limits_{\mathsf{M}}
\left(
\sum_{\xx, \xx' \atop \mbox{\tiny matched} } \|\xx - \xx'\|^p_\infty \;\;+
\sum_{\zz \atop \mbox{\tiny unmatched}}  \left(\frac{\pers(\zz)}{2}\right)^p
\right)^{1/p}
    \end{equation}
where the infimum ranges  over all partial matchings $\mathsf{M}$ between $D$ and $D'$.
One can show that $d_{W_p}$ defines a metric on the set
\[
\DD_p\; := \;\left\{D \in \DD \mid d_{W_p}(D, \varnothing) < \infty\right\}
\]
and that $(\DD_p, d_{W_p})$ is a complete  separable metric space \cite{Mileyko2011}. Moreover,
\begin{equation}\label{eq:FiltBottleneck}
\DD_{1} \subset \DD_{2} \subset \cdots \subset \DD_p \subset \cdots \subset \DD
\end{equation}
and   all  inclusions can be shown to be continuous with respect to the appropriate metrics.
In particular, if $\iota_p : \DD_p \hookrightarrow \DD$ is the inclusion map  and
$\mathcal{S} \subset \DD_p$ is relatively-compact with respect to $d_{W_p}$, then
$\iota_p\left(\overline{\mathcal{S}}\right)$ is compact (hence closed) with respect to $d_B$,
and the equality
\[
\iota_p\left(\overline{\mathcal{S}}\right)
=
\overline{
\iota_p\left(\overline{\mathcal{S}}\right)}
=
\overline{\iota_p(\mathcal{S})}
\]
shows that   relatively-compact subsets of $(\DD_p, d_{W_p})$ are also
relatively compact in $(\DD, d_B)$.
This implies that relatively-compact subsets  of $(\DD_p, d_{W_p})$ also satisfy the  three conditions of  Theorem
\ref{thm:Compactness}.
The characterization of relative-compactness in $(\DD_p, d_{W_p})$ is work of \cite{Mileyko2011}, and we will describe it next.
\begin{definition} A set $\mathcal{S} \subset \DD_p$ is \textbf{uniform} if for every $\epsilon >0 $ there exists $\alpha > 0$
so that
\[
\sum_{\xx \in S \atop \pers(\xx) < \alpha} \mu(\xx)\pers(\xx)^p \leq \epsilon
\]
for all $(S,\mu) \in \mathcal{S}$.
\end{definition}
The relevant characterization is as follows:
\begin{theorem}[\cite{Mileyko2011}]\label{thm:CompactnessWasserstin}
A set $\mathcal{S} \subset \DD_p$ is relatively compact in $(\DD_p, d_{W_p})$ if and only if it is bounded (with respect to $d_{W_p}$), off-diagonally birth bounded
(in the sense of Definition \ref{def:ODBB}) and uniform.
\end{theorem}

The next two examples illustrate how relative compactness in $(\DD, d_B)$ and $(\DD_p, d_{W_p})$
can exhibit very different behaviors, even if Theorems \ref{thm:Compactness} and \ref{thm:CompactnessWasserstin}
seem similar at first glance.\\

\noindent \textbf{Examples:}
\begin{enumerate}
  \item Let $\mathcal{S}$ be the set of persistence diagrams
    $D_{jk} = (S_{jk}, \mu_{jk})$, for $j,k \in \N$, with
    \[ S_{jk} = \left\{\left(0, \frac{1}{k n^{1/k}}\right) \Big\vert \; 1 \leq n \leq j \right\} \]
    and multiplicity function $\mu_{jk}\left(0 , \frac{1}{kn^{1/k}}\right) = 1$.
    Notice that   each $S_{jk}$ is  finite, and hence $\mathcal{S}\subset \DD_p$ for all  $p\in \N$.
    However, for $k = p$ the set $\{ D_{jp} \}_{j\in \N} \subset \mathcal{S}$
    is not bounded with respect to $d_{W_p}$, and thus $\mathcal{S}$ is not
    relatively compact in   $(\DD_p, d_{W_p})$ for any $p\in \N$.

    On the other hand, $d_B(D_{jk}, \varnothing )\leq 1$ for every $j,k\in \N$, so
    $\mathcal{S}$ is bounded with respect to $d_B$, and it  is clearly birth-bounded
    (hence ODBB). In order to see that $\mathcal{S}$ is UODF, fix $\epsilon >0$
    and let $k_0 \in \N$ be so that $\frac{1}{\epsilon} < k_0$.
    It follows that $D_{jk} \cap    \overline{\W^\epsilon} = \varnothing$ for  $k \geq k_0$
    and every $j\in \N$,
    so assume  $1 \leq k < k_0$, $j\in \N$,  and let $\left(0, \frac{1}{kn^{1/k}}\right) \in S_{jk} \cap \overline{\W^\epsilon}  $.
    Hence
    \[
    n \;\leq \; \frac{1}{(k\epsilon)^k}  \;\leq \; \frac{1}{\epsilon^k} \;< \; k_0^{k_0}
    \]
    and therefore $\Mult\left(D, \overline{\W^\epsilon}\right) < k_0^{k_0}$ for every $D\in \mathcal{S}$.
    This shows that $\mathcal{S}$ is UODF and thus relatively compact in $(\DD_,d_B)$.
  \item Let $\mathcal{S}$ be the set of persistence diagrams $D_k = (S_k,\mu_k)$, for $k\in \N$,
    with  \[S_k = \left\{(0,1), \left(0, n^{-1/k}\right)\Big\vert\; n \geq 2\right\}\]
    and multiplicity function $\mu_k(0,1)=k$, $\mu_k\left(0, n^{-1/k}\right)= 1$.
    The first thing to note is that $\mathcal{S} \subset \DD$, and that $\Mult\left(D_k, \overline{\W^{0.5}}\right) > k$ for every $k\in \N$.
    Therefore $\mathcal{S}$ is not UODF, and thus---by Theorem \ref{thm:Compactness}---it is not relatively compact in $(\DD, d_B)$.

    On the other hand, for each $p \in \N$  we have that
    $\mathcal{S} \cap \DD_p = \{D_k \mid 1\leq k  < p \}$, and thus
    \[
    \mathcal{S} = \bigcup_{p \in \N} \mathcal{S} \cap \DD_p .
    \]
    Since  $\mathcal{S} \cap \DD_p$ is  finite, then it is compact in $(\DD_p, d_{W_p})$ for every $p\in \N$.
\end{enumerate}

In summary, in order to understand relative compactness in $(\DD,d_B)$ it is not enough to
do so at each $(\DD_p, d_{W_p})$;
the class of relatively compact subsets of $(\DD,d_B)$ is in fact larger than the union over $p\in \N$ of those in
$(\DD_p, d_{W_p})$.

\subsection{Consequences of Thm.~\protect\ref{thm:Compactness}}
We note a few  consequences of our  characterization of relative compactness in $(\DD,d_B)$.
\begin{theorem}\label{thm:CompactEmptyInterior}
  Relatively compact subsets of $(\mathcal{D}, d_B)$ have empty interior.
\end{theorem}
\begin{proof} Let $\mathcal{S} \subset \mathcal{D}$ be relatively compact, and let
\[
\mathbb{B}_k =
\left\{
\xx \in  \mathbb{W} \;\Big\vert \; 0 \leq \birth(\xx) < C_k \, \mbox{ and }\, \frac{C}{k+1} \leq \pers(\xx) < \frac{C}{k}
\right\}
\;\; , \;\; k \in \mathbb{N}
\]
 be a sequence of boxes   (as defined in Eq.~\ref{eq:defBox}) so that $D\subset \bigcup_k\mathbb{ B }_k$ for every $D \in \mathcal{S}$.
 Fix $D \in \mathcal{S}$.
 We will show that any open ball around $D$ contains a persistence diagram $D'$ whose underlying set
 is not in the union of these boxes.
 Indeed, given $\epsilon > 0$, there exists $k_0\in \mathbb{N}$ so that $\frac{C}{k_0} < \frac{\epsilon}{2}$,
 and if $D'$ is the persistence diagram obtained from $D$ by adding the point $ \left(C_{k_0} + 1, C_{k_0} + 1 + \frac{C}{k_0}\right)$
 with multiplicity one,  then   $d_B(D,D') < \epsilon$, but $D' \not\subset \bigcup_k\mathbb{ B }_k$.

\end{proof}

%\LM{Jose, can you decide how much extra proof is needed around these corollaries?}

Recall that a topological space is    \textit{locally compact} if every
point has an open neighborhood contained in a compact set;
said open set is called a compact neighborhood of the point.
The following corollary is a direct consequence of
Theorem \ref{thm:CompactEmptyInterior}.

\begin{corollary}\label{coro:CompactNowhereDense}
The space of persistence diagrams $(\mathcal{D}, d_B)$ is not locally compact. Moreover, no diagram $D\in \mathcal{D}$ has a compact neighborhood.
\end{corollary}
\begin{proof}
Let $D\in \mathcal{D}$ and suppose, by way of contradiction,
that $D$ has a compact neighborhood.
That is, that there exist an open set $\mathcal{U} \subset \mathcal{D}$ and a compact
set $\mathcal{S} \subset \mathcal{D}$ so that
$D \in \mathcal{U} \subset \mathcal{S}$.
Since taking interiors preserves the order of inclusions, then
\[
\mathcal{U} = \mathsf{int}(\mathcal{U})
\subset
\mathsf{int}(\mathcal{S})
= \emptyset
\]
which   contradicts $D\in \mathcal{U}$.

\end{proof}

As we will see next, the lack of enough compact regions in the space of persistence diagrams  also has global implications
for learning tasks.
Recall that a subset of a topological space is called \textit{nowhere dense}
if its closure has empty interior.
It follows from Corollary \ref{coro:CompactNowhereDense}
that in the space of persistence diagrams
all compact sets are nowhere dense.
Moreover, since $(\mathcal{D}, d_B)$ is complete \cite{blumberg2014robust},
then the Baire Category Theorem \cite{baire1899fonctions}---which contends that no complete metric space can be written
as the countable union of nowhere dense subsets---implies the following.

\begin{corollary}
$\mathcal{D}$ cannot be written as the countable union of compact
subsets.
\end{corollary}

Many optimization tasks (e.g., gradient descent) leverage the compactness
of a space
to argue that solutions can be found as  limits to  sequential processes.
The fact that $(\DD,d_B)$ cannot be written as the countable union of compact subsets,
implies that non-convex global optimization problems in $\DD$ cannot be guaranteed
to be solvable via this type of local decompositions.
As we will see below, this also implies that it is not possible to find a metric
in $C(\DD,\mathbb{R})$
that characterizes  approximations with respect to the topology of uniform
convergence on compact sets.

\begin{definition}
Let $X,Y$ be topological spaces and let
$C(X,Y)$ denote the set of continuous functions from $X$ to $Y$.
Given  $K\subset X$ compact and $V\subset Y$ open,
let
\begin{equation*}
U(K,V) = \{f \in C(X,Y) \mid f(K) \subset V\}.
\end{equation*}
The collection
\begin{equation*}
\{U(K, V) \mid K\subset X \mbox{ compact},\;\; V\subset Y \mbox{ open}\}
\end{equation*}
forms a subbase for a topology on $C(X,Y)$, called the
\textit{compact-open topology}.
When $Y$ is a metric space,
a sequence of continuous functions $f_n : X \longrightarrow Y $, $n\in \N$,
converges to $f$ in the compact-open topology, if and only if $\{f_n\vert_K\}_{n\in \N}$
converges uniformly to $f\vert_K$ for each compact
set $K\subset X$.
\end{definition}

Since $\mathcal{D}$ cannot be written as a countable union of compact sets,
then we have the following.

\begin{corollary}\label{cor:NonMetrizable}
The compact-open topology on $C(\mathcal{D}, \mathbb{R})$ is not metrizable.
\end{corollary}
\begin{proof}
  See Example 2.2, Chapter IV, of \cite{conway2013course}.

\end{proof}

Sequential optimization tasks in machine learning
typically proceed by updating the current solution
until no further progress is made.
Solutions are often functions (e.g., a classifier) and progress
in the optimization is measured through loss functionals (e.g., like mean-squared error) and/or the distance between the current and prior state.
The latter being particularly useful in convex optimization.
Since compact-open
sequential convergence in $C(\mathcal{D}, \mathbb{R})$ cannot be measured
with a metric (Corollary \ref{cor:NonMetrizable}),
then other alternatives are needed.
We describe such methodologies in the next sections.

%--------------
%--------------
\section{Linearizing $\DD$}
\label{sec:LinearizingD}
%--------------
%--------------
%

The fact that the compact-open topology on $C(\mathcal{D},\mathbb{R})$ is not metrizable, implies that
  optimization with respect to compact convergence needs to be handled with care. The goal of this section is to provide methods for doing this.
We begin with a definition.

\begin{definition}\label{def:CoordinateSystem}
A \textbf{coordinate system} for $\mathcal{D}$ is a collection $\mathcal{F} \subset C(\mathcal{D}, \mathbb{R})$  which separates points.
That is, if $D,D'\in \mathcal{D}$ are distinct then there exists $F\in \mathcal{F}$ for which $F(D)\neq F(D')$.
\end{definition}

Of course one could take $\mathcal{F}$ to be the space of all  real-valued continuous functions on $\mathcal{D}$, but  this is an extreme case; the quality of a coordinate system is  determined by its size---the smaller the better.
The metaphor to keep in mind is Euclidean space, $\mathbb{R}^n$.
In this case, an oblique coordinate system (e.g.~Cartesian coordinates) is uniquely determined by a linear basis  for the space $\mathcal{L}(\mathbb{R}^n)$ of (continuous) linear functions from $\mathbb{R}^n$ to $\mathbb{R}$.

Our goal is to coordinatize the space of persistence diagrams by finding a continuous embedding of $\mathcal{D}$ into an appropriate topological vector space $V$,
and taking the restriction to $\mathcal{D}$ of elements from $\mathcal{L}(V)$.
The problem of  embedding persistence diagrams into topological vector spaces has received considerable
recent attention.
For instance, it is known that $(\DD_p, d_{W_p})$ does not admit an inner product
compatible with $d_{W_p}$ for any $p > 2$, and that
 $(\DD_{p},d_{W_p})$
does not admit a coarse embedding (this is weaker than being bi-Lipschitz) into any Hilbert space for any $ p>2$ \cite{wagner2019nonembeddability}.

In order to choose $V$ we will use two principles.
First, that persistence diagrams can be interpreted as Radon measures on $\mathbb{W}$ \cite{divol2021understanding,Chazal2016},
which suggests embedding $\mathcal{D}$ into the dual space of some set of continuous real-valued functions on $\W$.
Second,  that  $\mathcal{D}$ is a topological monoid:
the sum of two persistence diagrams $D,D' \in \DD$ is their disjoint union $D\sqcup D'$ as multisets,
the empty diagram $\varnothing$ is the identity, i.e.~$D \sqcup \varnothing = D$, and
the operation
$\sqcup: \DD\times \DD \longrightarrow \DD$
is associative and continuous (see also \cite{bubenik2019universality}).
In what follows we will construct an embedding $\nu:\DD\hookrightarrow V$
which recovers the measure-theoretic interpretation of persistence diagrams,
and preserves the monoidal structure of $\DD$ (see Theorem \ref{thm:Linearization}).
In addition, we will show that appropriate subsets of   $\mathcal{L}(V)$ will yield
coordinate systems for $\mathcal{D}$ (see Theorem \ref{thm:templates}), and these in turn will generate
dense subsets of $C(\mathcal{D}, \mathbb{R})$ with respect to the compact-open topology (see Theorem \ref{thm:approximation}).

%--------------
\subsection{Topological vector spaces, duals and their topologies}
%--------------
We will first review some basics of topological vector spaces, following \cite{conway2013course}.
Let $V$ be a topological vector space; that is, a vector space endowed with a topology so that addition and scalar multiplication are continuous
functions.
Its (topological) dual is the vector space
\[
V' = \{ T: V \longrightarrow \mathbb{R} \mbox{ , so that }  T \mbox{ is linear and continuous}\}
\]
In particular, if the topology on $V$
comes from a norm $\|\cdot\|_V$, then we write $V^*$
instead of $V'$.
If $V^*$ is endowed with the operator norm
\[
\|T\|_* = \sup_{ \|\vv\|_V =1}
\lvert T(\vv)\rvert,
\]
then $V^*$ is in fact a Banach space.
There are three standard topologies on $V^*$:
\begin{description}
  \item[Strong:] The strong topology is the one generated by the operator norm $\|\cdot\|_*$. A basis for open neighborhoods of a point $T\in V^*$ is given by sets of the form
    \[ B_\epsilon(T) =
    \left\{T' \in V^* \;\Big\vert\; \sup\limits_{\|\vv\|_V = 1}\lvert T(\vv) - T'(\vv)\rvert < \epsilon \right\}
    \] where $\epsilon > 0$.
    In particular, a sequence $\{T_n\}_{n\in \N} \subset V^*$ converges to $T\in V^*$ in the strong topology if and only if
$\{T_n(\vv)\}_{n\in \N}$ converges to $T(\vv)$  uniformly in $\vv \in V$.
  \item[Weak:] If $V^{**}$ denotes the dual of the normed space $(V^*, \|\cdot\|_*)$, then the weak topology on $V^*$ is the smallest topology so that every $\mathcal{T} \in V^{**}$ is continuous.
      A basis for open neighborhoods of a point $T\in V^*$ is given by sets of the form
    \[
    N(\mathcal{T}_1,\ldots, \mathcal{T}_I ; \epsilon) (T) =
\left\{
T' \in V^* \;\Big\vert\;
\max\limits_{1\leq i \leq I}
\lvert \mathcal{T}_i\left(T'\right) - \mathcal{T}_i\left(T\right)\rvert
< \epsilon
\right\}
    \]
     where $\mathcal{T}_1,\ldots, \mathcal{T}_I \in V^{**}$ and $\epsilon >0$.
In particular, $\{T_n\}_{n\in\N} \subset V^*$ converges to $T\in V^*$ in the weak topology
if and only if $\{\mathcal{T}(T_n)\}_{n\in \N}$
converges to $\mathcal{T}(T)$ for all $\mathcal{T} \in V^{**}$.
  \item[Weak-*:] The weak-* topology is the smallest topology so that for each $\vv \in V$, the resulting evaluation function
\[
\begin{array}{rccl}
e_\vv : & V^* & \longrightarrow & \mathbb{R} \\
& T& \mapsto & T(\vv)
\end{array}
\] is continuous.
A basis for open neighborhoods of $T\in V^*$ is given by sets of the form
\[
N(\vv_1,\ldots, \vv_I; \epsilon)(T) =
\left\{
T' \in V^* \;\Big\vert\; \max\limits_{1\leq i \leq I}
\lvert T'(\vv_i)-T(\vv_i) \rvert < \epsilon
\right\}
\] where $\vv_1,\ldots, \vv_I \in V$ and $\epsilon > 0$.
A sequence $\{T_n\}_{n\in \N} \subset V^*$
converges to $T\in V^*$ in the weak-* topology
if and only if $\{T_n(\vv)\}_{n\in \N}$ converges
to $T(\vv)$ for each $\vv \in V$. The convergence, however, need not  be uniform in $\vv$.
\end{description}
One can check that the weak-* topology is weaker than the weak topology, which in turn is weaker than the strong topology.

%--------------
\subsection{Linearizing the set of finite diagrams}
%--------------

It is useful to first illustrate some of the difficulties associated to finding embeddings for  the set of finite diagrams $\mathcal{D}_0$.
In what follows we will prove several negative results which will inform the choices in embedding $\DD$.
Indeed, the first thing to notice is that the set of compactly supported continuous  functions from $\W$ to $\mathbb{R}$, denoted $C_c(\W)$,
is a normed vector space if
endowed with the sup norm $\|\cdot \|_\infty $.

The Dirac mass centered at $\xx \in \W$ is the linear function
\[
\begin{array}{rccl}
  \delta_\xx :& C_c(\W) & \longrightarrow & \mathbb{R} \\
   & f  & \mapsto & f(\xx)
\end{array}
\]
and since $\lvert\delta_\xx(f)\rvert \leq \|f\|_\infty$
for each $f\in C_c(\W)$, it follows that
$\delta_\xx \in C_c(\W)^*$.
Let
\[
\begin{array}{rccc}
\nu_0 : &\mathcal{D}_0& \longrightarrow& C_c(\W)^* \\[.1cm]
& \varnothing & \mapsto & 0\\[.2cm]
& \varnothing \neq (S,\mu) & \mapsto & \sum\limits_{\xx \in S} \mu(\xx)\delta_\xx.
\end{array}
\]
It is not hard to see that
\begin{proposition}
$\nu_0 : \mathcal{D}_0 \longrightarrow C_c(\mathbb{W})^*$
is injective, and satisfies $\nu_0(D \sqcup D') = \nu_0(D ) + \nu_0(D') $ for every $D,D' \in \DD_0$.
\end{proposition}

Whether   $\nu_0$ is continuous or not depends
on the topology with which $C_c(\W)^*$ is endowed.
We start with the coarser topologies, but immediately have the following negative results.

\begin{proposition}
    If $C_c(\W)^*$ is endowed with the weak topology, then $\nu_0$ is discontinuous at every point.
\end{proposition}
\begin{proof}
Fix $D \in \mathcal{D}_0$, and let $D_n\in \mathcal{D}_0$ be the diagram obtained from
$D$ by adding the point $(1, 1 + 1/n)$ with multiplicity $n$.
It follows that $\{D_n\}_n$ converges to $D$ with respect to the bottleneck distance.
We contend that $\nu_0(D_n)$ does not converge to $\nu_0(D)$ with respect to the weak topology;
in other words, we will show that there exist $\epsilon_0 > 0$ and  a linear operator
$\mathcal{T}: C_c(\mathbb{W})^* \longrightarrow \mathbb{R}$, continuous with respect to the strong topology,
 so that $\lvert\mathcal{T}(\nu_0(D_n)) - \mathcal{T}(\nu_0(D))\rvert \geq \epsilon_0$ for infinitely many values of $n$.

Indeed, the first thing to notice is that since $\mathbb{W}$ is not compact,
$C_c(\mathbb{W})$ is not complete.
Its completion is the space
$C_0(\mathbb{W})$ of continuous functions on $\mathbb{W}$ which vanish at the diagonal $\Delta$ and at infinity.
More explicitly, $f\in C_0(\mathbb{W})$ if and only if it  is continuous,
and $f(x,y)\to 0$ whenever $y-x\to 0$ or $ x + y \to \infty$.
Let  $\iota : C_c(\W) \hookrightarrow C_0(\W)$ be the inclusion and let $ \iota^* : C_0(\W)^* \longrightarrow C_c(\W)^*$ be the induced homomorphism.
Since $C_c(\mathbb{W})$ is dense in $C_0(\mathbb{W})$, then
$\iota^*$ is an isometric isomorphism; its inverse
$\jmath^* : C_c(\mathbb{W})^* \longrightarrow C_0 (\mathbb{W})^*$
sends a continuous linear map   $T: C_c(\mathbb{W})\longrightarrow \mathbb{R}$ to
 its unique continuous linear extension  $\jmath^*(T): C_0(\mathbb{W}) \longrightarrow \mathbb{R}$.

Now, let $\varphi: \mathbb{R}^2 \longrightarrow [0,1]$ be a continuous (bump) function
so that
\[
\varphi(x,y) = \left\{
                 \begin{array}{ll}
                   1 & \hbox{ if } \max\{x, y-x\} < 2 \\[.2cm]
                   0 & \hbox{ if } \max\{x, y -x\} \geq 3.
                 \end{array}
               \right.
\]
It follows that $f(x,y) = (y-x)\cdot\varphi(x,y) \in C_0 (\mathbb{W})$,
and hence the evaluation function
$e_f  : C_0(\mathbb{W})^{*} \longrightarrow \mathbb{R}$ is a bounded linear operator.
Let $\mathcal{T}  : C_c(\mathbb{W})^* \longrightarrow \mathbb{R}$
be  the composition $e_f \circ \jmath^*$.
Then, for each $n > 1$
we have that
\[
\mathcal{T}(\nu_0 (D_n)) = \mathcal{T}(\nu_0(D)) + 1,
\]
and letting $\epsilon_0 = 1$ completes the proof.

\end{proof}

\begin{corollary} If $C_c(\W)^*$ is endowed with the strong topology, then $\nu_0$ is discontinuous at every point.
\end{corollary}
\begin{proof}
Indeed, the strong topology contains the weak topology.

\end{proof}

It is not until we pass to the weakest of the three standard topologies that we approach a useful result.

\begin{proposition}\label{prop:weak_star_continuity_finite}
If $C_c(\mathbb{W})^*$ is endowed
with the weak-* topology, then
$\nu_0 : \DD_0 \longrightarrow C_c(\W)^*$ is continuous.
\end{proposition}

Before presenting the proof, we have the following useful lemma.

\begin{lemma}\label{lemma:nu_f_continuous}
For each $f\in C_c(\W)$, the function
\begin{equation}\label{eq:def_nu_f}
\begin{array}{rccl}
 \nu_f : & \DD & \longrightarrow  & \mathbb{R} \\
  &(S,\mu) & \mapsto & \sum\limits_{\xx \in S  } \mu(\xx)f(\xx)
\end{array}
\end{equation}
is continuous.
\end{lemma}
\begin{proof}
The first observation is that the sum defining  $\nu_f$ is always finite, since the support of $f\in C_c(\W)$  intersects the underlying set of any persistence diagram at only finitely many points.
Let $D= (S,\mu)\in \DD$ and fix $\epsilon > 0$.
Since $\mathsf{supp}(f) \subset \W$ is compact, then $f$ is uniformly continuous.
Further, there exists
 $\delta > 0$ for which $\mathsf{supp}(f) \subset \W^{2\delta} $,
and
\begin{equation}\label{eq:continuity_epsilon}
\lvert f(\xx) - f(\yy)\rvert < \frac{ \epsilon} {\sum\limits_{\zz \in S \cap \W^{\delta}} \mu(\zz) }
\end{equation}
 whenever $\|\xx - \yy\|_\infty < \delta$.

Let $D' = (T,\alpha)\in \DD$ be given with $d_B(D,D') < \delta$.
We will show that  $\lvert \nu_f(D) - \nu_f(D')\rvert < \epsilon$.
Fix a $\delta$-matching
\[
\begin{array}{rccc}
\M : & A & \xrightarrow{\cong} & B\\
& \mathbin{\rotatebox[origin=c]{270}{$\subseteq$}} && \mathbin{\rotatebox[origin=c]{270}{$\subseteq$}}\\
& S_\mu && T_\alpha.
\end{array}
\]
This means that if $(\xx, k) \in A$ and $(\yy, n) = \mathsf{M}(\xx, k)$ are matched, then $\|\xx - \yy\|_\infty < \delta$;
and if $(\zz, m) \in \left(S_\mu \smallsetminus A \right) \cup\left( T_\alpha \smallsetminus B\right)$ is unmatched,
 then $\pers(\zz) < 2\delta$.
In this case   $\zz \notin \mathsf{supp}(f)$, which implies $f(\zz) = 0$.
Hence
\begin{eqnarray*}
  \nu_f(D) &=& \sum_{\xx \in S} \mu(\xx)f(\xx) \\
%   &=& \sum_{(\xx, k) \in S_\mu} f(\xx) \\
   &=& \sum_{(\xx, k) \in A} f(\xx)
\end{eqnarray*}
and similarly,
\begin{eqnarray*}
  \nu_f(D') &=&
  % \sum_{(\yy, n) = \mathsf{M}(\xx,k) \atop (\xx,k)\in S_\mu'} f(\yy).
  \sum_{(\yy, n) \in B} f(\yy).
\end{eqnarray*}
Therefore
\begin{eqnarray*}
  \lvert \nu_f(D) - \nu_f(D')\rvert  &=& \left\lvert\sum_{(\xx, k) \in A} f(\xx) -   \sum_{(\yy,n) \in B}
  % \sum_{(\yy, n) = \mathsf{M}(\xx,k) \atop (\xx,k)\in S_\mu'}
  f(\yy) \right\rvert
   \\
   &\leq&
   \sum_{(\yy, n) = \mathsf{M}(\xx,k) \atop (\xx,k)\in A}
    \lvert f(\xx) - f(\yy)\rvert
\end{eqnarray*}
where each term $\lvert f(\xx) - f(\yy)\rvert$ is potentially
nonzero only when $\xx $ or $\yy  $ are  in $\mathsf{supp}(f) \subset \W^{2\delta}$.
Since in this case $\|\xx - \yy\|_\infty < \delta$, we would get  $\xx,\yy \in \W^\delta$.
Combining this observation with equation (\ref{eq:continuity_epsilon}) completes the proof.

\end{proof}

\begin{proof}[Proposition \ref{prop:weak_star_continuity_finite}]
Let $D  \in \DD_0$, and    fix a weak-* basic neighborhood
$N(f_1, \ldots, f_I ; \epsilon)$  for $\nu_0(D)$.
Notice that for each $i=1,\ldots, I$ we have  $\nu_0(D)(f_i) = \nu_{f_i}(D)$.
Since $\nu_{f_i}$ is continuous at $D$, then given $\epsilon > 0$ there exists $\delta_i > 0$ so that
$d_B(D,D') < \delta_i$ implies $\lvert \nu_{f_i}(D) - \nu_{f_i}(D')\rvert < \epsilon$.
If we let $\delta = \min\{\delta_1 , \ldots, \delta_I\}$, it follows that
whenever $d_B(D,D') < \delta$ then for all $i=1,\ldots, I$
\begin{equation*}
\lvert \nu_0(D)(f_i) - \nu_0(D')(f_i)\rvert  =   \lvert \nu_{f_i}(D) - \nu_{f_i}(D')\rvert < \epsilon.
\end{equation*}
This shows that $\nu_0(D') \in N(f_1,\ldots, f_I; \epsilon)$ and hence $\nu_0$ is continuous.

\end{proof}

These results imply that out of the three standard topologies on $C_c(\mathbb{W})^*$,
the weak-* topology is the only one for which $\nu_0$ yields a continuous embedding of $\mathcal{D}_0$
into $C_c(\W)^*$.
The question now is whether this embedding can be extended to $\mathcal{D}$.
The answer, as it turns out, is negative.
\begin{proposition}
If $C_c(\W)^*$ is endowed with the weak-* topology, then
$\nu_0: \DD_0 \longrightarrow C_c(\W)^*$ cannot be   continuously extended to any $D\in \DD\smallsetminus\DD_0$.
\end{proposition}
\begin{proof}
Assume, by way of contradiction, that $\nu_0$ extends continuously to some  $D = (S,\mu)\in \DD\smallsetminus \DD_0$.
If for each $n\in \N$ we let $D_n$ be the restriction of $D$ to $\W^{1/n}$, then
$D_n \in \DD_0$ for all $n\in \N$, and  the sequence  $\{D_n\}_{n\in \N} $ converges to $D$ with respect to the bottleneck distance.
By the continuity assumption of $\nu_0$ at $D$, we have that
\[
\nu_0(D) = \lim_{n\to \infty} \nu_0(D_n)
\]
where convergence is with respect to the weak-* topology.
In other words,
\[
\nu_0(D)(f) = \lim_{n\to \infty} \nu_0(D_n)(f)
\]
for every $f\in C_c(\W)$.
It follows that, given $f\in C_c(\W)$,  there exists $N_f\in \N$ so that
$\mathsf{supp}(f) \subset \W^{1/n}$ for all $n\geq N_f$,
and therefore the sequence $\nu_0(D_n)(f)$ becomes constant
and equal to
\begin{equation}\label{eq:Extension_nu_0}
\sum\limits_{\xx \in S} \mu(\xx) f(\xx).
\end{equation}
We claim that if $C_c(\W)$ is endowed with the sup norm
$\|\cdot\|_\infty$, then the linear function
\[
\begin{array}{rccl}
 \nu_0(D): & C_c(\W) & \longrightarrow  & \mathbb{R} \\
   & f & \mapsto & \sum\limits_{\xx \in S} \mu(\xx) f(\xx)
\end{array}
\]
is discontinuous at every point.
To this end, we will show that given $f\in C_c(\W)$ there exists
a sequence $\{f_n\}_{n\in \N} \subset C_c(\W)$ which converges to $f$
with respect to $\|\cdot\|_\infty$, but for which $\{\nu_0(D)(f_n)\}_{n\in \N}$ does not
converge to $\nu_0(D)(f)$.
This would contradict $\nu_0(D) \in C_c(\W)^*$.

Indeed, since $(S,\mu) = D \notin \DD_0$, then there exists a sequence $\{\xx_n\}_{n\in \N} \subset S \smallsetminus \mathsf{supp}(f)$ so that
$\pers(\xx_n) $ is strictly decreasing as $n$ goes to infinity.
Therefore, it is possible to construct a sequence
$\{r_n\}_{n\in \N}$ of positive real numbers,
so that the balls $B_{r_n}(\xx_n)\subset \W$ are all disjoint
 and disjoint with the support of $f$.
Let $\phi_n : \mathbb{R}^2 \longrightarrow [0,\infty)$ be the bump function
\[
\phi_n(\xx) = \frac{\max\{ 0 \; ,\;r_n - 2\|\xx - \xx_n\| \}}{r_n}
\]
supported on the closure of $B_{\frac{r_n}{2}}(\xx_n)$, and let
\[
f_n = f + \frac{\phi_1 + \cdots + \phi_n}{n}.
\]
It follows that $\{f_n\}_{n\in \N}$ is a sequence
of continuous and compactly supported functions on
$\W$, so that
$\|f_n - f\|_\infty < \frac{1}{n} $ for all $n\in \N$,
and for which
\begin{eqnarray*}
  \nu_0(D)(f_n)-  \nu_0(D)(f) &=&    \frac{1}{n}\nu_0(D)(\phi_1 + \cdots + \phi_n)\\[.3cm]
    &\geq &
    \frac{1}{n}\big(\mu(\xx_1)\phi_1(\xx_1) + \cdots + \mu( \xx_n)\phi_n(\xx_n) \big) \\[.4cm]
    &\geq&   1.
\end{eqnarray*}
Hence   $\{\nu_0(D)(f_n)\}_{n\in \N}$
does not converge to $\nu_0(D)(f)$, and  $\nu_0(D)$ is discontinuous at $f$.

\end{proof}

%-------------------
\subsection{Linearizing infinite diagrams}
\label{sec:LinInfDiag}
%-------------------

There are two main lessons to draw from the previous results:
First, that even though there is a candidate for extending $\nu_0$ to infinite diagrams,
namely Eq.~\ref{eq:Extension_nu_0}, the topology on $C_c(\W)$ induced by the sup norm is inadequate as it does not have enough open sets.
The second lesson is that a  weak-* topology on the dual of $C_c(\W)$ is the most likely to ensure continuity when embedding $\DD$.
In what follows we will describe a (locally convex) topology on $C_c(\W)$, and a corresponding weak-* topology on the topological dual $C_c(\W)'$ with the required properties.
We will utilize the theory of \emph{locally convex topological vector spaces}, which generalize  Banach spaces, and provide a rich framework in which to study weak topologies.
For a more detailed account we direct the interested reader to Chapters IV  and V of \cite{conway2013course}.

Let $\{K_{n}\}_{n\in \N}$ be a sequence of compact subsets of $\mathbb{W}$ so that $K_n \subset K_{n+1}$ for all $n\in \N$, and for which
\begin{equation*}
\mathbb{W} = \bigcup_{n\in \N} K_n.
\end{equation*}
It follows that each vector space
\[
C_c(K_n) = \{ f \in C(\mathbb{W}) \mid \mathsf{supp}(f) \subset K_n\}
\]
is a Banach space if endowed with the sup norm $\|\cdot\|_\infty$; in particular it is a \emph{locally convex space}.

\begin{definition}
A locally convex space is a topological vector space
$V$, whose topology is generated
by a family $\mathscr{P}$ of seminorms on $V$ which separate points.
More specifically,  $\mathscr{P} $ is a collection $ \{\rho_\alpha\}_{\alpha \in \Gamma}$ of continuous functions
$\rho_\alpha : V \longrightarrow [0,\infty)$
so that
\begin{enumerate}
  \item $\rho_\alpha(\uu + \vv) \leq \rho_\alpha(\uu) + \rho_\alpha(\vv)$ for all $\uu,\vv \in V$,
  \item $\rho_\alpha(\lambda \uu) = \lvert \lambda\rvert \rho_\alpha(\uu)$ for all scalars $\lambda$,
  \item $\rho_\alpha(\uu) = 0$ for all $\alpha \in \Gamma$ implies $\uu = \mathbf{0}$
\end{enumerate}
and   the topology of $V$ is the weakest for which
all the $\rho_\alpha$'s are continuous.
\end{definition}

In particular, all normed spaces are locally convex: any norm is a seminorm, and the norm topology is the smallest for which the norm is a continuous function.
Notice also that each inclusion
\begin{equation*}
C_c(K_n) \subset  C_c(K_{n+1}) \qquad n\in \N
\end{equation*}
is continuous and that
\begin{equation*}
C_c(\mathbb{W}) = \bigcup_{n\in \N} C_c(K_n).
\end{equation*}
The \textbf{strict inductive limit topology} on $C_c(\mathbb{W})$ is the finest locally convex topology so that each inclusion $C_c(K_n) \hookrightarrow C_c(\mathbb{W})$ is continuous.
In this topology,   a linear map $T: C_c(\mathbb{W}) \longrightarrow Y$
to a locally convex space $Y$ is continuous  if and only if
the restriction of $T$ to each $C_c(K_n)$ is continuous.
Moreover, this topology is independent of the choice of compact sets $\{K_n\}_{n\in\N}$ filtering $\W$.
%\reviewer{    Say what the family of seminorms on $C_c(W)$ is. }

Let $C_c(\mathbb{W})'$ denote the topological dual
of $C_c(\mathbb{W})$ with
respect to the strict inductive limit topology, and endow $C_c(\mathbb{W})'$
with the weakest topology so that for each $f\in C_c(\mathbb{W})$ the resulting evaluation function
\[
\begin{array}{rccl}
e_f : &C_c(\mathbb{W})'& \longrightarrow &\mathbb{R} \\
&T&\mapsto& T(f)
\end{array}
\]
is continuous. This is the corresponding weak-* topology.
It follows that a basis for neighborhoods of a point $T \in C_c(\mathbb{W})'$ is given by sets of the form
\[
N(f_1,\ldots, f_I;\epsilon)(T) =
\left\{
\tilde{T} \in C_c(\mathbb{W})' \;:\;
\lvert(T - \tilde{T})(f_i)\rvert < \epsilon \; , \; i= 1,\ldots, I
\right\}
\]
where $f_1,\ldots, f_I \in C_c(\mathbb{W})$ and $\varepsilon>0$.
Here is the main theorem of this section.

\begin{theorem}\label{thm:Linearization}
Given a persistence diagram $D =(S,\mu) \in \DD$ and a  function $f \in C_c(\W)$, define
\begin{equation}\label{eq:defNu}
\nu_D(f) := \sum_{\xx \in S} \mu(\xx) f(\xx).
\end{equation}
If $C_c(\W)$ is endowed with the
strict inductive limit topology,
and $C_c(\W)'$ is its topological dual endowed with the corresponding weak-* topology,
then
\[
\begin{array}{rccl}
\nu :& \mathcal{D} &\longrightarrow & C_c(\mathbb{W})' \\
& D& \mapsto &\nu_D
\end{array}
\]
is continuous, injective and satisfies
$\nu(D\sqcup D') = \nu(D) + \nu(D')$ for all $D,D'\in \DD$.
\end{theorem}
\begin{proof}
First, we ensure that $\nu_D$ is well defined.
Fix $D = (S,\mu)\in \mathcal{D}$ and $n\in \N$. Then $S\cap K_n$ is a finite
set and hence
for each
$f\in C_c(K_n)$ it follows that
\begin{equation*}
\nu_D(f) = \sum_{\xx \in S} \mu(\xx) f(\xx) < \infty.
\end{equation*}
As for   continuity of $\nu_D$,  fix $f_0 \in C_c(K_n)$, let $\varepsilon> 0$, and
let
\begin{equation*}
\delta < \frac{\varepsilon}{\sum\limits_{\zz\in S\cap K_n } \mu(\zz)}.
\end{equation*}
If $f\in C_c(K_n)$ is so that $\|f - f_0\|_\infty < \delta$,
then
\begin{eqnarray*}
\lvert\nu_D(f) - \nu_D(f_0)\rvert & =  &\left\lvert \sum_{\xx \in S} \mu(\xx)\big(f(\xx) - f_0(\xx)\big)\right\rvert  \\
   &\leq& \sum_{\zz \in S \cap K_n} \mu(\zz)\|f - f_0\|_\infty \\ \\
&<& \varepsilon.
\end{eqnarray*}
Therefore $\nu_D$ is a real-valued continuous linear
function on $C_c(K_n)$ for each $n$.
This shows that $\nu(D) = \nu_D \in C_c(\W)'$ for all $D\in \DD$.

To see that $\nu : \DD \longrightarrow C_c(\W)'$ is continuous,
we proceed exactly as in the proof of Proposition \ref{prop:weak_star_continuity_finite}.
Indeed,
let $D\in \DD$,
and fix a basic neighborhood $N(f_1,\ldots, f_I ; \epsilon)$ for $\nu(D)$.
For each $i =1,\ldots, I$ the function $\nu_{f_i} : \DD \longrightarrow \mathbb{R}$,
$\nu_{f_i}(D) = \nu_{D}(f_i)$,
is continuous by Lemma \ref{lemma:nu_f_continuous}, and hence
there exists $\delta > 0$ such that
\[
\lvert\nu(D)(f_i) - \nu(D')(f_i)\rvert = \lvert\nu_{f_i}(D) - \nu_{f_i}(D')\rvert < \epsilon
\]
for all $i=1,\ldots, I$,
whenever $d_B(D,D') < \delta$.

Injectivity of $\nu$ is deduced from the following observation.
If $(S,\mu), (T,\alpha) \in \DD$ are distinct,
then we can assume without loss of generality that there exists  $\xx \in S$
such that either: $\xx \notin T$,
or $\xx \in T$ and $\mu(\xx) \neq \alpha(\xx)$.
Let $f\in C_c(\W)$ be such that $\mathsf{supp}(f)\cap (S \cup T) = \{\xx\}$,
and for which $f(\xx) =1$.
If $\xx \notin T$, then
\begin{equation*}
\nu(S,\mu)(f) = \mu(\xx) \neq 0 = \nu(T,\alpha)(f).
\end{equation*}
Similarly, for the case where $\xx \in T$ we have
\begin{equation*}
\nu(S,\mu)(f) = \mu(\xx) \neq \alpha(\xx) = \nu(T,\alpha)(f),
\end{equation*}
which completes the proof.

\end{proof}

The Riesz-Markov representation theorem---see for instance Theorem 2.14 and Theorem 2.17 of \cite{rudin2006real}---contends that if $T: C_c(\W) \longrightarrow \mathbb{R}$ is linear and satisfies $T(f) \geq 0$ whenever $f(\xx)\geq 0$ for all $\xx\in \W$,
then there exists a unique positive Radon   measure $\eta$ on $\W$ so that
\[
\int_\W f d\eta = T(f)
\]
for all $f\in C_c(\W)$.
Specifically, $\eta$ is Borel regular and  $\eta(K) < \infty$ for every compact set $K \subset \W$.
Applying this theorem to elements in the image of $\nu: \DD \longrightarrow C_c(\W)'$
 implies that $ \nu(D)$ is a Radon measure on $\W$ for each $D\in \DD$.
This, of course, can be derived  directly from the definition
of $\nu(D)$ by writing it  (see Eq. \ref{eq:defNu}) as a sum of Dirac delta masses
\begin{equation}\label{eq:def_nu_delta}
\nu(S,\mu) = \sum_{\xx \in S} \mu(\xx)\delta_\xx
\end{equation}
The representation-theoretic view, however,  has the following advantages.
The first is that it recovers the interpretation
of persistence diagrams as rectangular measures
introduced by Chazal et. al. in \cite{Chazal2016}. Indeed,  Eq. \ref{eq:def_nu_delta}
yields exactly the counting measure
(Theorem 3.19, Eq. 3.6) from \cite{Chazal2016}.
The second advantage is that it provides
a natural framework in which to generalize persistence
diagrams: from those
in Eq. \ref{eq:def_nu_delta}, to general Radon measures on $\W$.
This viewpoint has been recently studied by
Divol and Lacombe in \cite{divol2021understanding}.
As they show, it
allows one to apply the mature
theoretical and computational tools from (partial) optimal transport,
to problems in the space of persistence diagrams
endowed with the Wasserstein distance.
It would be interesting to see---though outside  the scope of this work---if the approximation methods presented here, in particularly those in the next section, apply in the greater generality of persistence diagrams as Radon measures.

\section{Approximating Continuous Functions on Persistence Diagrams}
\label{sec:Approximating}
As we saw in Theorem \ref{thm:Linearization}, the function $\nu : \DD \longrightarrow C_c(\W)'$ provides a continuous
embedding so that  $\nu(D\sqcup  D')= \nu(D) + \nu(D')$
for all $D,D' \in \DD$.
We can now proceed to the task of finding coordinate systems  for $\DD$ (see Definition \ref{def:CoordinateSystem}).
The first thing to note is  that composing $\nu$ with elements from $C_c(\W)''$, the topological dual of $C_c(\W)'$, yields continuous real-valued  functions on $\DD$.
By construction, these functions also respect the monoidal structure $\sqcup$ of $\DD$.
The elements of $C_c(\W)''$ can be characterized as follows.
\begin{theorem}\label{thm:V_iso_V2dual}
Let $V$ be a locally convex space, and endow its topological dual $V'$  with the associated weak-* topology.
That is, the smallest topology such that all the evaluations
\[
\begin{array}{cccc}
  e_\vv & V' & \longrightarrow  & \mathbb{R} \\
   & T & \mapsto & T(\vv)
\end{array}
\]
for $\vv\in V$, are continuous.
Then  the function
\[
\begin{array}{rccl}
  e : & V & \longrightarrow  & V'' \\
   & \vv & \mapsto & e_\vv
\end{array}
\] is an isomorphism of locally convex spaces.
\end{theorem}

\begin{proof}
See Theorem 1.3 in Chapter V of \cite{conway2013course}.

\end{proof}
Applying this theorem to the locally convex space $C_c(\W)$, topologized  with the strict inductive limit topology,
implies that the elements of $C_c(\W)''$ are evaluations $e_f$, with   $f\in C_c(\W)$  uniquely determined.
Composing $e_f$ with $\nu$ yields a continuous function
$e_f \circ \nu : \DD\longrightarrow \mathbb{R}$ which preserves the monoidal structure $\sqcup$ of $\DD$.
Moreover,
given $D \in \DD$ we have that
\[
e_f\circ \nu(D) =
 \nu_D(f) = \nu_f (D)
\]
where $\nu_f : \DD \longrightarrow \mathbb{R}$ is defined by Eq.~\ref{eq:def_nu_f}.
We saw in Lem.~\ref{lemma:nu_f_continuous} that these types of functions
are indeed  continuous, but now we have the full picture:
they arise exactly as the continuous linear functions on a linearization of $\DD$.
The goal now is to construct coordinate systems for $\DD$ by selecting appropriate subsets
of $C_c(\W)$.
The   sets of interest are defined next.

\begin{definition}
A \textbf{template system} for $\DD$ is a collection $\mathcal{T}\subset C_c(\W)$ so that
\[
\mathcal{F}_{\mathcal{T}} = \{\nu_f \mid f\in \mathcal{T}\}
\]
 is a coordinate system (see Defn.~\ref{def:CoordinateSystem}) for $\DD$.
The elements of $\mathcal{T}$ are called template functions.
\end{definition}

The point of working with these template systems is that they can be used to approximate continuous functions on persistence diagrams, as given by the following theorem.

\begin{theorem}\label{thm:approximation}
Let $\mathcal{T}\subset C_c(\W)$ be a template system for $\DD$,
let $\mathcal{C} \subset \mathcal{D}$ be   compact,   and let   $F : \mathcal{C} \longrightarrow \mathbb{R}$ be continuous.
Then for every $\varepsilon > 0 $ there exist $N\in \N$, a polynomial
$p \in \mathbb{R}[x_1,\ldots, x_N]$
and template functions
$f_1,\ldots, f_N \in \mathcal{T}$   so that
\[
\left\lvert
 p\big( \nu_{f_1}(D) ,\ldots,  \nu_{f_N}(D)\big)
-
F(D)
\right\rvert < \varepsilon
\]
for every $D\in \mathcal{C}$.
That is,
the collection of functions of the form
\begin{equation}\label{eq:DenseSystem}
D \mapsto p\big( \nu_{f_1}(D) ,\ldots,  \nu_{f_N}(D)\big)
\end{equation} is dense in $C(\DD, \mathbb{R})$ with respect to the compact-open topology.
\end{theorem}

\begin{proof}
Let $\mathcal{T} \subset C_c(\W)$ be a template system for $\DD$ and let
\[
\mathcal{F} = \{\nu_f \mid f\in \mathcal{T}\} \subset C(\DD,\mathbb{R})
\] be the corresponding
coordinate system.
Let $\mathcal{A} \subset C(\DD, \mathbb{R})$ denote the algebra generated by $\mathcal{F}\cup \{1\}$.
In other words, $\mathcal{A}$ is the set of finite linear combinations of finite products of elements from $\mathcal{F}\cup \{1\}$.
It follows that every element of $\mathcal{A}$ can be written as
\[
p(\nu_{f_1},\ldots, \nu_{f_N})
\]
for some collection of templates $f_1 ,\ldots, f_N \in \mathcal{T}$ and some polynomial $p \in \mathbb{R}[x_1,\ldots, x_N]$.
Let $\iota : \mathcal{C} \hookrightarrow \DD$ be the inclusion and  $\iota^*: C(\DD,\mathbb{R}) \longrightarrow C(\mathcal{C}, \mathbb{R})$
 the corresponding restriction homomorphism.
Now, since $\mathcal{F}$ separates points in $\DD$ and $\mathcal{F} \subset \mathcal{A}$,
then  $\iota^*(\mathcal{A})$ is an algebra which separates points in $C(\mathcal{C}, \mathbb{R})$
and contains the nonzero constant functions.
The result follows from the Stone-Weierstrass theorem, which contends that any such algebra is dense with respect to  the sup norm.

\end{proof}

The main question now is how to go about constructing template systems in practice.
The next theorem elucidates a method for  producing countable template systems for $\DD$ by translating and re-scaling the support of any nonzero $f\in C_c(\W)$.
This shows, in particular,  that there are plenty of coordinate systems for the space of persistence diagrams, and helps explain why we refer to
nonzero elements in $C_c(\W)$  as templates.

\begin{theorem}\label{thm:templates}
Let $f\in C_c(\W)$, $n \in \N, \mm \in \mathbb{Z}^2$
and define the re-scales and translates
\begin{equation*}
% \label{MultiResolution}
f_{n,\mm} (\xx) =
%\frac{1}{2^{\|\mathbf{n}\|_\infty}}
f\left( n\xx + \frac{\mm}{n} \right).
\end{equation*}
If $f$ is nonzero, then
\begin{equation*}
\TT = \left\{f_{n,\mm} \mid n \in \N,\mm\in \mathbb{Z}^2\right\} \cap C_c(\W)
\end{equation*}
is a template system for $\DD$.
Moreover, if
$f$ is Lipschitz, then   the elements of the associated coordinate system
\[
\left\{\nu_{f_{n,\mm}}= f_{n,\mm} \circ \nu \mid f_{n,\mm} \in \TT\right\}
= \mathcal{F}_\TT
\]
are Lipschitz on any relatively  compact set
 $\mathcal{S} \subset \DD$.
That is,  the coordinate system  associated to a nonzero Lipschitz template  function is stable  on relatively compact subsets  of $\DD$.
\end{theorem}

\begin{proof}
In order to show that $\mathcal{F}_{\mathcal{T}}$ separates points in $\DD$,  let $ (S,\mu),  (T,\alpha)\in \DD$ be distinct diagrams,  and assume without loss of generality that there exists $\yy  =(y_1,y_2)\in S$  so that either:
$\yy \notin T$; or $\yy \in T$ and
$\mu(\yy) \neq \alpha(\yy)$.
Our strategy will be to find an element $f_{n,\mm} \in \TT$ so that $f_{n,\mm}(\yy) \neq 0$, and $f_{n,\mm}(\xx) = 0$ for all other $\xx \in S \cup T$.

To begin, let  $\zz = (z_1,z_2)\in \W$ be so that $f(\zz) \neq 0$.
By continuity of $f$ with respect to the Euclidean norm $\|\cdot\|$, which is equivalent to
the sup norm $\|\cdot\|_\infty$, there exists $r > 0$ so that
\begin{equation*}
B^\infty_r(\zz) :=
\left\{\xx \in \mathbb{R}^2 \mid \|\xx - \zz\|_\infty < r \right\}
\subset
\supp(f).
\end{equation*}
Moreover, since $\mathsf{supp}(f)\subset \W$ is compact,
then there exists $s > r$ so that $\supp(f) \subset B^\infty_{s}(\zz)$.
Putting this together, we have $r < s$ so that
\begin{equation*}
  B_r^\infty(\zz) \subset \supp(f) \subset B_s^\infty(\zz).
\end{equation*}

Fix $\varepsilon>0$ small enough so that $B^\infty_\varepsilon(\yy) \subset \W$ and  $B^\infty_\varepsilon(\yy)\cap (T \cup S) = \{\yy\}$.
Also, let   $n\in \N$  be large  enough so that
$n  \geq \max\left\{\tfrac{1}{r}, \tfrac{2s}{\epsilon}\right\}$.
What we will show now is that it is possible  to find
 $\mm \in \mathbb{Z}^2$ for which
$n\yy + \frac{\mm}{n} \in B^\infty_r(\zz)$,
and so that  $\xx \notin B_\epsilon^\infty(\yy)$
implies $n\xx + \frac{\mm}{n} \notin B_{s}^\infty(\zz)$.
Indeed, define
\[
L_j(t) = nt + (z_j - ny_j)
\]
for $j = 1,2$ with $\zz = (z_1,z_2)$ and $\yy = (y_1,y_2)$.
This function has the property that $L_j(y_j) = z_j$.
Further, if $|t-y_j| > \varepsilon$, then
\[
\lvert L_j(t) - z_j\rvert
\;=\; \lvert (nt + z_j) - ( ny_j + z_j) \rvert
\;=\; n\lvert t-y_j\rvert
\;>\; \frac{2s}{\varepsilon} \varepsilon
\;=\; 2s.
\]
Let $k_j \in \mathbb{Z}$ be the unique integer so that
\begin{equation*}
k_j \leq z_j -ny_j < k_j +1.
\end{equation*}
By dividing the interval $[k_j , k_j + 1)$ into $n$ subintervals of length $\frac{1}{n}$, we have that there exists a unique integer $0\leq \ell_j < n $ so that
\begin{equation*}
k_j + \frac{\ell_j}{n}
\leq
z_j - n  y_j
<
k_j + \frac{\ell_j + 1}{n}.
\end{equation*}
Let $m_j = nk_j + \ell_j$,
and $\mm = (m_1,m_2)$.
It follows
that
\[
\left\|\zz -  \left(n\yy + \frac{\mm}{n}\right)\right\|_\infty
< \frac{1}{n} \leq  r
\]
and therefore
$ f_{n,\mm}(\yy) = f\left(n\yy + \frac{\mm}{n}\right)  \neq 0$.
Moreover, if $\xx \notin B^\infty_\epsilon(\yy)$ and $j\in \{1,2\}$ is so that $\lvert x_j  - y_j \rvert \geq \epsilon$,
then
\[
\lvert nx_j - n y_j\rvert =  \lvert L_j(x_j) - z_j\rvert >
2s,
\]
and therefore
\begin{eqnarray*}
  \left\|\left(n \xx  + \frac{\mm}{n}\right) - \zz\right\|_\infty &\geq & \left\lvert  \left( nx_j + \frac{m_j}{n}\right) - z_j \right\rvert \\
  & =  & \left\lvert  nx_j - ny_j - \left(z_j  - \left(ny_j + \frac{m_j}{n}\right) \right)\right\rvert \\
  &\geq& \lvert nx_j - ny_j \rvert - \left\lvert z_j  - \left(ny_j + \frac{m_j}{n}\right)\right\rvert \\
  &>& 2s - r \\
  &\geq & s,
\end{eqnarray*}
showing that $n\xx + \frac{\mm}{n} \notin B_{s}^\infty(\zz)$,
which in turn implies $f_{n,\mm}(\xx) = 0$.

Let us see  that the support of $f_{n,\mm}$ is a bounded subset of $\W$.
To this end, let  $\xx \in \supp(f_{n,\mm})$.
Hence $n\xx + \tfrac{\mm}{n} \in \supp(f)$, and  $\| (n\xx + \tfrac{\mm}{n} ) - \zz\|_\infty < s$.
Then for $j=1,2$,
\begin{eqnarray*}
\lvert nx_j-ny_j
      &\leq& \left\lvert (nx_j - ny_j) + z_j-\left(ny_j + \frac{m_j}{n}\right) \right\rvert
              + \left\lvert z_j - \left(ny_j + \frac{m_j}{n}\right)\right\rvert\\
      &=& \left\lvert nx_j  + z_j- \frac{m_j}{n} \right\rvert
              + \left\lvert z_j - \left(ny_j + \frac{m_j}{n}\right)\right\rvert\\
      && \leq s + r < 2s,
\end{eqnarray*}
and thus
$ \lvert x_j - y_j \rvert < \frac{2s}{n} < \varepsilon$.
Therefore $\xx \in B_\varepsilon^\infty(\yy) \subset \W$, and so $f_{n,\mm} \in C_c(\W) $.

Thus far we have that $f_{n,\mm}(\yy) \neq 0 $,
and that if  $\xx \notin B_\epsilon^\infty(\yy) $
then $f_{n,\mm}(\xx) = 0$.
This observation, paired with $B_\epsilon^\infty(\yy)\cap (S \cup T) = \{\yy\}$, implies that
\[
\nu_{f_{n,\mm}} (S,\mu) = \mu(\yy)f_{n,\mm}(\yy)  \neq 0.
\]
If $\yy \notin T$ then we have that $\nu_{f_{n,\mm}}(T,\alpha)  = 0 $;
and if $\yy \in T$ then
\[\nu_{f_{n,\mm}}(T,\alpha) = \alpha(\yy)f_{n,\mm}(\yy) \neq \mu(\yy)f_{n,\mm}(\yy)
=
\nu_{f_{n,\mm}}(S,\mu),
\]
showing that $\mathcal{F}_{\mathcal{T}}$ separates points in  $\DD$.

Let us now show that if $f$ is Lipschitz and $\mathcal{S} \subset \DD$ is relatively compact,
then the elements of $\mathcal{F}_{\mathcal{T}}$ are Lipschitz on $\mathcal{S}$.
Indeed, let $D,D' \in \mathcal{S}$, and let $\delta > 0$ be so that $d_B(D,D') < \delta$.
Moreover, fix a $\delta$-matching   between $D $ and $D' $.
Recall that this means that if $\xx, \xx'$ are matched,
then $\|\xx - \xx'\|_\infty < \delta$,
and that
if for $\zz = (z_1,z_2) \in \W$ unmatched we let
\[
\bar{\zz} = \left(\frac{z_1 + z_2}{2}, \frac{z_1 + z_2}{2}\right)
\]
then $\|\zz - \bar{\zz}\|_\infty  < \delta$.

Since $\mathcal{S}$ is relatively compact, then it is uniformly off-diagonally finite (see Def. \ref{def:UODF}),
and hence there exists a uniform upper bound $\beta > 0$ for the multiplicity in $\mathsf{supp}(f)$ of any diagram in $\mathcal{S}$---see Definition \ref{defn:MultiplicityOfDgm}.
Now, if $L > 0$ is the Lipschitz constant of $f$ with respect to the sup norm $\|\cdot\|_\infty$ on $\W$,
and  $n\in \N$, $\mm \in \mathbb{Z}^2$ are so that $f_{n,\mm} \in \mathcal{T}$, then $f_{n,\mm} \in C_c(\W)$ is also Lipschitz with constant
$nL$.
Moreover, if $\zz \in \W$ is unmatched, then
$f_{n,\mathbf{m}}(\bar{\zz})= 0$ and
\begin{eqnarray*}
  \lvert \nu_{f_{n,\mm}}(D) - \nu_{f_{n,\mm}}(D')\rvert \;\;&\leq &
\sum_{\xx, \xx'  \atop \mbox{\tiny matched} } \lvert f_{n,\mm}(\xx) - f_{n,\mm}(\xx')\rvert +
\sum_{\zz \atop \mbox{\tiny unmatched}}  \lvert f_{n,\mm}(\zz)\rvert   \\[.2cm]
    &\leq &   \sum_{\xx \tiny{\mbox{ or }} \xx'  \in\, \mathsf{supp}(f)\atop \mbox{\tiny matched} } nL\delta \;\;+
\sum_{\zz \,\in \,\mathsf{supp}(f) \atop \mbox{\tiny unmatched}}  \lvert f_{n,\mm}(\zz) - f_{n,\mm}(\bar{\zz})\rvert\\[.2cm]
   &\leq & 2n\beta L  \delta.
\end{eqnarray*}
Since this inequality holds for any $\delta > d_B(D,D')$, it readily follows that
\[
\lvert \nu_{f_{n,\mm}}(D) - \nu_{f_{n,\mm}}(D')\rvert \leq 2n\beta  L  d_B(D,D'),
\]
and hence $\nu_{f_{n,\mm}}$ is Lipschitz on $\mathcal{S}$.

\end{proof}

\begin{remark} Recall that if $\mathcal{S} \subset \DD$ is relatively-compact, then there exist boxes $\mathbb{B}_k \subset \W$, $ k\in \N$,
so that $D\subset \bigcup_{k\in \N} \mathbb{B}_k$ for every $D\in \overline{\mathcal{S}} $ (see Eq. \ref{eq:defBox} and Fig. \ref{fig:BoxNotationExample}).
This implies that in order to find approximations to a continuous function $F: \overline{\mathcal{S}} \longrightarrow \mathbb{R}$, it suffices to start with a nonzero $f\in C_c(\W)$
and take only the re-scaled translates $f_{n,\mathbf{m}}$ for which \[
\mathsf{supp}(f_{n,\mathbf{m}}) \cap \bigcup_{k\in \N} \mathbb{B}_k \neq \emptyset.\]
\end{remark}

\begin{remark}
Let us say a few words on aligning the approximation methodologies outlined in
Theorems \ref{thm:approximation} and \ref{thm:templates}---where one
has access to infinitely many template functions---with practical algorithmic
implementations that are inherently finite.
A useful point of reference is the implementation  of generalized linear models of a real variable $x$, in order to approximate a continuous function
$F: I \subset \mathbb{R} \longrightarrow \mathbb{R}$  on a compact interval $I$.
The Stone-Weierstrass theorem implies that $F$ can be uniformly approximated in $I$ via polynomials, or equivalently,
via linear combinations of the monomials
$1,x,x^2,\ldots ,x^n$ for $n\in \N$ arbitrary.
In theory, one would need infinitely many monomials for
arbitrary approximations, but in reality
with finite training data, only finitely many monomials
are relevant and even advisable.
Indeed, this is the common bias-variance tradeoff where
tools like cross validation can be used.
The same view applies to learning with template functions
on persistence diagrams.
Given finite training data $D_1,\ldots, D_N \in \mathcal{S} \subset \mathcal{D}$
and $F(D_1),\ldots, F(D_N)$, then only finitely many  re-scaled
translates $f_{n,\mm}$ of $f\in C_c(\W)$ are relevant
to the problem at hand.
These can be interpreted as a user-provided hyperparameter
for the model---like the maximum degree of monomials in linear regression---or they
 can be derived from adaptive methods
as described in \cite{tymochko2019adaptive} or \cite{polanco2019adaptive}.
The main idea being that one can identify
those compact regions in $\W$ which are most relevant to the learning task at hand.
\end{remark}

\subsection{The Need for Compactness}
In Section \ref{sec:CounterExamples} we provided examples
of non-relatively-compact  sets $\mathcal{S} \subset \mathcal{D}$ which   satisfied only two out of the three conditions from
 Theorem   \ref{thm:Compactness}.
We will revisit these examples  next to see how the approximation strategy described in Theorem \ref{thm:approximation} can fail in the absence of compactness.
Specifically, we will construct continuous functions
$F: \overline{\mathcal{S}} \longrightarrow \mathbb{R}$
on the closure of $\mathcal{S} \subset \mathcal{D}$, which cannot be uniformly approximated by functions
of the form (\ref{eq:DenseSystem}).
Indeed, let $\mathcal{S} = \left\{D_n = (S_n , \mu_n) \mid n \in \N\right\}$ be:
\begin{enumerate}
  \item Not Uniformly Off-Diagonally Finite:  $S_n = \{(0,1)\}$ with $\mu_{n}(0,1) = n$. Let $f\in C_c(\W)$ be so that
      $f(0,1) =1$, and let
       $F: \mathcal{D} \longrightarrow \mathbb{R}$ be defined
      as $F(D) =  e^{\nu_f(D)}$.
      Note that  $F$ is continuous on $\mathcal{D}$,   and that $F(D_n) = e^n$ for every $n\in \N$.

      Let $f_1,\ldots, f_N \in C_c(\W)$ and let $p \in \mathbb{R}[x_1,\ldots, x_N]$ be a polynomial of degree $k\geq 0$.
      Since
      \[
      \lvert \nu_{f_j}(D_n)\rvert  = n\lvert f_j(0,1)\rvert \leq n\|f_j\|_\infty
      \]
      for every $1\leq j \leq N$ and every $n\in \N$, then
      \[
      \lim_{n \to \infty}
      \frac{\left\lvert
      p\left(\nu_{f_1}(D_n), \ldots, \nu_{f_N}(D_n) \right)
      \right\rvert}{n^{k+1}} = 0.
      \]
      If we had the uniform bound
      \[
      \left\lvert
      F(D_n) -
       p\left(\nu_{f_1}(D_n), \ldots, \nu_{f_N}(D_n) \right)
      \right\rvert < \epsilon
      \]
      for some $0 < \epsilon < \infty$ and every $n\in \N$, then
      dividing both sides of the inequality by $n^{k+1}$ and taking the limit as $n\to \infty$
      would imply that
      \[
      \lim_{n\to\infty} \frac{e^n}{n^{k+1}} = 0.
      \]
      This is a contradiction since the above limit is $\infty$, and
      therefore the restriction of $F$ to $\overline{\mathcal{S}}$ cannot
      be uniformly approximated by functions of the form (\ref{eq:DenseSystem}).
   \\
  \item Not Off-Diagonally Birth Bounded:  $S_n = \{(n, n+1)\}$ with $\mu_{n}(n,n+1) = 1$.
      The first thing to note is that $d_B(D_n , D_m) = \frac{1}{2}$ for every
      $n\neq m \in \N$, and thus
      $\mathcal{S}$ has no accumulation
      points.
      This implies that  $\overline{\mathcal{S}} = \mathcal{S}$, and that
      any real-valued function on $\overline{\mathcal{S}}$ is continuous. Let
      $F: \overline{\mathcal{S}}\longrightarrow \mathbb{R}$ be defined as
      \[
      F(D_n) = n \;\;\; ,\;\; n\in \N.
      \]
      If $f_1,\ldots, f_N \in C_c(\W)$ and $p \in \mathbb{R}[x_1,\ldots, x_N]$, then
      \[
      \nu_{f_1}(D_n) = \cdots = \nu_{f_N}(D_n) = 0
      \]
      for all $n$ large enough, and thus there exists $N_0\in \N$ so that
      \[
      p\left(\nu_{f_1}(D_n), \ldots,  \nu_{f_N}(D_n)\right) = p(0,\ldots, 0)
      \]
     for every $n \geq N_0$.
      It follows that
      \[
      \lim_{n\to \infty}
      \left\lvert
      F(D_n)
      -
      p\left(\nu_{f_1}(D_n), \ldots,  \nu_{f_N}(D_n)\right)
      \right\rvert
      =\infty
      \]
      and thus $F$ cannot be uniformly approximated in $\overline{\mathcal{S}}$
      by functions of the form (\ref{eq:DenseSystem}).
       \\
  \item Not Bounded:   $S_n = \{(0,n)\}$ with $\mu_{n}(0,n)= 1$. Let
  $F : \mathcal{D} \longrightarrow \mathbb{R}$ be
      \[
      F(D) = 2d_B(D, \varnothing).
      \]
      It follows that $F$ is continuous and that  $F(D_n) = n$ for every $n\in \N$.
      The argument now proceeds exactly as in (2) above.
\end{enumerate}

\subsection{Compact Approximations with the Wasserstein Distance}
Thus far we have established that  $\nu : \mathcal{D} \longrightarrow C_c(\W)'$
is a linearization of  $(\mathcal{D}, d_B)$ (see Theorem \ref{thm:Linearization}),
and that the rescales/translates of any nonzero $f\in C_c(\W)$ can be
used to construct compact-open dense subsets of $C(\mathcal{D}, \mathbb{R})$ (see Theorems \ref{thm:templates} and \ref{thm:approximation}).
A natural question is whether similar results hold for $(\mathcal{D}_p, d_{W_p})$---i.e., for the $p$-Wasserstein distance, $p \geq 1$.
We will see next that this is indeed the case.

Recall that $\mathcal{D}_p \subset \mathcal{D}$ for every $p \in \N$, and that
this inclusion is (uniformly) continuous. It follows that,

\begin{corollary} $\nu$ restricts to a $d_{W_p}$-continuous function
$\nu: \mathcal{D}_p  \longrightarrow C_c(\W)'$ with the same properties
as in Theorem \ref{thm:Linearization}
\end{corollary}

This implies that $\nu_f : \mathcal{D}_p \longrightarrow \mathbb{R}$ is continuous
for every $f\in C_c(\W)$, and thus any template system $\mathcal{T} \subset C_c(\W)$ for $\mathcal{D}$ is also a template system for $\mathcal{D}_p$.
Using the first half of Theorem \ref{thm:templates}, and following the same proof as in Theorem \ref{thm:approximation}, we have that

\begin{corollary} Let $f\in C_c(\W)$ be nonzero, and let
\[
\TT = \left\{f_{n,\mm} \mid n \in \N,\mm\in \mathbb{Z}^2\right\} \cap C_c(\W).
\]
Then the set of functions of the form
\[
D \mapsto P\big(\nu_{f_1}(D), \ldots, \nu_{f_N}(D)\big)
\]
for $N\in \N$, $P\in \mathbb{R}[x_1,\ldots, x_N]$ and $f_1,\ldots, f_N \in \mathcal{T}$,
is compact-open dense in $C(\mathcal{D}_p, \mathbb{R})$.
\end{corollary}

\section{Example template functions}
\label{sec:TemplateFunctionExamples}

At this point in the story, we shift our view from theory to practice, as the mathematical framework built to this point leaves open the choice of template system.
In our experiments, we use two collections of functions, but we have no reason to suspect that these are the only or even the best available options.
The first, which we call tent functions, are described in Sec.~\ref{ssec:Tents}.
The second are interpolating polynomials, traditionally used for approximating functions, which are described in Sec.~\ref{ssec:InterpPoly2}.

For the entirety of this section, we will define functions on the birth-lifetime plane as this simplifies notation substantially.
We use the tilde to denote the portions that are defined in this plane to emphasize the change from the birth-lifetime plane.
So, let $\widetilde \W= \{(x,y) \mid x \geq 0, y > 0\}$; that is, the conversion of $\W$ to the birth-lifetime plane.
Likewise, let $\widetilde{\W^\varepsilon} = \{(x,y) \in \widetilde \W \mid y > \varepsilon\}$ so that it is the conversion of $\W^\varepsilon$ to the birth-lifetime plane.
Given   $\xx = (a,b) \in \W$, we write $\tilde {\xx }= (a,b-a) \in \widetilde \W$ for the converted point.
Given a diagram $D = (S,\mu)$, we write $\widetilde D = (\widetilde S,\widetilde \mu)$ where $\widetilde S = \{ \tilde \xx \mid \xx \in S\}$ and $\tilde \mu(\widetilde \xx) = \mu(\xx)$.

%-------------------------------
\subsection{Tent functions }
\label{ssec:Tents}
%-------------------------------

\begin{figure}[tb]
	\centering
	\includegraphics[width = \textwidth]{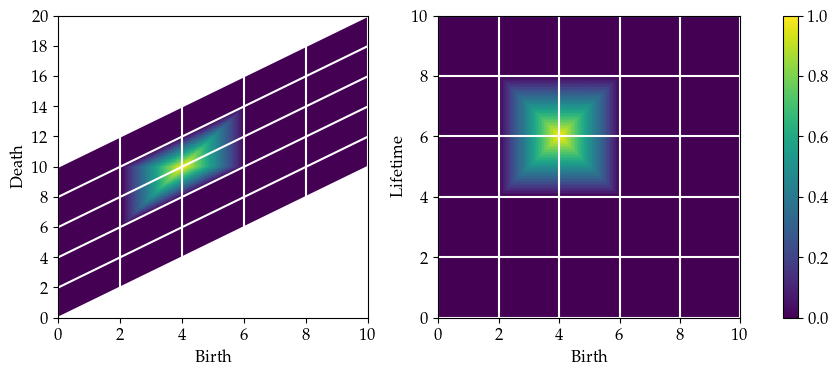}
	\caption{Tent function $g_{(4,6),2}$ drawn in the birth-death plane (left) and in the birth-lifetime plane (right) for the grid determined by $d=5$, $\delta = 2$, and $\varepsilon = 0$. Note that the template system for tent functions require $\varepsilon > 0$, which would shift this function up by $\varepsilon$. }
	\label{fig:TentFuncExample}
\end{figure}

We first define a template system in the birth-lifetime plane which we call tent functions.
Given $\aa = (a,b) \in \widetilde \W$ and a radius   $0 <\delta < b$,  define the tent function on $\widetilde \W$ to be
\begin{equation*}
	g_{\aa, \delta}(x, y) = \left\lvert 1-\frac{1}{\delta} \max\{\lvert x-a\rvert, \lvert y- b\rvert\}\right\rvert_+
\end{equation*}
where $\lvert r\rvert_+ = \max\{r ,0 \}$ for $r\in \mathbb{R}$.
As $\delta < b$, this function has support in the compact box $[a-\delta, a+\delta] \times [ b- \delta, b+\delta] \subset \widetilde \W$.

Given a persistence diagram $D = (S,\mu)$, the tent function is defined to be the sum over the evaluation on the points in the diagram, namely
\begin{equation*}
G_{\aa,\delta}(D) = \widetilde G_{\aa,\delta}(\tilde D) = \sum_{\tilde \xx \in \tilde S} \widetilde\mu(\xx) \cdot g_{\aa, \delta}(\tilde \xx).
\end{equation*}
We use $G$ or $\widetilde G$ depending on whether we want our input to be a diagram in the birth-death or birth-lifetime plane, respectively, but all subscript notation is written in the birth-lifetime plane for ease of notation.

We then have the following theorem.
\begin{theorem}
The collection of tent functions
\begin{equation*}
	\left \{ G_{\aa,\delta} \mid \aa = (a,b) \in \widetilde \W, 0 <\delta < b\right \}
\end{equation*}
separates points in $\DD$.
\end{theorem}
\begin{proof}
Let
$D_1 = (S_1,\mu_1)$ and $D_2 = (S_2,\mu_2) \in \DD$ be distinct.
WLOG there is an $\xx  \in S_1$ for which either (i) $\xx \not \in S_2$ or (ii) $\xx \in S_2$ but $\mu_1(\xx) > \mu_2(\xx)$.
For ease of notation, assume in case (ii) that $\xx \in S_2$ and $\mu_2(\xx) = 0$.
Then we always have $\xx \in S_2$ and  $\mu_1(\xx) > \mu_2(\xx)$.

Let $\widetilde \xx = (a,b)$.
For any $\delta$, define $\widetilde{B_\delta} =[a-\delta, a+\delta] \times [b - \delta, b + \delta]$ and note that this is the support of $g_{\xx,\delta}$.
As $D_1$ and $D_2$ are in $\DD$, both diagrams have finite multiplicity in $\widetilde{\W^{b/2}}$.
So, there exists a $\delta <b/2$ so that  $\widetilde {S_1} \cap \widetilde{B_\delta} = \{\xx\} = \widetilde {S_2} \cap B$.
As $\xx$ is the only point in either diagram in the support of $g_{\xx,\delta}$,
\begin{equation*}
	G_{\xx,\delta}(D_1) = \mu_1(\xx) > \mu_2(\xx) = G_{\aa,\delta}(D_2).
\end{equation*}
Thus, the collection of tent functions separates points.

\end{proof}

For practical purposes, we pick a subset of these tent functions.
Let $\delta > 0$ be the partition scale, let $d$ the number of subdivisions along the diagonal (resp. $y$ axis), and let $\varepsilon>0$ be the upward shift.
In our experiments described in Sec.~\ref{sec:Experiments}, we use the collection of tent functions given by
\begin{equation}
  \label{eq:TentFunctions}
  \left\{  G_{(\delta i,\delta j + \varepsilon), \delta } \mid 0 \leq i \leq d, 1 \leq j \leq d \right\}.
\end{equation}
That is, these are the tent functions centered at a regular grid shifted up by $\varepsilon$ to ensure that $ g$ is supported on a compact set in $\widetilde \W$.
See Fig.~\ref{fig:TentFuncExample} for an example.

%-------------------------------
\subsection{Interpolating  polynomials}
\label{ssec:InterpPoly2}
%-------------------------------
Say we are given a nonempty, finite set of distinct mesh values $\AA = \{a_i\}_{i=0}^m \subset \mathbb{R}$ and a collection of evaluation values $\{c_i \in \mathbb{R}\}$, the first goal is to build a polynomial such that $f(a_i) = c_i$ for all $i$.
The Lagrange polynomial $\ell_j^\AA(x)$ corresponding to node $a_j$ is defined as
\begin{equation}
\ell^\AA_j(x) = \prod\limits_{i\neq j}\frac{x-a_i}{a_j-a_i}.
\end{equation}
Note that this function satisfies
%--------------------------------
\begin{equation*}
	\ell^\AA_j(a_k) =	\begin{cases}
						1, & j= k, \\
						0, & {\rm otherwise},
					\end{cases}
          \qquad \text{ and } \qquad
  \sum\limits_{j=0}^{m}{\ell^\AA_j(x)} = 1.
\end{equation*}
The Lagrange interpolation polynomial is then simply $f(x) = \sum_{j=0}^m c_j\ell^\AA_j(x) $.
Note that for numerical stability, one must work with the barycentric form of Lagrange interpolation formula described by \cite{Berrut2004} and shown in Appendix \ref{sec:interpAlg}.

We will now  use these polynomials to create functions on $\widetilde\W$.
Assume we have two collections of mesh points
$\AA = \{a_i\}_{i=0}^m \subset \mathbb{R}$ and
$\BB = \{b_i\}_{i=0}^n \subset \mathbb{R}_{>0}$
so that $(a_i,b_j) \in \widetilde \W$ for all $i,j$.
Then, given a collection of evaluation points $\CC = \{c_{i,j}\}_{i,j} \subset \mathbb{R}$, we want to build a polynomial for which $f(a_i,b_j) =c_{i,j}$.
Note that in general the evaluation points $\CC$ are not part of the persistence diagrams to be evaluated; nevertheless, these values are not needed in our construction but we do keep track of their coefficients.
We define the 2D interpolating polynomial for the collection $\AA, \BB, \CC$ to be
%--------------------------------
\begin{equation}
  \label{eq:interpolatingPolynomial}
f(x, y) = \sum_{i=0}^{m}\sum_{j=0}^{n} {c_{i,j} \cdot g\left(\ell^\AA_i(x) \cdot \ell^\BB_j(y)\right)}.
\end{equation}
%--------------------------------
where $g(\cdot)$ is either the identity function or $g(\cdot) = \lvert \cdot\rvert$; in our experiments we used the latter for the simple reason that it seemed to give better results.
We now evaluate $f$ at each of the $N$ query points which are the points of a persistence diagram in $\widetilde \W$ to obtain $N$ equations that we can write as
%--------------------------------
\begin{equation}
\label{eq:featurization}
\mathbf{f} = \mathbf{L} \, \mathbf{c},
\end{equation}
%--------------------------------
where $\mathbf{L}$ is an $((m+1) \times (n+1))\times N$ matrix, and $\mathbf{f}$ is an $(m+1) \times (n+1)$ vector obtained by concatenating a 2D mesh, similar to the one shown in Fig.~\ref{fig:2dMesh}, row-wise.
Note that this representation implies an ordering of the points of the persistence diagram, but after the next step, the order will not matter.

Each column of matrix $\mathbf{L}$ in Eq.~\eqref{eq:featurization} represents a vector that describes the contributions of all the query points to the corresponding entry in $\mathbf{f}$.
Renumbering the entries in vector $\mathbf{f}$ according to $(i,j)\mapsto i(n+1) + j + 1=r$ where $r \in \{0, \ldots, (m+1)(n+1) \}$, we can now assign a score for each point in the mesh using the map $S_r: \mathbb{R}^N \to \mathbb{R}$, i.e., by operating on the rows of matrix $\mathbf{L}$ according to
\begin{equation}
S_r = \sum\limits_{j=0}^{N-1}{L_{j,r}}.
\end{equation}
Choosing a larger base mesh implies using a higher degree polynomial in the interpolation.
Therefore, the role of increasing the degree of the polynomial is similar to the role of increasing the number of tent functions.
A larger mesh leads to more features which gives a tool for either increasing or reducing the number of features.
The former improves the fit to the training set, while the latter reduces the number of features which allows mitigating overfitting effects.
While any class of interpolating polynomials can be used, in this study we chose Chebyshev interpolating polynomials due to their excellent approximation properties, see \cite{Trefethen2012}.
Appendix \ref{sec:interpAlg} describes how to use the interpolation matrices separately obtained for each of the birth times and lifetimes of a given persistence diagram to construct $L_{j,r}$.

\begin{figure}
  \centering
  \includegraphics[width = .5\textwidth]{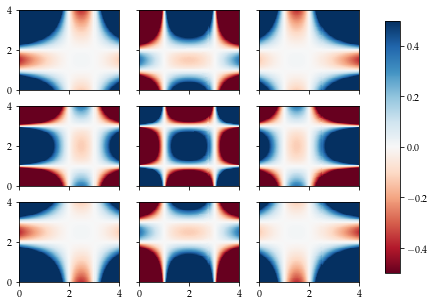}
  \caption{An example of interpolating polynomials for the mesh $\AA = \BB = \{1,2,3\}$ for evaluation values $\CC_{i,j} = \{ c_{i,j} = 1, c = 0 \text{ else}\}$.
	Viewing the grid of images, $f^{\AA,\BB}_{i,j}$  is drawn at location $(i,j)$ in the figure, with $(1,1)$ at the bottom left.}
  \label{fig:InterpPolyEx}
\end{figure}

Fix a compact region $K \subset \widetilde \W$, an $\varepsilon > 0$ such that
\[
\overline{K^\varepsilon} = \{\xx \in \mathbb{R}^2 \mid  d_K(\xx) \leq \varepsilon\} \subset \widetilde{\W}
\;\;\;\;\;\;\;\; \mbox{where } \;\;\;\;\;\;\;\;
d_K(\xx) = \inf_{\yy \in K} \|\xx - \yy\|
\]
and  a collection of mesh points $\{(a_i,b_j)\}\subset K$ given by $\AA$ and $\BB$ as above.
Define $h_{K,\varepsilon}$ to be a continuous function on $\widetilde \W$ such that
\begin{equation*}
	h_{K,\varepsilon}(\xx) =
	\begin{cases}
	1 & \xx \in K\\
	0 & d_K(\xx) \geq \varepsilon
	\end{cases}
\end{equation*}
For instance, one can let $h_{K,\varepsilon}(\xx) = \max \left\{ 0 \; , \; 1 - \frac{1}{\epsilon}d_K(\xx) \right\} $.
Note that the support of this function is contained in the compact set $\overline{K^\varepsilon}\subset \widetilde \W$.

Let $\CC_{i,j}$ be the collection of evaluation values which are entirely 0 except for $c_{i,j} = 1$.
Define $f^{\AA,\BB}_{i,j} = f_{i,j}$ to be the interpolating polynomial (Eq.~\ref{eq:interpolatingPolynomial}) for this setup.
Then the function on diagrams is defined to be
\begin{equation*}
	F^{\AA,\BB,K,\varepsilon}_{i,j}(D)  =
\widetilde F^{\AA,\BB,K,\varepsilon}_{i,j}(\widetilde D)
:=
\sum_{\tilde \xx \in \tilde S}
	\mu(\tilde \xx) \cdot
	f^{\AA,\BB}_{i,j}(\tilde \xx) \cdot
	h_{K,\varepsilon}(\tilde \xx).
\end{equation*}
See Fig.~\ref{fig:InterpPolyEx} for examples of these functions for different parameter choices.
We have the following theorem to show that these interpolating polynomials can be used as template functions.

\begin{theorem}
The collection of interpolating polynomials
\begin{equation*}
	\left \{F^{\AA,\BB,K,\varepsilon}_{i,j} \right\}
\end{equation*}
separates points, where the collection varies over all choices of compact $K \subset \widetilde \W$, $\delta \in \mathbb{R}_{>0}$, and of mesh $\AA$, $\BB$ as specified above.
\end{theorem}

\begin{proof}
We are given two persistence diagrams
$D_1 = (S_1,\mu_1)$ and $D_2 = (S_2,\mu_2) \in \DD$, with $D_1 \neq D_2$.
WLOG there is an $\xx =(a,b) \in S_1$ for which either (i) $\xx \not \in S_2$ or (ii) $\xx \in S_2$ but $\mu_1(\xx) > \mu_2(\xx)$.
To avoid case checking, we assume as before  that $\xx \in S_2$ with $\mu_2(\xx) = 0$ in the later case.
This way, in both cases have $\xx \in S_2$ and  $\mu_1(\xx) > \mu_2(\xx)$.

Choose a compact set $K\ni \xx$ and $\varepsilon$ both small enough so that $\overline{K^\varepsilon} \cap S_1 = \xx$ and   $\overline{K^\varepsilon} \cap S_2 = \xx$.
If $\widetilde \xx = (a,b)$, set $\AA = \{a \}$ and $\BB = \{b\}$.
Note that in this overly simplistic setup, $\ell^\AA(x) = x/a$ and $\ell^\BB(y) = y/b$, so the only interpolating polynomial is $f(x,y) = g((xy)/(ab))$.
Whether $g$ is the identity or the absolute value function, $f$ evaluates to 1 at $\widetilde \xx$.
Because the only point in either diagram inside $\overline{K^\varepsilon}$ is $\xx$, $h_{K,\varepsilon}(\tilde \xx) = 1$ and $h_{K,\varepsilon}(\tilde \yy) = 0$ for every other $\yy \in S_1 \cup S_2$, so
\begin{equation*}
	F(D_i) =
	\sum_{\tilde \xx \in \tilde S}
	\mu(\tilde \xx) \cdot
	f(\tilde \xx) \cdot
	h_{K,\varepsilon}(\tilde \xx)
	=   \mu_i(\xx)
\end{equation*}
 for $i = 1,2$.
 Thus $F$ separates the two diagrams.

\end{proof}

In our experiments, we set $K$ to be a box $[A,A'] \times [B,B']$ with $B>0$, $\varepsilon$ to be either machine precision or $B/2$, and use the non-uniform Chebyshev mesh as seen in Fig.~\ref{fig:2dMesh} (\cite{Trefethen2012}).
 % \LM{Definition of Chebyshev points?}
Our naieve implementation of this featurzation was quite slow; see Sec.~\ref{sec:interpAlg} for an explanation of the vectorization used to implement and speed up the code.

%--------------------------------
\begin{figure}
\centering
\includegraphics[width=0.5\textwidth]{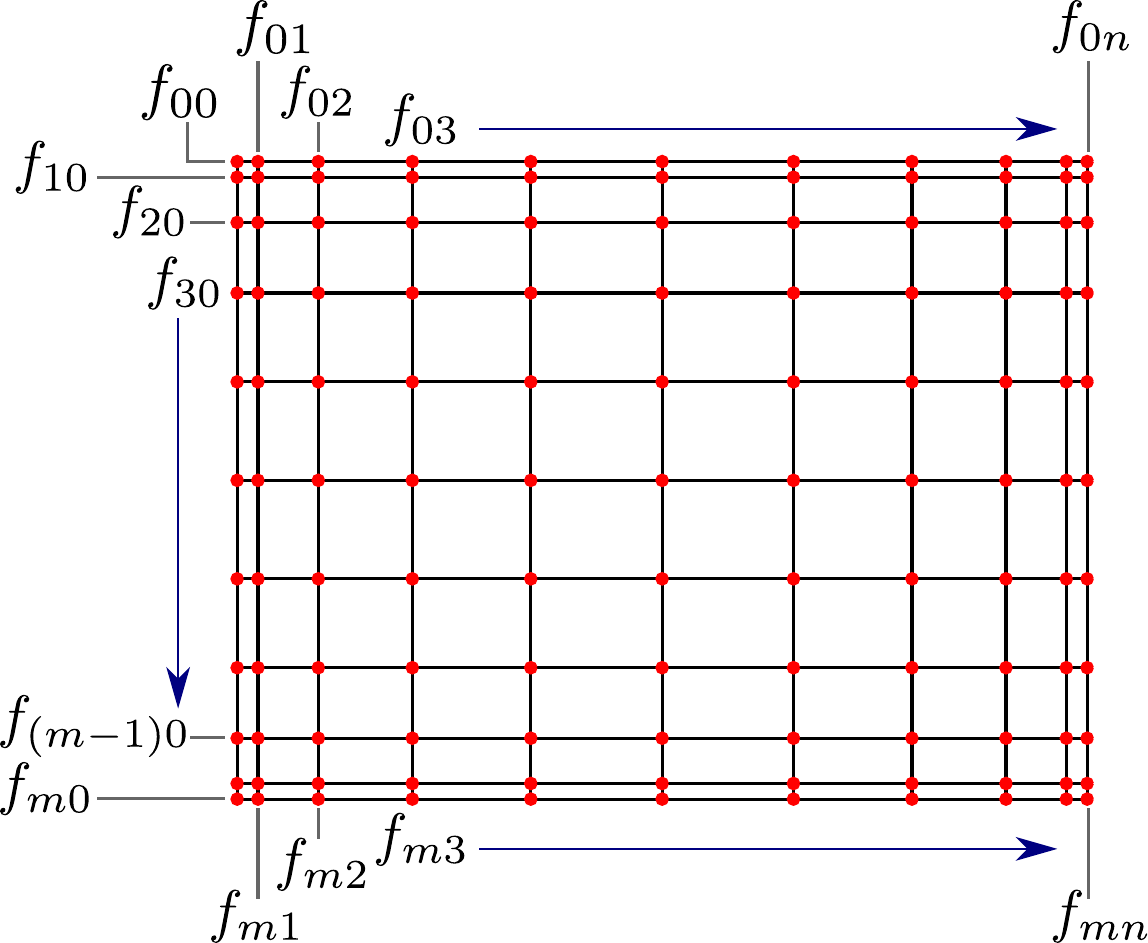}
\caption{An example $11 \times 11$ 2D grid with $m=n=10$ defined using the Chebyshev points of the second kind. }
\label{fig:2dMesh}
\end{figure}
%--------------------------------

% Content for the Regularized Regression section
% Search for \input{sec-Regression} in the main document

\section{Implementing regularized regression/classification with templates}
\label{sec:RidgeRegression}
The results thus far imply that template systems on $\mathcal{D}$ can be used
to featurize persistence diagrams;
these vectorizations, in turn, can be used as inputs to machine learning algorithms
for classification and regression tasks.
We describe next one avenue for implementing these ideas  in practice.
Indeed, given a finite collection of labeled persistence diagrams
\[\{ (D_m, \ell_m)\}_{m = 1}^M \subset \mathcal{D} \times L\]
with $L\subset \mathbb{R}$,
and a template system $\mathcal{T} \subset C_c(\W)$,
the goal is to find $N\in \N$,  template
functions $f_1,\ldots, f_N \in \mathcal{T}$,
and a polynomial $p\in \mathbb{R}[x_1,\ldots, x_N]$,
such that the function
\[
\begin{array}{rccl}
  P : & \mathcal{D} & \longrightarrow & \mathbb{R} \\
   &  D & \mapsto & p\big(\nu_{f_1}(D), \ldots, \nu_{f_N}(D)\big)
\end{array}
\]
satisfies $P(D_m) \approx \ell_m$ for $m=1,\ldots, M$.
It follows from Thm.~\ref{thm:approximation} that this process results in arbitrarily accurate approximations on compact subsets of  $\mathcal{D}$, provided the labels $\ell_m$ vary continuously.
In practice the template functions $f_1,\ldots, f_N$ can be either provided by the user from
a specific class, like tents or interpolating polynomials (sec. \ref{sec:TemplateFunctionExamples}),
or can be derived from adaptive and data-driven strategies  as described in related work of the authors;
please refer to \cite{tymochko2019adaptive} or \cite{polanco2019adaptive}
% \footnote{Code available at \url{https://github.com/lucho8908/adaptive_template_systems} }
for several adaptive strategies and comparisons.
Given    template functions $f_1,\ldots, f_N \in C_c(\W)$, the optimal polynomial $p\in \mathbb{R}[x_1,\ldots, x_N]$ is uniquely determined by its vector of coefficients, $\aa \in \mathbb{R}^k$.
 We will make this explicit with the notations $p_\aa$ and $P_\aa$, and an optimization will be set up in order to determine $\aa\in \mathbb{R}^k$ from the available labeled data.

The error of fit $P_\aa(D_m) \approx \ell_m$ is measured in the usual way via  a loss function
\[
\mathcal{E} : \mathbb{R} \times L \longrightarrow \mathbb{R}
\]
where common choices include:
\begin{description}
  \item[Square] is given by $\mathcal{E}_{sq}(t,\ell) = (t - \ell)^2$, $L = \mathbb{R}$,
      and yields a least-squares regression.
      Can handle multi-class classification.
  \item[Hinge] is given by $\mathcal{E}_{hg}(t,\ell) = \max\{0,1 - \ell\cdot t\}$, with $L = \{-1,1\}$,  and appears in the soft-margin classifier of support vector machine.
  \item[Logistic] is given by the log-loss
  $\mathcal{E}_{log}(t,\ell)= \ln\left(1 + e^{-\ell\cdot t}\right)$, with $L = \{-1,1\}$,  and yields logistic regression.
\end{description}
Meanwhile, the complexity of the model can be measured, for instance, via a regularization function
\begin{equation*}
\Omega : \mathbb{R}^k \longrightarrow [0,\infty).
\end{equation*}
The regularized optimization scheme looking to minimize the regularized mean loss is
\begin{equation*}
  % \label{eq:Minimization1} %From liz: I commented this out because we never referenced the equation.
\aa = \argmin_{\vv \in \mathbb{R}^k}
\frac{1}{M}
\sum_{m=1}^M
\mathcal{E}
\left(
P_\vv\big(D_m\big), \ell_m
\right) + \alpha \Omega(\vv)
\end{equation*}
where $\alpha > 0$ is the regularization parameter, often chosen from the set $\{10^n\}_{n \in \mathbb{Z}}$.

\paragraph{Visualization of Coefficients.}
Our collections of template functions have a uniquely 2d geometric flavor.
In particular, for both tent functions and Lagrange polynomials on our formulation, we have a function for each $a_i,b_j$ location on a mesh.
This means that we can pull the coefficients $\vv$ determined in the optimization back to the grid which built them for visualization by drawing a heat map with $v_{i,j}$ drawn at $(i,j)$ to more fully understand the model.
Examples of this are shown in Figs.~\ref{fig:RegressionCoefficients} and \ref{fig:ManifoldCoefficients}.
We can use these heat maps can be used to help  localize the important features in the learning task.
However, since we did not put a sparsity penalization term (e.g.,  $L_1$ regularization or lasso) in the regression problem, the highlighted pixels may be more than what is needed to localize the problem.

% Content for the Experiments section
% Search for \input{sec-Experiments} in the main document

\section{Experiments}
\label{sec:Experiments}
\subsection{Code}

Code for doing classification and regression using tent and interpolating polynomial functions is available in the python \texttt{teaspoon} package\footnote{\url{https://github.com/lizliz/teaspoon}}.
Classification and regression were done using ridge regression, where the loss and regularization functions are both square.
Computation is done using \texttt{RidgeClassifierCV} and \texttt{RidgeCV} functions from the  \texttt{sklearn} package.
These functions are the counterparts of \texttt{Ridge} and \texttt{RidgeCV} functions with built-in cross-validation of the regularization parameter, $\alpha$.
Unless otherwise noted, we have used the default in sklearn, which chooses  $\alpha$ from the set $\{ 10^n \mid n \in \{-1, 0, 1 \}\}$.

The main hyperparameters to be chosen prior to running experiments are those controlling the number of template functions used.
The number of tent functions is controlled by $d$ in Eq.~\ref{eq:TentFunctions}; the default is $d=10$ unless otherwise noted.
We then choose $\delta$ and $\varepsilon$ which control the support of the functions.
If they are not specified in advance, we use the following procedure.
First, determine a bounding box is for the points in all diagrams in the test set.
This is done by setting $\varepsilon$ to be half the minimium persistence of all points in the diagrams, and then setting $\delta$  to be the smallest value so that the box $[0,\delta d] \times [0 \delta d + \varepsilon]$ contains all points in the birth-lifetime plane.
Then in the notation of Eq.~\ref{eq:TentFunctions}, the support of all tent functions is  contained in the box $[-\delta, (d+1) \delta] \times [\varepsilon, (d+1) \delta + \varepsilon]$.

For interpolating polynomials, the number and the type of the mesh points must be specified.
We used the roots of the Chebyshev polynomials of the first kind.
The number of the used functions is controlled by the length of the meshes $\mathcal{A}$ and $\mathcal{B}$ in Eq.~\ref{eq:interpolatingPolynomial}, i.e., the number of the chosen Chebyshev points in each direction.
The default length is $m=n = 10$ unless otherwise noted.

All experiments were run using \texttt{seed = 48824} and 33\% of the data reserved for testing.
Scores for classification experiments are reported using the percent that were correctly classified.
Scores for regression experiments are reported using the coefficient of determination, $R^2$.
Note that this latter score can potentially take negative values; perfect regression would score 1, and a method which returns the constant prediction of the expected value is given a score of 0.

\subsection{Off-diagonal, normally distributed points}
\label{ssec:NormalDistPoints}

We generate diagrams from the following procedure.
Given $\mu$ and $\sigma$, draw $n$ points from the gaussian $N(\mu,\sigma)$ on $\mathbb{R}^2$.
Retain all points which are are in $\W$.
For our simulations, we fixed $\sigma = 1$ and varied $\mu$.
Examples of two overlaid example diagrams are shown in Fig.~\ref{fig:RandomDgms}.

\begin{figure}[!htb]
	\centering
	\includegraphics[width = .49\textwidth]{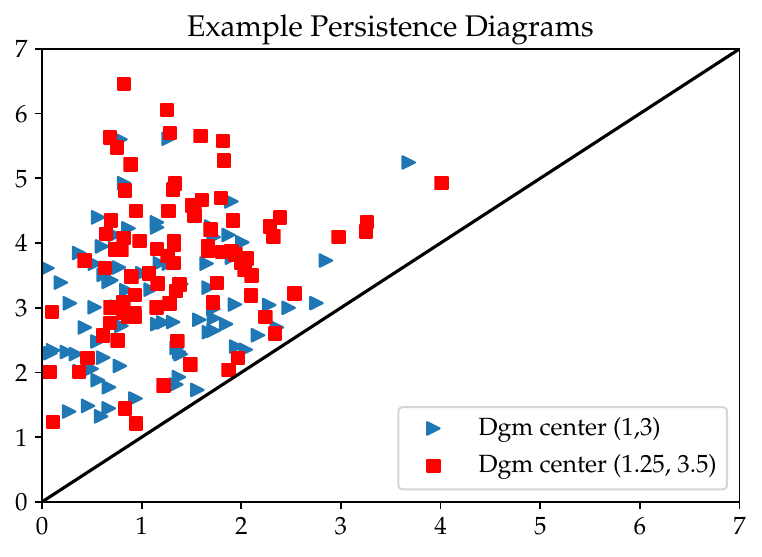}
	\caption{
	Example diagrams randomly generated by the procedure described in Sec.~\protect\ref{ssec:NormalDistPoints}.
	Two diagrams drawn from a distributions with different choices of $\mu$ are shown.
	 }
	\label{fig:RandomDgms}
\end{figure}

\textit{Classification.}
We tested our classification accuracy with the following experiment.
We chose two collections, $A$ and $B$, of 500 persistence diagrams each generated by drawing $n=20$ points (note that this means there are at most $n$ points).
The means $\mu$ were different: $\mu_A = (1,3)$ and $\mu_B$ was varied along the line $(1,3)t + (2,5)(1-t)$ for $t \in [0,1]$.
% See Table \ref{tab:OffDiagNormal_Classification} for further chosen parameters for the study.
Classification accuracy using tent and interpolating polynomial functions is shown in Fig.~\ref{fig:NormalPts_Classification}.
As expected, the correct classification percentage for the test set is around 50\% when $\mu_A \sim \mu_B$, and improves as they move farther apart.
In particular, by the time the means are at distance apart equal to the standard deviation used for the normal distribution ($\sigma = 1$), classification is well above 90\%.
For this particular experiment, we do not see any difference between the choice of template function used.

\begin{figure}[tb]
	\centering
	\includegraphics[width = .7\textwidth]{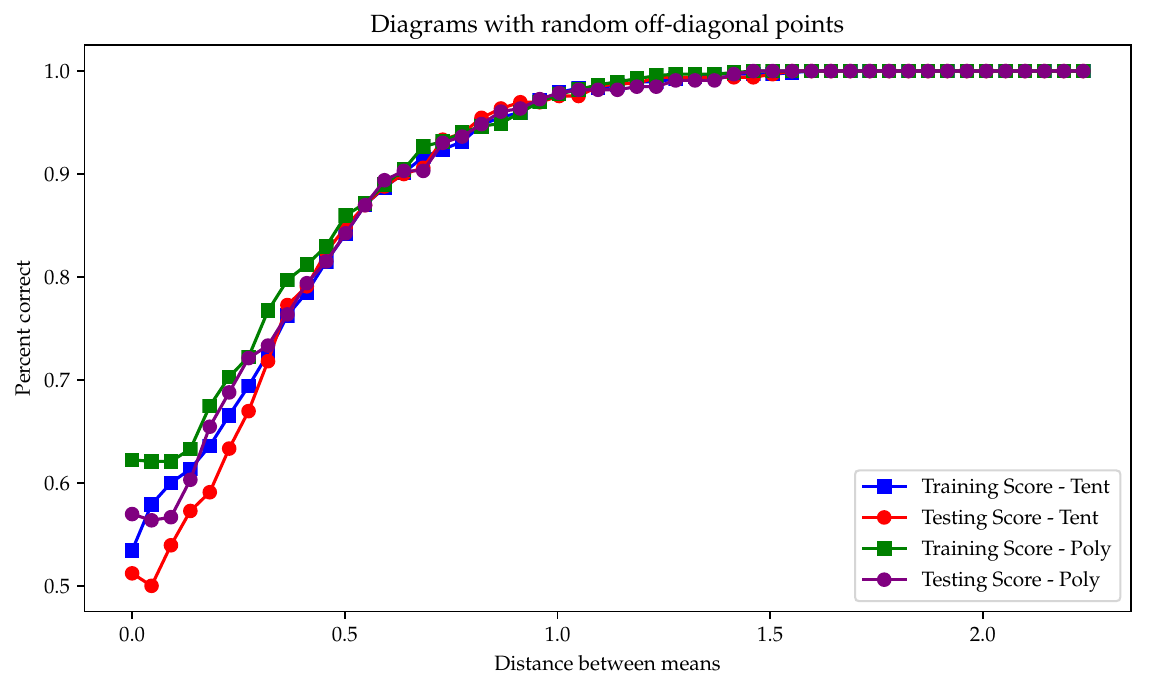}
	\caption{Results from  classification test for pairs of choices of $\mu$ with $\mu_{A} = (1,3)$ and $\mu_{B}$ chosen on the line from $(1,3)$ to $(2,5)$.}
	\label{fig:NormalPts_Classification}
\end{figure}

\textit{Regression.}
We further ran two regression versions of the experiment as follows.
In the first test, we drew 500 diagrams from the above procedure with a  choice of center $\mu$ drawn uniformly on the line segment $t(1,3) + (1-t)(6,8)$, $t \in [0,1]$.
Then, we predicted the distance from $\mu$ to $(1,3)$.
We call this the ``line'' experiment.
Second, we drew $\mu$ from the normal distribution $N((1,3),1)$ and again predicted the distance of $\mu$ from the point $(1,3)$.
We call this the ``ball'' experiment.

Each of these experiments was run 10 times, and the results can be seen in Table \ref{tab:RegressionResults}.
Example predictions for single runs can be seen in Fig.~\ref{fig:RegressionExperiment}, and the coefficients for these examples are in Fig.~\ref{fig:RegressionCoefficients}.
Note that the coefficients are drawn at the location of their index.
In particular, the interpolating polynomials are determined using a non-uniform mesh, so the heatmap for these coefficients does not align with the location of the associated point.

\begin{table}[tb]

  \centering
	\begin{tabular}{r||cc|cc}
		& Tents - Train & Tents - Test & Polynomials - Train & Polynomials - Test\\ \hline
	Line & $0.977 \pm 0.002$ &	$0.970  \pm 0.004$ & $0.979 \pm 0.001$ & 	$0.970 \pm  0.002$\\
	Ball & $0.823 \pm  0.023$ & $0.782 \pm  0.023$ &$0.832 \pm 0.021$ & $0.786 \pm 0.021$
	\end{tabular}
	\caption{The $R^2$ results of the regression tests described in Sec.~\protect\ref{ssec:NormalDistPoints}.}
	\label{tab:RegressionResults}

\end{table}

\begin{figure}[tb]
	\centering
	\includegraphics[width = .39\textwidth]{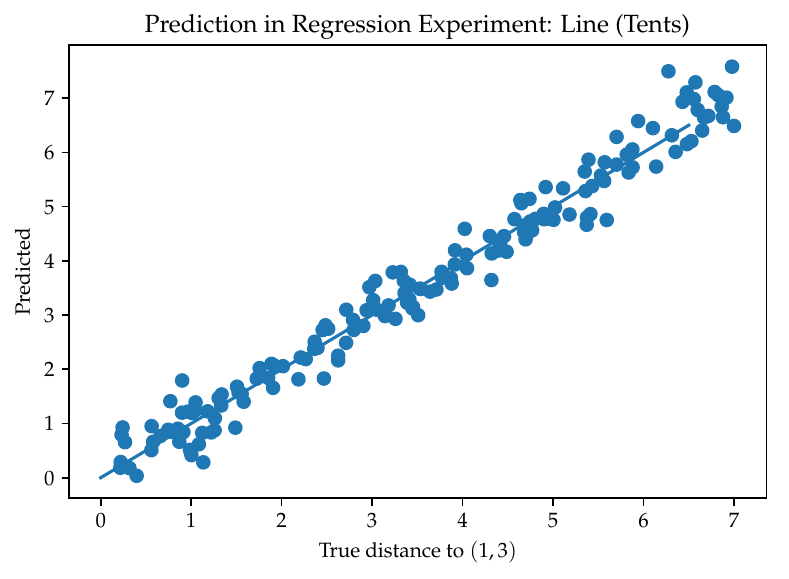}
	\includegraphics[width = .39\textwidth]{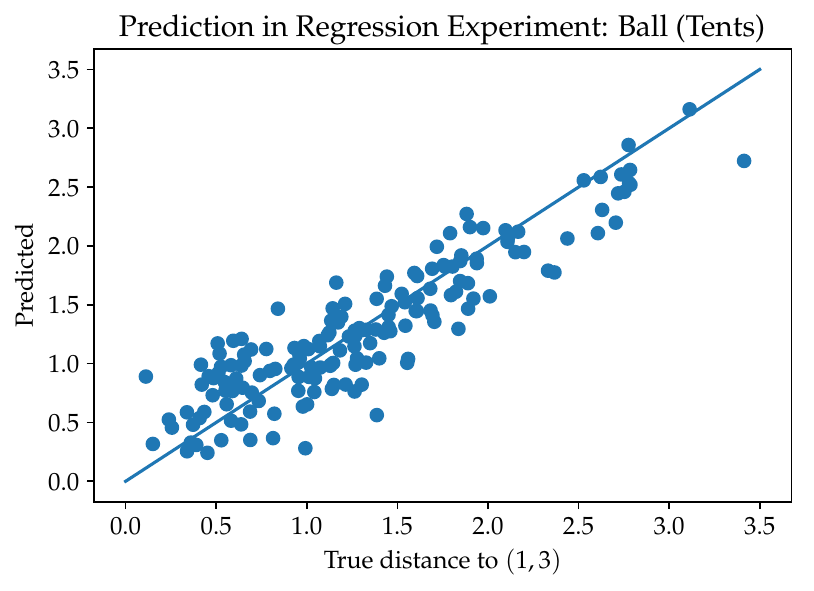}

	\includegraphics[width = .39\textwidth]{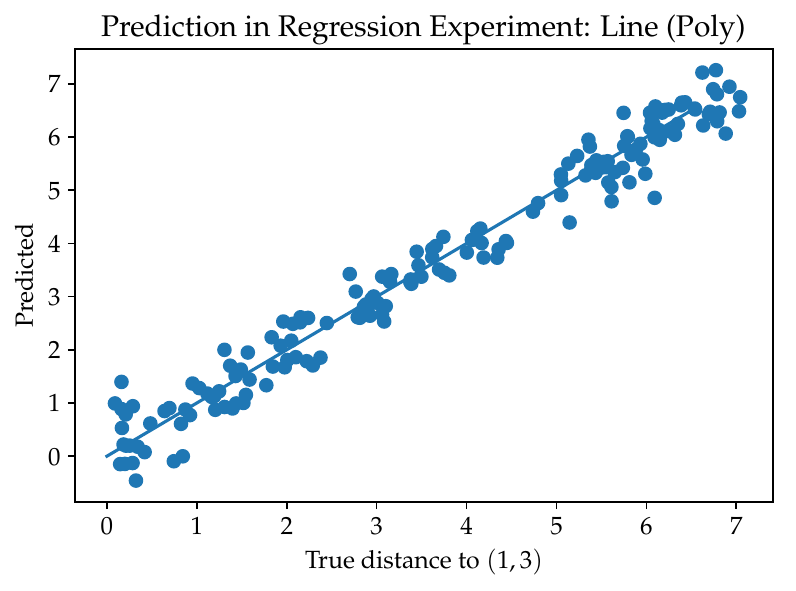}
	\includegraphics[width = .39\textwidth]{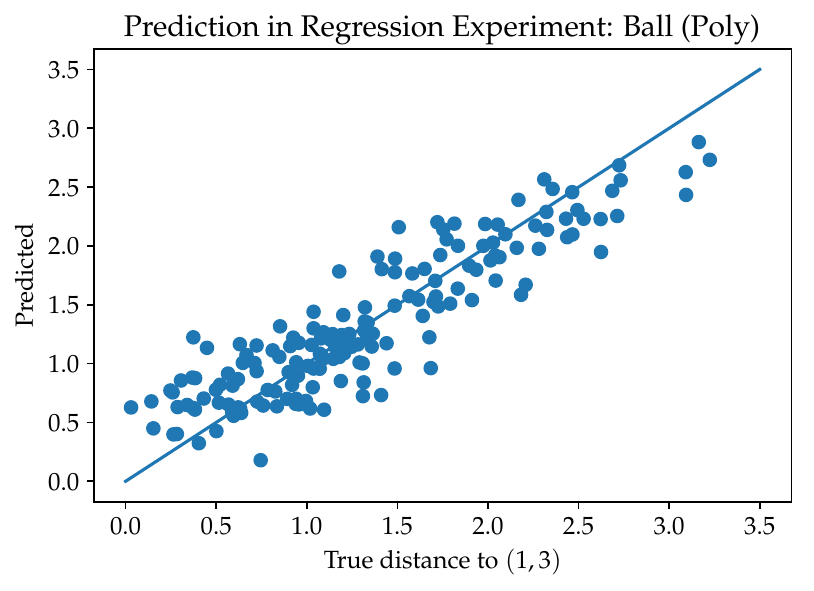}

	\caption{True vs predicted of distance to starting mean for the ball and line experiments described in Sec.~\protect\ref{ssec:NormalDistPoints}.}
	\label{fig:RegressionExperiment}
\end{figure}

\begin{figure}[tb]
	\centering
	\includegraphics[width = .39\textwidth]{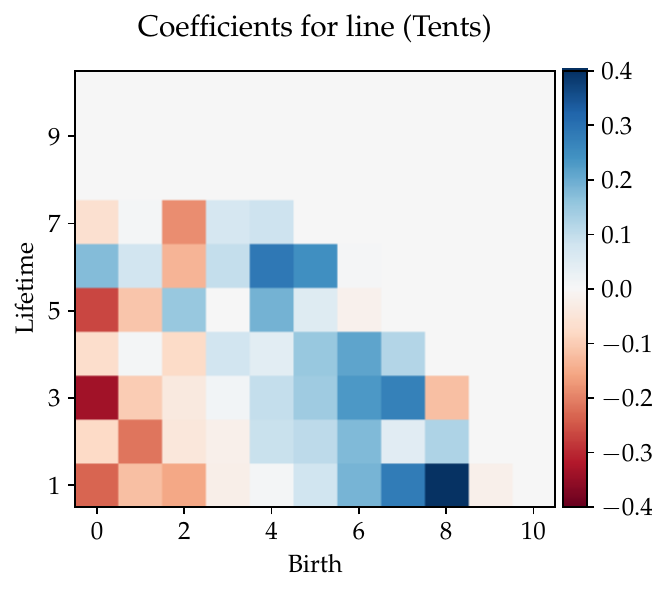}
	\includegraphics[width = .39\textwidth]{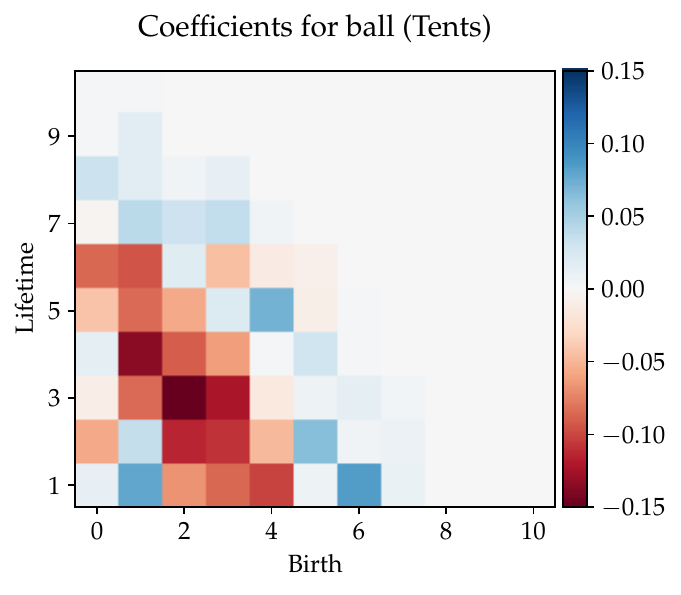}

	\includegraphics[width = .39\textwidth]{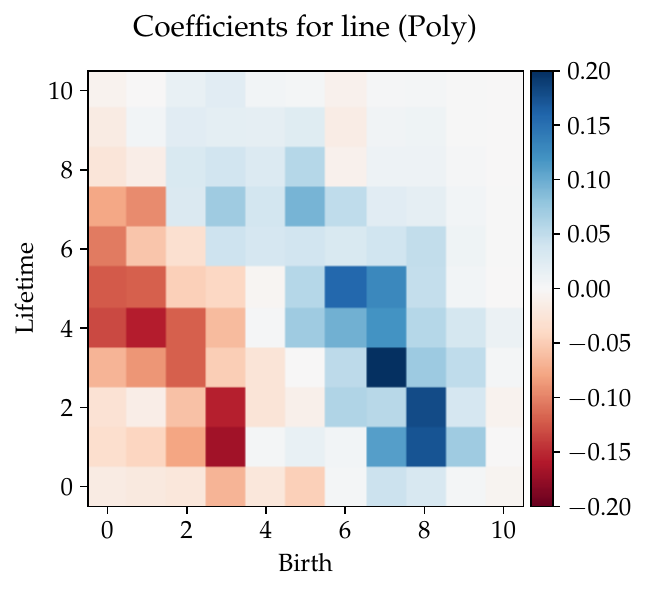}
	\includegraphics[width = .39\textwidth]{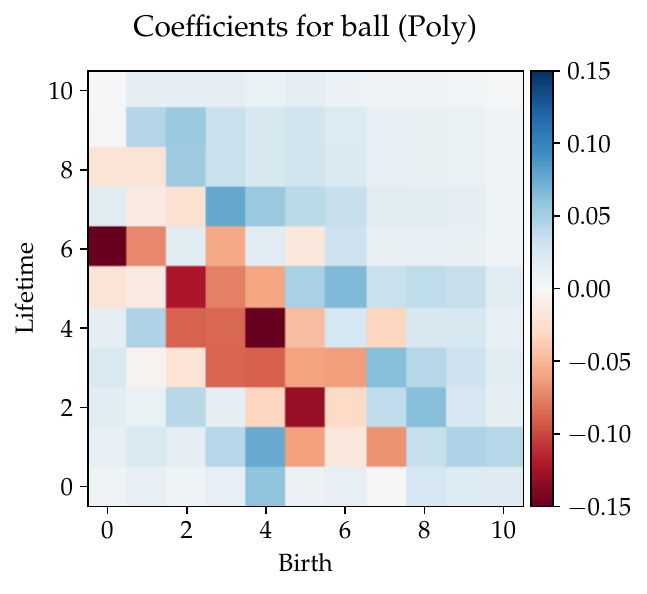}

	\caption{Coefficients for regression experiment with mean drawn from a line (left) and from a ball (right).  The top row uses tent functions, the bottom uses interpolating polynomials. The $x$ and $y$ coordinates correspond to the index of the test function used.  }
	\label{fig:RegressionCoefficients}
\end{figure}

%------------------------

\subsection{Manifold experiment}
\label{ssec:ManifoldExperiment}

Following an experiment run in \cite{Adams2017}, we generated collections of point clouds drawn from different manifolds embedded in $\mathbb{R}^2$ or $\mathbb{R}^3$.
Each point cloud has $N=200$ points.
The categories are as follows:
\begin{description}
	\item [Annulus.] Points drawn uniformly from an annulus with inner radius 1 and outer radius 2.
	\item [{3 clusters.}] The 200 points are drawn from one of three different normal distributions, with means $(0,0)$, $(0,2)$ and $(2,0)$ respectively, and all with standard deviation 0.05.
	\item [{3 clusters of 3 clusters.}]
	The points are drawn from normal distributions with standard deviation 0.05 centered at the points $(0,0)$, $(0,1.5)$, $(1.5,0)$, $(0,4)$, $(1,3)$, $(1,5)$, $(3,4)$, $(3,5.5)$, $(4.5,4)$.
	\item [{Cube.}] Points drawn uniformly from $[0,1]^2 \subseteq \mathbb{R}^2$.
	\item [{Torus.}] Points drawn uniformly from a torus thought of as rotating a circle of radius 1 in the $xz$-plane centered at (2,0) around the $z$-axis.  The generation of the points is done using the method from \cite{Diaconis2013}.
	\item [{Sphere.}] Points drawn from a sphere of radius 1 in $\mathbb{R}^3$ with uniform noise in $[-0.05,0.05]$ added to the radius.
\end{description}
Examples of each of these can be seen in Fig.~\ref{fig:Manifolds}.
Code for generation of these point clouds as well as the full dataset can be found in the \texttt{teaspoon} package at \texttt{teaspoon.MakeData.} \texttt{PointCloud.testSetManifolds}.

The choice of tent function parameters was done as follows.
We determined the bounding box necessary to enclose the training set diagrams in the (birth, lifetime) plane and added padding of $0.05$.
We fixed $d=10$; $\epsilon$ was chosen to be half the minimum lifetime over all training set diagrams; then $\delta$ was chosen to ensure coverage of the bounding box.

We reserved 33\% of the data for testing and trained a regression model on the remaining data.
% The experiment was run for four different numbers of diagrams, each choice of which was run once using only degree one monomials, and once using up to degree two.
The results of this experiment averaged over 10 runs can be seen in Table \ref{tab:ManifoldResults}.
In this experiment, particularly when we have 50 or more diagrams per class, we see excellent ($\geq99\%$) classification.

\begin{figure}[tb]
	\centering
	\includegraphics[width = \textwidth]{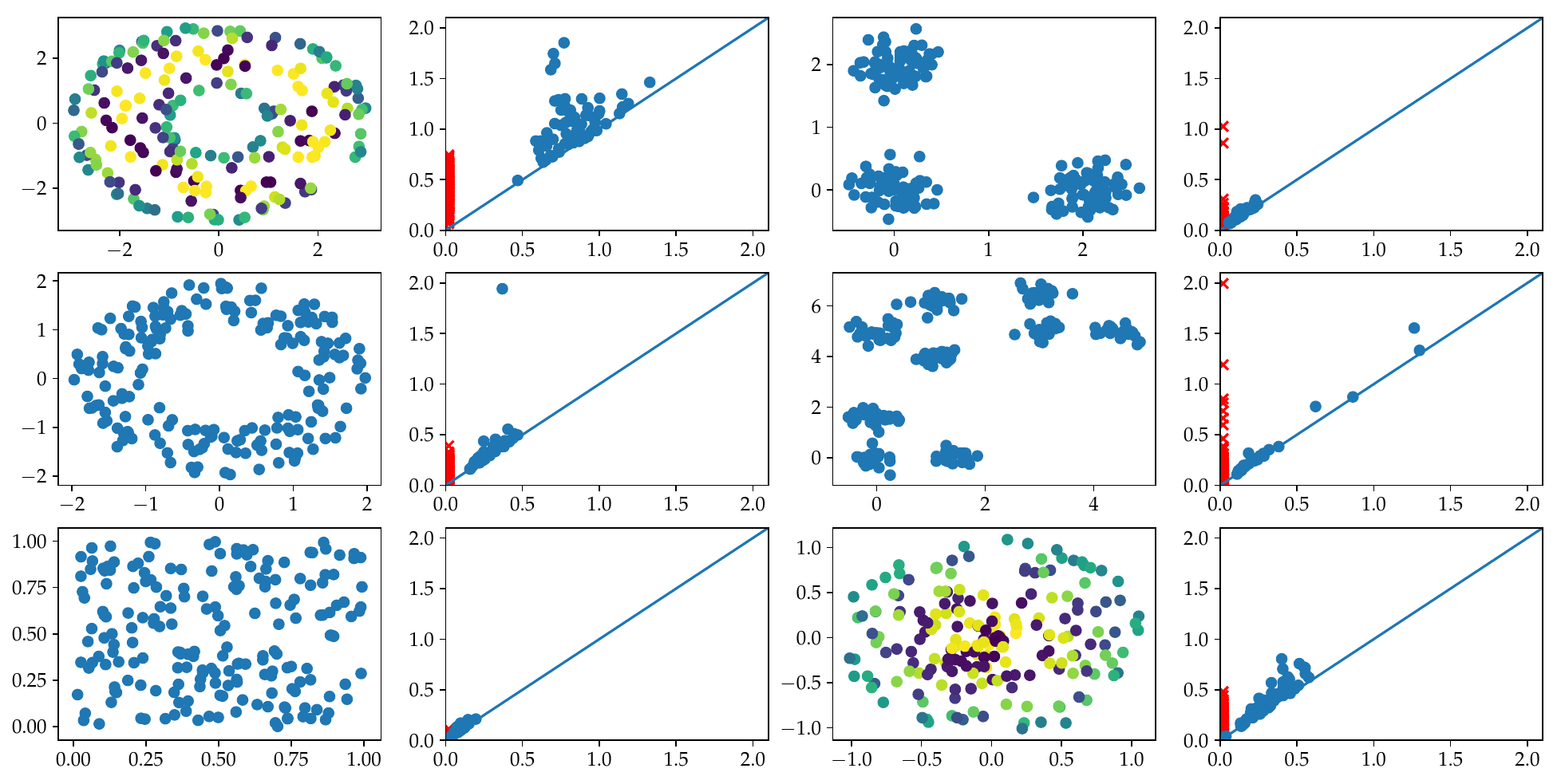}
	\caption{
	Example point clouds from the experiment described in Sec.~\protect\ref{ssec:ManifoldExperiment}.
	From top left reading across rows, the point clouds are a torus (in $\mathbb{R}^3$ but drawn with color as the third coordinate), 3 clusters, annulus, three clusters of three clusters, uniform box, and a sphere (again drawn with third coordinate as color).
	The associated persistence diagrams are shown to the right of the point clouds; the 1-dimensional diagram is given by blue dots and the 0-dimensional diagram is shown by red $x$'s.
	Diagrams are drawn with the same axis values.
	}
	\label{fig:Manifolds}
\end{figure}

\begin{table}[tb]
	\centering
	\begin{tabular}{|r|cccc|}
	\toprule
	  & \multicolumn{2}{c}{Tents} & \multicolumn{2}{c}{Polynomials} \\
	   No.~Dgms & Train & Test & Train & Test \\
	% numDgms & trainScore - tents & train - tents - sd & testScore - tents & test - tents - sd & trainScore - poly & train - poly - sd & testScore - poly & train - poly - sd \\
	\midrule
10 & $ 99.8 \% \pm 0.9$ & $96.5\% \pm 3.2$ & $99.8 \%\pm 0.9$ & $ 95.0 \%\pm 3.9$ \\
25 & $ 99.9 \% \pm  0.3$ & $ 99.0\% \pm 1.0$ & $99.7 \%\pm 0.5$ & $ 97.6 \%\pm 1.5$ \\
50 & $ 99.9 \% \pm  0.2$ & $ 99.9\% \pm 0.3$ & $100 \%\pm 0$ & $ 99.2 \%\pm 0.9$ \\
100 & $99.8 \% \pm 0.1$ & $99.7\% \pm 0.4$ & $99.6 \%\pm 0.2$ & $ 99.3 \%\pm 0.5$ \\
200 & $99.5 \% \pm 0.1$ & $99.5\% \pm 0.3$ & $99.2 \%\pm 0.2$ & $98.9 \%\pm  0.5$ \\
%  25 &            $   99.9\% \pm 0.3 $ &   $99\%     \pm 1    $   & $100\%    \pm 0   $ & $99\%     \pm 1 $   \\
%  10 &            $  99.75\%  \pm 0.75$  &  $96.5\%   \pm 3.2  $   & $100\%    \pm 0   $ & $94\%     \pm 5.39$ \\
%  50 &            $   99.9\%  \pm 0.2 $ &   $99.9\%   \pm 0.3  $   & $99.95\%  \pm 0.15$ & $99.6\%   \pm 0.49$ \\
% 100 &            $  99.83\%  \pm 0.16$  &  $99.7\%   \pm 0.25 $   & $99.65\%  \pm 0.2 $ & $99.29\%  \pm 0.52$ \\
% 200 &            $  99.53\%  \pm 0.17$  &  $99.47\%  \pm 0.29 $   & $99.29\%  \pm 0.17$ & $98.96\%  \pm 0.67$ \\
	\bottomrule
	\end{tabular}

\caption{Results from the manifold test described in Sec.~\protect\ref{ssec:ManifoldExperiment} for different numbers of examples drawn for each type of manifold. The reported results are averaged over 10 experiments each.}
	\label{tab:ManifoldResults}
\end{table}

\begin{figure}[tb]
	\centering
	\includegraphics[width = \textwidth]{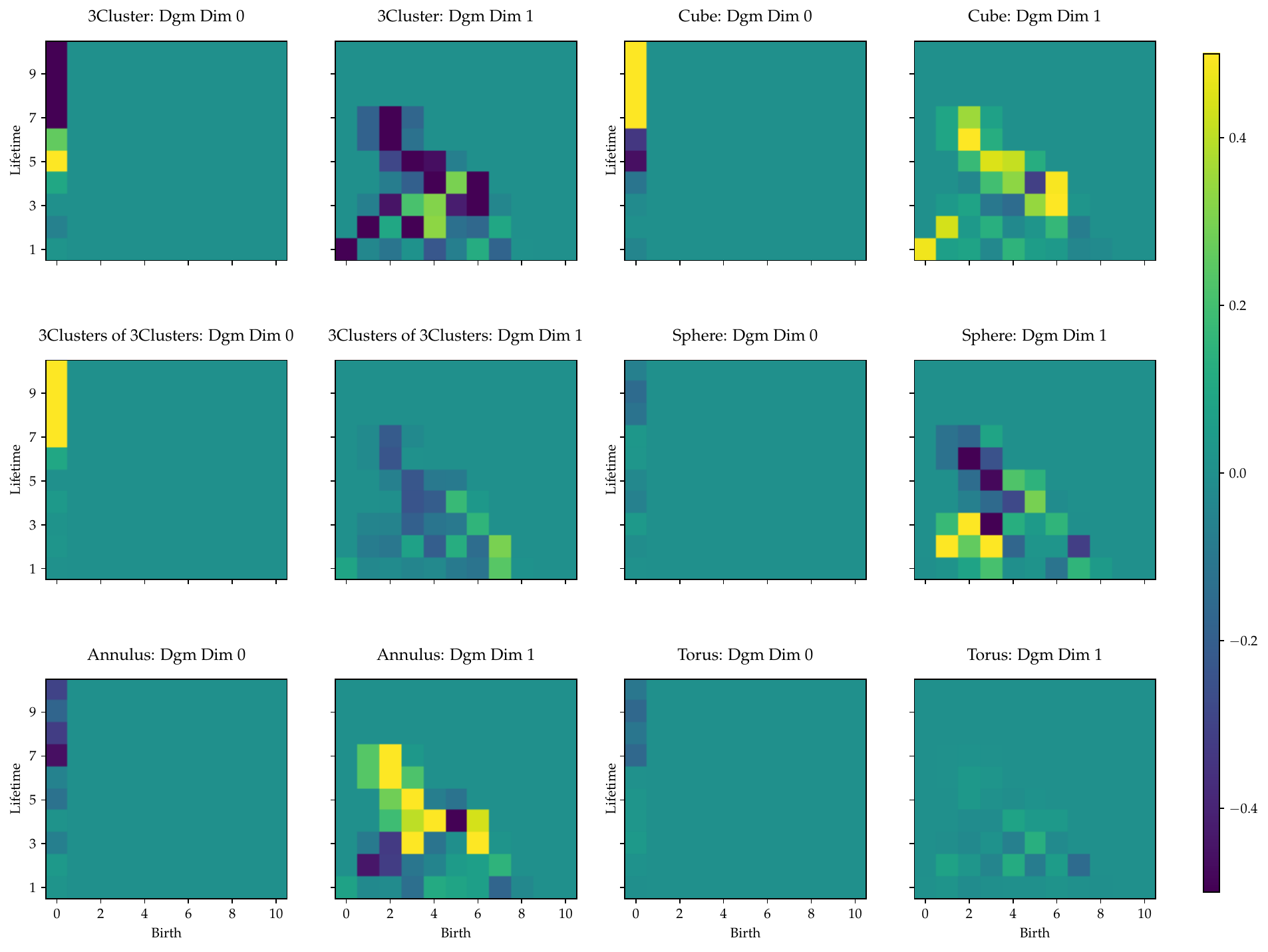}
	\caption{Coefficients for the manifold experiment using tent functions run with 100 diagrams each.}
	\label{fig:ManifoldCoefficients}
\end{figure}

\subsection{Shape Data}
\label{sec:UliShapeData}

We compared our results to the kernel method results reported in \cite{Reininghaus2015} by applying feature functions to the same data set from that paper.
In particular, the synthetic SHREC 2014 data set (\cite{Pickup2014}) consists of 3D meshes of humans in different poses.
The people are labeled as male, female, and child  (five each); and each person assumes one of 20 poses.
Reininghaus et al.~defined a function on each mesh using the heat kernel signature (\cite{Sun2009}) for 10 parameters and computed the 0- and 1-dimensional diagrams of each.

We start with this data set of 300 pairs of persistence diagrams (0- and 1-dimensional) for each of the 10 parameter values, and predicted the human model; i.e.~which of the 15 people were represented by each mesh.
A comparison of the results reported in \cite{Reininghaus2015} with our method using polynomial functions is shown in Table \ref{tab:UliShapeResults}.
Additional results using tent functions are provided in the appendix (Table \ref{tab:UliShapeResultsTents}).

For this experiment, polynomial features (Table \ref{tab:UliShapeResults}) were considerably more successful than the tent functions.
Further, using the 0- and 1-dimensional persistence diagrams together was largely better than the 1-dimensional diagram alone, and considerably better than the 0-dimensional diagram alone.
The average classification rates were improved in four out of the ten parameter choices; results with intersecting confidence intervals occurred in an additional four out of ten parameter choices.

\begin{table}
\centering

\scalebox{.8}{
\begin{tabular}{|l|c|cc|cc|cc|}
\toprule
&& 	  \multicolumn{2}{c}{Dim 0} & \multicolumn{2}{c}{Dim 1} & \multicolumn{2}{c}{Dim 0 \& Dim 1} \\
freq &  MSK      & Train & Test &   Train & Test     &   Train & Test    \\
1 & \cellcolor{blue!25}$94.7\% \pm 5.1$ &
    $94.3\% \pm 0.5$ & $67.1\% \pm 4.7$ & % False & False &
    $99.1\% \pm 0.3$ & $85.4\% \pm 3.0$ & % False & False &
    $99.8\% \pm 0.3$ & \cellcolor{orange!25}$90.4\% \pm 5.3$ % & False & True \\
    \\
2 & \cellcolor{blue!25}$99.3\% \pm 0.9$ &
    $92.1\% \pm 1.4$ & $60.8\% \pm 6.3$ & % False & False &
    $99.9\% \pm 0.3$ & $89.9\% \pm 1.5$ & % False & False &
    $100.\% \pm 0.0$ & $95.1\% \pm 2.4$ % & False & False \\
    \\
3 & \cellcolor{blue!25}$96.3\% \pm 2.2$ &
    $83.4\% \pm 2.4$ & $45.1\% \pm 2.9$ & % False & False &
    $99.6\% \pm 0.5$ & $88.9\% \pm 3.0$ & % False & False &
    $99.7\% \pm 0.5$ & $90.0\% \pm 2.0$ % & False & False \\
    \\
4 & \cellcolor{blue!25}$97.3\% \pm 1.9$ &
    $74.7\% \pm 2.0$ & $37.4\% \pm 4.7$ & % False & False &
    $99.1\% \pm 0.7$ & $85.2\% \pm 2.5$ & % False & False &
    $98.6\% \pm 0.9$ & $84.8\% \pm 3.9$ % & False & False \\
    \\
5 & \cellcolor{blue!25}$96.3\% \pm 2.5$ &
    $65.3\% \pm 2.9$ & $27.8\% \pm 5.0$ & % False & False &
    $99.2\% \pm 0.7$ & \cellcolor{orange!25}$93.0\% \pm 2.2$ & % False & True &
    $99.7\% \pm 0.4$ & \cellcolor{orange!25}$93.3\% \pm 2.2$ % & False & True \\
    \\
6 & \cellcolor{blue!25}$93.7\% \pm 3.2$ &
    $67.2\% \pm 2.5$ & $36.5\% \pm 3.6$ & % False & False &
    $99.2\% \pm 0.5$ & \cellcolor{orange!25}$93.4\% \pm 2.8$ & % False & True &
    $98.8\% \pm 0.5$ & \cellcolor{orange!25}$92.9\% \pm 1.8$ % & False & True \\
    \\
7 & $88.0\% \pm 4.5$ &
    $71.5\% \pm 2.8$ & $40.9\% \pm 4.1$ & % False & False &
    $98.3\% \pm 0.7$ & \cellcolor{blue!25}$96.6\% \pm 0.7$ & % True &  True &
    $99.0\% \pm 0.4$ & \cellcolor{orange!25}$95.6\% \pm 1.4$ % & True &  True \\
    \\
8 & \cellcolor{orange!25}$88.3\% \pm 6.0$ &
    $84.2\% \pm 3.3$ & $63.0\% \pm 4.5$ & % False & False &
    $99.0\% \pm 0.5$ & \cellcolor{orange!25}$93.0\% \pm 1.8$ & % True &  True &
    $99.6\% \pm 0.4$ & \cellcolor{blue!25}$94.0\% \pm 2.2$ % & True &  True \\
    \\
9 & \cellcolor{orange!25}$88.0\% \pm 5.8$ &
    $83.5\% \pm 2.7$ & $62.4\% \pm 5.0$ & % False & False &
    $98.4\% \pm 1.2$ & \cellcolor{blue!25}$92.9\% \pm 1.5$ & % True &  True &
    $98.5\% \pm 1.3$ & \cellcolor{orange!25}$92.6\% \pm 2.1$ % & True &  True \\
    \\
10 & \cellcolor{orange!25}$91.0\% \pm 4.0$ &
    $79.8\% \pm 2.7$ & $59.0\% \pm 4.6$ & % False & False &
    $96.9\% \pm 0.6$ & \cellcolor{blue!25}$92.1\% \pm 1.7$ & % True &  True &
    $97.7\% \pm 1.1$ & \cellcolor{orange!25}$89.5\% \pm 4.6$ % & False & True \\
    \\
\midrule
\bottomrule
\end{tabular}
}
\caption{
Results of classification of shape data discussed in Sec.~\protect\ref{sec:UliShapeData}. The function used was the Chebyshev polynomial of the second kind.  The MSK column gives the original results from \protect\cite{Reininghaus2015}; the subsequent columns use the 0-dimensional diagrams only, the 1-dimensional diagrams only, and both, respectively.
Scores highlighted in blue give best average score MSK vs.~template functions; scores highlighted in orange have overlapping intervals with the best score. Compare this to the results with tent functions, Table \protect\ref{tab:UliShapeResultsTents}.}
\label{tab:UliShapeResults}
\end{table}

% \subsection{Digits a la Adcock}
%**********************************
\subsection{Rossler Periodicity}
%**********************************
\begin{figure}[!htb]
	\centering
	\includegraphics[width=0.7\textwidth]{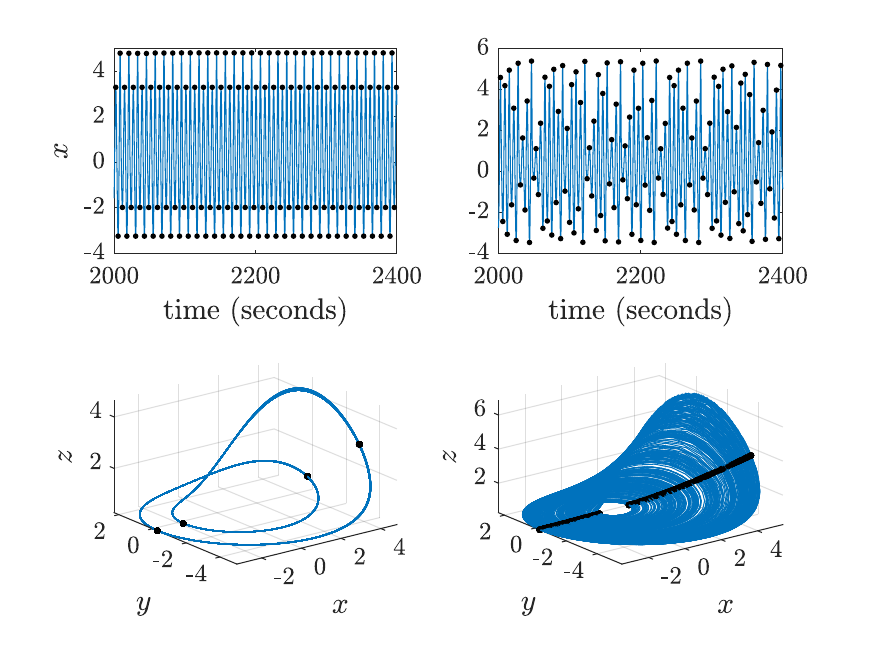}
	\caption{Time series for the $x$ value in a Rossler system (top row) and the corresponding phase portrait in the $(x, y, z)$ space (bottom row) for a periodic case with $\alpha=0.37$ (left column), and a chaotic case with $\alpha=0.42$ (right column).
	 The superimposed black dots correspond to the extrema of $x$.}
	\label{fig:sampleTimeSeries}
\end{figure}
We tested our machine learning approach on time series simulated from the Rossler system (\cite{McCullough2015})
%---------------------------
\begin{align}
\begin{split}
  \dot{x} &= -y - z, \\
  \dot{y} &= x + \alpha y, \\
  \dot{z} &= \beta + z, (x - \gamma),
\end{split}
\end{align}
%---------------------------
where the overdot denotes a derivative with respect to time.
We used an explicit Runge-Kutta (4,5) formula to solve the Rossler system  for $\beta=2$, $\gamma=4$, and $1201$ evenly spaced values of the bifurcation parameter $\alpha$ where $0.37 \leq \alpha \leq 0.43$.
For each value of $\alpha$ a set of initial conditions was sampled from uniformly distributed values in $[0, 1]$.
We simulated $2 \times 10^4$ points using a time step of $0.2$ seconds.
Half of the simulated points were discarded, and only the second half of the $x$ variable data was used in the current analysis.
The left and right columns of Fig.~\ref{fig:sampleTimeSeries} show two examples of the resulting time series: one periodic with $\alpha=0.37$, and one chaotic with $\alpha=0.42$, respectively.
The first row of the figure shows the time series after dropping the first half of the simulated data, and the bottom row shows the corresponding phase space.
The black dots in Fig.~\ref{fig:sampleTimeSeries} represent the extrema of $x$ which were accurately computed using a modified version of Henon's algorithm (\cite{Henon1982,Palaniyandi2009}).
These dots are basically the Poincar\'e points obtained by finding all the intersections of the $x$ trajectory with the surface $\dot{x}=0$.

The two examples in Fig.~\ref{fig:sampleTimeSeries} show how the bifurcation parameter $\alpha$ can influence the system behavior.
This dependence on $\alpha$ is further illustrated in the top graph of Fig.~\ref{fig:lyapExponent} which depicts the maximum Lyapunov exponent computed using the algorithm described by \cite{Benettin1980,Eckmann1985}; and \cite{Sandri1986}.
The bottom graph of Fig.~\ref{fig:lyapExponent} shows the score of the zero-one test for chaos (\cite{Gottwald2004,Gottwald2009,Gottwald2016}): a binary test that yields a score of $0$ for regular dynamics, and $1$ for chaotic dynamics.
In order to avoid the failure of the $0$-$1$ test due to oversampling, the test was applied to the subsampled data which was obtained by retaining every sixth point from the original signal.
The periodic windows shown in Figs.~\ref{fig:bifDiagram} and \ref{fig:lyapExponent} were identified by examining the plot of the bifurcation diagram in Fig.~\ref{fig:bifDiagram} as well as plots of the maximum Lyapunov exponent and the $0$-$1$ test scores in Fig.~\ref{fig:lyapExponent}.

\begin{figure}[!htb]
	\centering
	\includegraphics[width=0.6\textwidth]{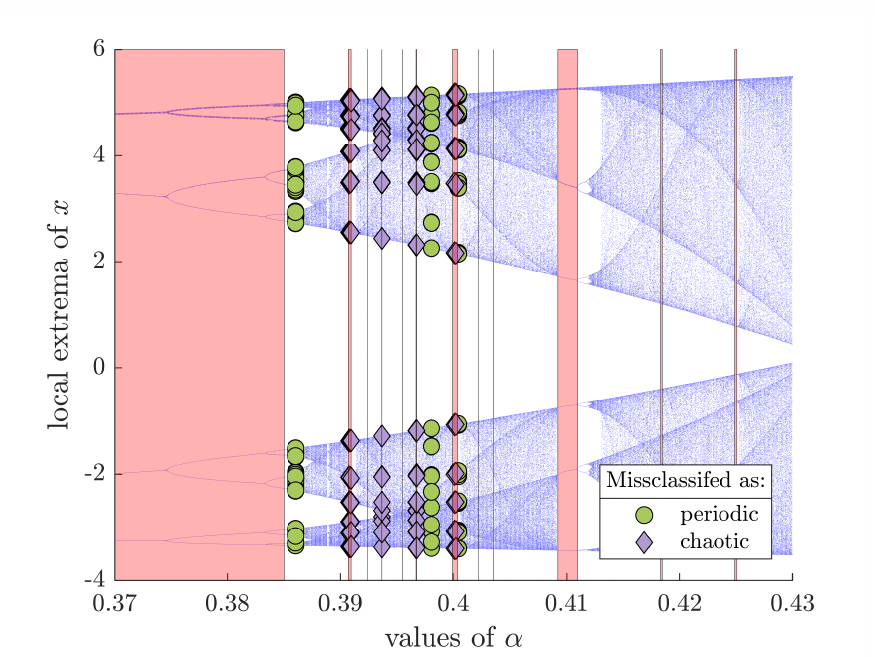}
	\caption{The bifurcation diagram for the Rossler system with $\alpha$ as the bifurcation parameter and the extrema of $x$ as the response parameter.
	The shaded windows indicate the regions that were tagged as periodic.
	The misclassified points are superimposed with diamonds indicating the points that were incorrectly identified as chaotic, while dots indicate that the algorithm incorrectly identified chaotic points as periodic.}
	\label{fig:bifDiagram}
\end{figure}

\begin{figure}[!htb]
	\centering
	\includegraphics[width=0.6\textwidth, trim={0 0.35in 0 0}, clip]{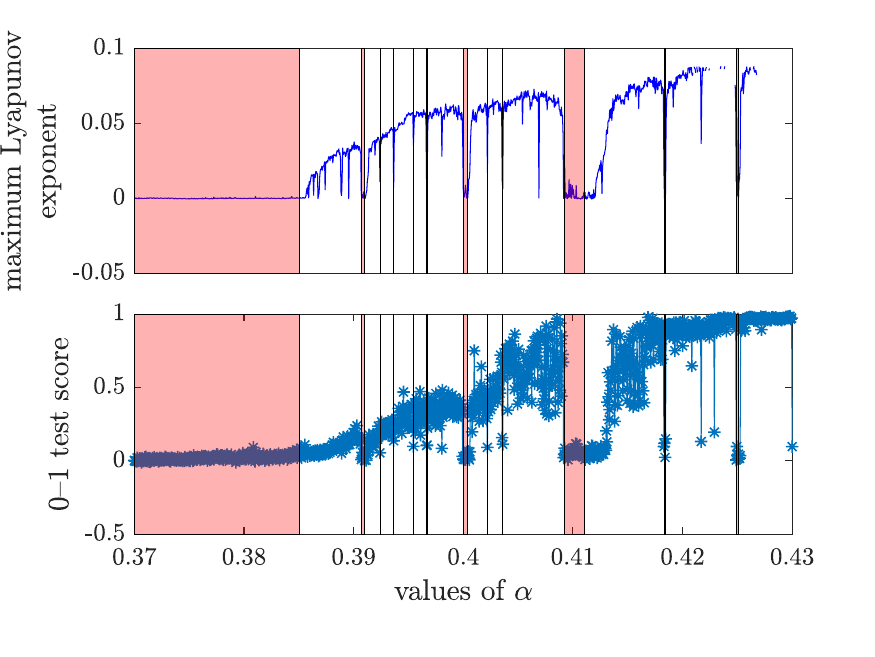}
	\caption{\textbf{Top:} The maximum Lyapunov exponent for the Rossler system as a function of the bifurcation parameter $\alpha$.
	\textbf{Bottom:}~The score of the $0$-$1$ test for chaos where $0$ indicates periodicity, while $1$ indicates chaos.
	The shaded windows denote the regions that were tagged as periodic.}
	\label{fig:lyapExponent}
\end{figure}

We applied the feature function method on the resulting data set using the tent functions.
In this experiment, we set $d = 10$, $\delta = 0.4$, and $\varepsilon $ to be machine precision.
% Table \ref{tab:RosslerParameters} shows the parameters used for applying the described approach to the Rossler experiment.
For this test, we got a score of $98.9\%$ on the training set, and $97.2\%$ on the test set.
The misclassified time series are superimposed on the bifurcation diagram in Fig.~\ref{fig:bifDiagram}.
In this figure, green circles show the values of $\alpha$ that were tagged chaotic but that the algorithm identified as periodic.
Similarly, purple diamonds indicate the $\alpha$ values that were identified as chaotic even though they were tagged as periodic.
It can be seen, unsurprisingly, that misclassification occurs near the transitions of the system behavior from chaotic to periodic or vice versa.

We note that tagging of the data used for both training and testing was performed by inspecting the bifurcation diagram, the maximum Lyapunov exponent plot, and the $0$-$1$ test.
Therefore, for the very few misclassified $\alpha$ values we would actually conjecture that our approach can provide a check for the correctness of tagging in the testing and training sets especially for the boundary cases.

\section{Discussion}
\label{sec:Discussion}

In this paper, we have provided a new method for the featurization of persistence diagrams through the use of template functions; that is, collections of functions compactly supported on the upper half plane away from the diagonal whose induced functions on diagrams separate points.
To do this, we further gave a complete description of compact sets in persistence diagram space endowed with the bottleneck distance.

This method of featurization allows for a great deal of flexibility for the end user.
In particular, we have provided two options for template functions, tent functions and interpolating polynomials, but surely there are many other collections of functions which could be tested for optimizing classification and regression tasks.
We showed these two functions worked quite well on standard experiments, as well as in comparison to other methods available in the literature.

We find the particular results of the SHREC data set (Sec.~\ref{sec:UliShapeData}) to be quite fascinating due to the vast improvement seen from tent functions to interpolating polynomials.
The usual knee-jerk reaction to setting up these featurization methods for persistence diagrams is that localization is key.
This was the impetus for creation of the tent functions as they have support contained in a small box, so each tent function truly only sees a small window of the diagram.
Meanwhile, the interpolating polynomials are nonzero well away from their chosen ``basepoint'' so the fact that these functions work at all is surprising to say the least.

Template function featurization has also found successful applications in the literature. Specifically, in \cite{Yesilli2019} the authors show a comparison between several TDA-based featurization methods such as persistence landscapes, persistence images, and Carlsson coordinates for classifying chatter in metal cutting from vibration signals.
That paper shows that template functions yield higher accuracy in comparison to some of the existing featurization methods.
Further, template functions (and Carlsson coordinates) were shown to be computationally the fastest methods.
Another application of template functions can be found in \cite{yesilli2019chatter}, where the authors show that template functions can be used to detect chatter from simulated, noisy vibration signals in milling--another  metal cutting process.
In addition to its speed, template functions were shown to yield high classification accuracy and did not require the user to manually identify the classification features, which is the case in traditional chatter classification methods.

Future work is certainly needed to expore the available options that can be utilized within the provided mathematical framework.
Because the template system definition is rather broad, we suspect there are many possible collections which could be utilized to improve the results seen here.
We have begun to see this in follow-up work which takes a data driven approach to adaptively chose the relevant template functions \cite{tymochko2019adaptive,polanco2019adaptive}.
We hope to better understand this behavior in future work.

% Acknowledgements should go at the end, before appendices and references

\acks{
JAP acknowledges the support of the National Science Foundation (NSF) under grants DMS-1622301,  CCF-2006661,
CAREER award  DMS-1943758, and  DARPA under grant HR0011-16-2-003.
EM was supported by the NSF through grants CMMI-1800466, DMS-1800446, CCF-1907591, and CCF-2106578.
FAK was supported by the NSF through grants CMMI-1759823 and DMS-1759824.
}


% \newpage

\appendix
%********************************
\section{Implementation of the interpolating polynomials algorithm}
%\label{sec:bary_interp}
\label{sec:interpAlg}
%********************************
%--------------------------------
In this appendix, we give more details on the implementation of the interpolating polynomials described in Section \ref{ssec:InterpPoly2}.
The barycentric formula for Lagrange interpolation described by \cite{Berrut2004} is given by
\begin{equation}
\label{eq:bary_lagAppdx}
f(x) := \sum\limits_{j=0}^m{\ell^\AA_j(x) c_j} = \frac{\sum\limits_{j=0}^m{\frac{w_j}{x-\tilde{x}_j}}c_j}{\sum\limits_{j=0}^k{\frac{w_j}{x-\tilde{x}_j}}};
\text{ where }
  \quad w_j=\frac{1}{\ell'(a_j)};
\quad \ell'(a_j) = \prod\limits_{i=0, i\neq j}^{m}{(a_j-a_i)},
\end{equation}
while $\AA = \{a_i\}_{i=0}^m \subset \mathbb{R}$ is a finite set of distinct mesh values, and $\{c_i \in \mathbb{R}\}$ is a collection of evaluation values.
The function in Eq.~\eqref{eq:bary_lagAppdx} has the property that $f(a_i) = c_i$ for all $i$, and it also satisfies the partition of unity condition $\sum\limits_{j=0}^{m}{f(x)}= 1, \,\, \forall \, x $.

Barycentric Lagrange interpolation is often used for approximating $\mathbb{R}$-valued functions and there are efficient algorithms for obtaining the weights associated with it.
However, in our formulation we need to an interpolating polynomial over an $\mathbb{R}^2$-valued function.
Therefore, we next describe how to expand the algorithm for interpolating a scalar valued function to interpolating a function on the plane.
Note that the notation used here is self-contained from Section \ref{ssec:InterpPoly2}.

We assume that our planar mesh is the outer product of $m+1$ mesh points along the birth time $x$-axis, and $n+1$ points along the lifetime $y$-axis.
We also assume that the persistence diagram has $N$ pairs of (birth, lifetime) points.
%----------------------------------
\begin{enumerate}%[nosep]
\item Get $\tilde{\gamma}$ and $\phi$ which correspond to the interpolation matrices along the $x$-mesh and the $y$-mesh, respectively.
These are the matrices that describe the linear transformation from the $m+1$ mesh points of birth times ($n+1$ mesh of lifetimes) to the corresponding interpolated values of the $N$ query birth times ($N$ query lifetimes) for a given diagram.
This step is equivalent to separately obtaining the interpolation matrices for the birth times and the lifetimes.
\item Set $\gamma=\tilde{\gamma}^T$.
\item \begin{enumerate}
		\item Replicate each column in $\gamma$ $n+1$ times to obtain $\Gamma$ whose dimensions are $(m+1)\times (N\times(n+1))$.
		\item Unravel $\phi$ row-wise into a row vector, then replicate each row $m+1$ times to obtain $\Phi$ whose dimensions are $(m+1)\times(N\times(n+1))$.
	  \end{enumerate}
\item Use element-wise multiplication to obtain $\tilde{\Psi}=\Gamma \cdot \Phi$, where $\cdot$ means element-wise multiplication, and $\tilde{\Psi}$ has dimension $(m+1)\times(N\times(n+1))$.
\item \begin{enumerate}
	  \item Split $\tilde{\Psi}$ into $N$ chunks of $(m+1)\times (n+1)$ matrices along the columns axis.
	  \item Concatenate the split pieces row-wise to obtain an $(N\times(m+1))\times(n+1)$ matrix $\Psi$.
	  \end{enumerate}
\item Reshape $\Psi$ by concatenating each $(m+1)\times (n+1)$ piece row-wise to obtain an $N \times ((m+1)\times(n+1))$ matrix $\Xi$.
\item Let the 2D base mesh be given as
\begin{equation*}
\begin{bmatrix}
	f_{00} & f_{01} & \ldots & f_{0n} \\
	f_{10} & f_{11} & \ldots & f_{1n} \\
	\vdots &        &        & \vdots \\
	f_{m0} & f_{m1} & \ldots & f_{mn}
\end{bmatrix},
\end{equation*}
where $f_{ij} = f(x_i, y_j)$ and $(x_i, y_j)$ is a unique point in the 2D mesh.
Define the vector $[f_{00} \, f_{01} \, \ldots \, f_{mn}]$ which is obtained by unraveling the 2D mesh row-wise.
\item We can interpolate the query points $(x_q, y_q)$ using
\begin{equation*}
p(x_q, y_q) = \begin{bmatrix}
			\ell_0(x_0) \ell_0(y_0) & \ldots & \ell_m(x_0) \ell_n(y_0) \\
			\ell_0(x_1) \ell_0(y_1) & \ldots & \ell_m(x_1) \ell_n(y_1) \\
			\vdots                  &        & \vdots \\
			\ell_0(x_{N-1}) \ell_0(y_{N-1}) & \ldots & \ell_m(x_{N-1}) \ell_n(y_{N-1})
			\end{bmatrix}
			\begin{bmatrix}
			f_{00} \\
			f_{01} \\
			\vdots \\
			f_{mn}
			\end{bmatrix}.
\end{equation*}
\end{enumerate}
%----------------------------------
Here is a sketch of the resulting matrices:
\begin{align*}
\tilde{\gamma} &=
	\begin{bmatrix}
		\ell_0(x_0) & \ell_1(x_0) & \ldots & \ell_m(x_0) \\
		\vdots      &             &        & \vdots \\
		\ell_0(x_{N-1}) & \ell_1(x_{N-1}) & \ldots & \ell_m(x_{N-1})
    \end{bmatrix}_{N\times(m+1)},\\
\phi &=
	\begin{bmatrix}
		\ell_0(y_0) & \ell_1(y_0) & \ldots & \ell_n(y_0) \\
		\vdots      &             &        & \vdots \\
		\ell_0(y_{N-1}) & \ell_1(y_{N-1}) & \ldots & \ell_n(y_{N-1})
    \end{bmatrix}_{N\times(n+1)},\\
\gamma = \tilde{\gamma}^T &=
	\begin{bmatrix}
		\ell_0(x_0) & \ell_0(x_1) & \ldots & \ell_0(x_{N-1}) \\
		\ell_1(x_0) & \ell_1(x_1) & \ldots & \ell_1(x_{N-1})\\
		\vdots      &             &        & \vdots \\
		\ell_m(x_0) & \ell_m(x_1) & \ldots & \ell_m(x_{N-1})
    \end{bmatrix}_{(m+1)\times N},
\end{align*}
\begin{equation*}
\Gamma =
	\begin{bmatrix}
		\ell_0(x_0) & \ell_0(x_0) & \ldots & \ell_0(x_0) & \ldots & \ell_0(x_{N-1}) & \ell_0(x_{N-1}) & \ldots & \ell_0(x_{N-1})\\
		\ell_1(x_0) & \ell_1(x_0) & \ldots & \ell_1(x_0) & \ldots & \ell_1(x_{N-1}) & \ell_1(x_{N-1}) & \ldots & \ell_1(x_{N-1})\\
		\vdots &  & \vdots &  & \vdots &  & \vdots &  & \vdots\\
		\ell_m(x_0) & \ell_m(x_0) & \ldots & \ell_m(x_0) & \ldots & \ell_m(x_{N-1}) & \ell_m(x_{N-1}) & \ldots & \ell_m(x_{N-1})
	\end{bmatrix}
\end{equation*}
where $\Gamma$ has dimension $(m+1)\times (N\times(n+1))$.

\begin{equation*}
\Phi =
	\begin{bmatrix}
		\ell_0(y_0) & \ell_1(y_0) & \ldots & \ell_n(y_0) & \ldots & \ell_0(y_{N-1}) & \ell_1(y_{N-1}) & \ldots & \ell_n(y_{N-1}) \\
		\ell_0(y_0) & \ell_1(y_0) & \ldots & \ell_n(y_0) & \ldots & \ell_0(y_{N-1}) & \ell_1(y_{N-1}) & \ldots & \ell_n(y_{N-1}) \\
		\vdots &  & \vdots &  & \vdots &  & \vdots &  & \vdots\\
		\ell_0(y_0) & \ell_1(y_0) & \ldots & \ell_n(y_0) & \ldots & \ell_0(y_{N-1}) & \ell_1(y_{N-1}) & \ldots & \ell_n(y_{N-1})
	\end{bmatrix}
\end{equation*}
where $\Phi$ has dimension $(m+1)\times(N\times(n+1))$.

We can now compute the elementwise product $\Psi = \Gamma \cdot \Phi$, which has the dimension $(m+1)\times(N\times(n+1))$.

We then need to apply the following operations: (i) reshaping $\Psi$ to obtain $\hat{\Psi}_1$ given by
\begin{equation*}
\hat{\Psi}_1 =
	\begin{bmatrix}
		\ell_0(x_0)\ell_0(y_0) & \ell_0(x_0)\ell_1(y_0) & \ldots & \ell_0(x_0)\ell_n(y_0)  \\
		\ell_1(x_0)\ell_0(y_0) & \ell_1(x_0)\ell_1(y_0) & \ldots & \ell_1(x_0)\ell_n(y_0)  \\
		\vdots &  & \vdots &  \\
		\ell_m(x_0)\ell_0(y_0) & \ell_m(x_0)\ell_1(y_0) & \ldots & \ell_m(x_0)\ell_n(y_0)  \\
		\vdots &  & \vdots &  \\
		\ell_0(x_{N-1})\ell_0(y_{N-1}) & \ell_0(x_{N-1})\ell_1(y_{N-1}) & \ldots & \ell_0(x_{N-1})\ell_n(y_{N-1}) \\
		\ell_1(x_{N-1})\ell_0(y_{N-1}) & \ell_1(x_{N-1})\ell_1(y_{N-1}) & \ldots & \ell_1(x_{N-1})\ell_n(y_{N-1}) \\
		\vdots &  & \vdots &  \\
		\ell_m(x_{N-1})\ell_0(y_{N-1}) & \ell_m(x_{N-1})\ell_1(y_{N-1}) & \ldots & \ell_m(x_{N-1})\ell_n(y_{N-1})
	\end{bmatrix}.
\end{equation*}

(ii) unraveling $\hat{\Psi}_1$ into an $N\times((m+1)\times(n+1))$ matrix $\hat{\Psi}_2$ given by
\begin{equation}
\hat{\Psi}_2 =
	\begin{bmatrix}
		\ell_0(x_0) \ell_0(y_0) & \ldots & \ell_0(x_0) \ell_n(y_0) & \ldots & \ell_m(x_0) \ell_n(y_0) \\
		\vdots                  &                                  & \vdots &                         \\
		\ell_0(x_k) \ell_0(y_k) & \ldots & \ell_0(x_k) \ell_n(y_k) & \ldots & \ell_m(x_k) \ell_n(y_k) \\
		\vdots                  &                                  & \vdots &                         \\
		\ell_0(x_{N-1}) \ell_0(y_{N-1}) & \ldots & \ell_0(x_{N-1}) \ell_n(y_{N-1}) & \ldots & \ell_m(x_{N-1}) \ell_n(y_{N-1})
	\end{bmatrix}.
\end{equation}

The collection of all the scores constitutes the feature vector corresponding to the chosen base mesh point and to the query points where the latter are the persistence diagram points.
In this study we summed the rows of $\hat{\Psi}_2$ after taking the absolute value of each entry.
The resulting number represents the score at each base mesh point.
If the persistence diagram contains the mesh points and we want to find the interpolated values at query points $p_{\rm interp}$, then we would compute $p_{\rm interp.}=\hat{\Psi}_2\, f$.

The implementation of this algorithm can be found in the teaspoon package at\\  \texttt{teaspoon.ML.feature\_functions.interp\_polynomial}.

\section{Additional shape data results}

This appendix gives additional results for the SHREC data set described in Section \ref{sec:UliShapeData} using tent functions instead of interpolating polynomials.
Table \ref{tab:UliShapeResultsTents} should be compared to the results of Table \ref{tab:UliShapeResults}.

\begin{table}[!htb]
\centering
\scalebox{.8}{
\begin{tabular}{|l|c|cc|cc|cc|}
\toprule
     &           & 	  \multicolumn{2}{c}{Dim 0} & \multicolumn{2}{c}{Dim 1} & \multicolumn{2}{c}{Dim 0 \& Dim 1} \\
freq &  MSK      & Train & Test                 &   Train & Test            &   Train & Test    \\
\midrule
1    & \cellcolor{blue!25} $94.7\% \pm  5.1$ &  $8.3\%  \pm  0.5$ &  $3.4\%  \pm 1.1$ &  $8.1\%  \pm  0.2$ &  $3.7\%  \pm 0.5$  &  $8.2 \% \pm  0.3$ &  $ 3.5\% \pm 0.5$\\
2    & \cellcolor{blue!25} $99.3\% \pm  0.9$ &  $8.3\%  \pm  0.3$ &  $3.4\%  \pm 0.7$ &  $8.2\%  \pm  0.5$ &  $3.5\%  \pm 1.1$  &  $8.56 \% \pm  0.4$ &  $ 3.0\% \pm 1.0$\\
3    & \cellcolor{blue!25} $96.3\% \pm  2.2$ &  $66.5\% \pm  2.7$ &  $31.8\% \pm 4.8$ & $50.6\% \pm  2.1$ &  $31.1\% \pm 4.0$ &  $80.5\% \pm  1.3$ &  $44.4\% \pm 4.3$\\
4    & \cellcolor{blue!25} $97.3\% \pm  1.9$ &  $46.2\% \pm  2.5$ &  $27.0\% \pm 3.8$ & $83.1\% \pm  1.6$ &  $63.5\% \pm 4.6$  &  $89.1\% \pm  1.5$ &  $69.0\% \pm 4.9$\\
5    & \cellcolor{blue!25} $96.3\% \pm  2.5$ &  $28.5\% \pm  1.4$ &  $18.9\% \pm 4.0$ & $75.2\% \pm  2.6$ &  $58.3\% \pm 4.6$  &  $76.8\% \pm  2.7$ &  $58.4\% \pm 7.9$\\
6    & \cellcolor{blue!25} $93.7\% \pm  3.2$ &  $25.4\% \pm  1.8$ &  $19.0\% \pm 2.4$ & $96.5\% \pm  1.1$ & \cellcolor{orange!25} $88.7\% \pm 2.4$  &  $96.8\% \pm  0.67$ &  \cellcolor{orange!25} $89.9\% \pm 1.7$ \\
7    &  $88.0\% \pm  4.5$
     &  $19.4\% \pm  2.6$ &  $10.0\% \pm   3.4$
     &  $98.2\% \pm  0.5$ & \cellcolor{orange!25} $93.6\% \pm 1.9$
     &  $98.3\% \pm  0.6$ & \cellcolor{blue!25} $94.1\% \pm 2.5$\\
8    &  \cellcolor{orange!25} $88.3\% \pm  6.0$
     &  $10.8\% \pm  2.6$ &  $3.6\%  \pm   2.4$
     &  $91.9\% \pm  0.9$ & \cellcolor{orange!25} $88.8\% \pm 2.7$
     &  $91.9\% \pm  1.2$ & \cellcolor{blue!25}$89.7\% \pm 3.3$\\
9    &  $88.0\% \pm  5.8$ \cellcolor{blue!25}&  $10.6\% \pm  2.7$ &  $4.3\% \pm 2.2$ & $63.8\% \pm 2.7$ & $53.3\% \pm 5.9$ & $64.9\% \pm  2.3$ & $53.7\% \pm 3.8$ \\
10   &  $91.0\% \pm  4.0$\cellcolor{blue!25} &  $9.2\%  \pm  2.3$ &  $3.6\% \pm 1.7$ & $27.0\% \pm 3.9$ & $16.2\% \pm 3.2$ & $27.3\% \pm  3.4$ & $18.6\% \pm 5.6$ \\
\bottomrule
\end{tabular}
}
\caption{Results of classification of shape data discussed in Section \protect\ref{sec:UliShapeData}.
The functions used are the tent functions with $d=20$, and a ridge regression classifier. The MSK column gives the original results from \cite{Reininghaus2015}; the subsequent columns use the 0-dimensional diagrams only, the 1-dimensional diagrams only, and both, respectively.
Scores highlighted in blue give best average score MSK vs.~template functions; scores highlighted in orange have overlapping intervals with the best score.
}
\label{tab:UliShapeResultsTents}
\end{table}

\vskip 0.2in
\bibliography{TentFunctions-bio}

\end{document}